\documentclass[aps, pra,reprint, showpacs,nofootinbib,
superscriptaddress,10pt,floatfix,longbibliography]{revtex4-2}
\usepackage{amsmath, amsthm, amssymb}
\usepackage{adjustbox}
\usepackage{graphicx}
\usepackage{ragged2e, physics}

\usepackage{yfonts}
\usepackage{bm}
\usepackage{rotating}
\usepackage{booktabs}
\usepackage{natbib}
\usepackage[normalem]{ulem}
\usepackage{soul}
\usepackage{etoolbox}    
\AtBeginEnvironment{thebibliography}{%
  \sloppy
  \setlength\emergencystretch{1em}%
}
\AtEndEnvironment{thebibliography}{\fussy}

\allowdisplaybreaks
\newcommand{\mycommand}[1]{\csname #1command\endcsname}
\usepackage{appendix}
\usepackage[caption=false]{subfig} 
\usepackage{comment}
\usepackage{enumitem}
\usepackage{color}
\usepackage{algorithm}
\usepackage[utf8]{inputenc}
\usepackage{braket}
\usepackage{qcircuit}
\usepackage[overload]{empheq}
\usepackage{hyperref}
\usepackage{xcolor}
\usepackage[english]{babel}
\usepackage[flushleft]{threeparttable}

\usepackage{algcompatible}
\usepackage[noend]{algpseudocode}
\usepackage{mathtools}

\usepackage{xparse}

\NewDocumentCommand{\INTERVALINNARDS}{ m m }{
    #1 {,} #2
}
\NewDocumentCommand{\interval}{ s m >{\SplitArgument{1}{,}}m m o }{
    \IfBooleanTF{#1}{
        \left#2 \INTERVALINNARDS #3 \right#4
    }{
        \IfValueTF{#5}{
            #5{#2} \INTERVALINNARDS #3 #5{#4}
        }{
            #2 \INTERVALINNARDS #3 #4
        }
    }
}

\usepackage{lipsum}

\usepackage{qcircuit}



\begin{document}
\newtheorem{theorem}{\bf Theorem}[section]
\newtheorem{proposition}[theorem]{\bf Proposition}
\newtheorem{definition}[theorem]{\bf Definition}
\newtheorem{corollary}[theorem]{\bf Corollary}
\newtheorem{example}[theorem]{\bf Example}
\newtheorem{exam}[theorem]{\bf Example}
\newtheorem{remark}[theorem]{\bf Remark}
\newtheorem{lemma}[theorem]{\bf Lemma}
\newtheorem{statement}[theorem]{\bf Statement}
\newcommand{\nrm}[1]{|\!|\!| {#1} |\!|\!|}

\newcommand{\calL}{{\mathcal L}}
\newcommand{\calX}{{\mathcal X}}
\newcommand{\calA}{{\mathcal A}}
\newcommand{\calB}{{\mathcal B}}
\newcommand{\calC}{{\mathcal C}}
\newcommand{\calK}{{\mathcal K}}
\newcommand{\C}{{\mathbb C}}
\newcommand{\R}{{\mathbb R}}
\newcommand{\U}{{\mathrm U}}
\renewcommand{\SS}{{\mathbb S}}
\newcommand{\LL}{{\mathbb L}}
\def\kernel{\mathop{\rm kernel}\nolimits}
\def\sigan{\mathop{\rm span}\nolimits}

\newcommand{\klasse}{{\boldsymbol \Delta}}

\newcommand{\ba}{\begin{array}}
\newcommand{\ea}{\end{array}}
\newcommand{\von}{\vskip 1ex}
\newcommand{\vone}{\vskip 2ex}
\newcommand{\vtwo}{\vskip 4ex}
\newcommand{\dm}[1]{ {\displaystyle{#1} } }

\newcommand{\be}{\begin{equation}}
\newcommand{\ee}{\end{equation}}
\newcommand{\beano}{\begin{eqnarray*}}
\newcommand{\eeano}{\end{eqnarray*}}
\newcommand{\inp}[2]{\langle {#1} ,\,{#2} \rangle}
\def\bmatrix#1{\left[ \begin{matrix} #1 \end{matrix} \right]}
\def \noin{\noindent}
\newcommand{\evenindex}{\Pi_e}

\newcommand{\tb}[1]{\textcolor{blue}{ #1}}
\newcommand{\tm}[1]{\textcolor{magenta}{ #1}}
\newcommand{\tre}[1]{\textcolor{red}{ #1}}
\newcommand{\snote}[1]{\textcolor{blue}{Shantanav: #1}}



\def \K{{\mathbf k}}
\def \N{{\mathbb N}}
\def \R{{\mathbb R}}
\def \F{{\mathbb F}}
\def \C{{\mathbb C}}
\def \Q{{\mathbb Q}}
\def \Z{{\mathbb Z}}
\def \I{{\mathbb I}}
\def \D{{\mathcal D}}
\def \H{{\mathcal H}}
\def \P{{\mathcal P}}
\def \M{{\mathcal M}}
\def \B{{\mathcal B}}
\def \O{{\mathcal O}}
\def \calG{{\mathcal G}}
\def \PO{{\mathcal {PO}}}
\def \X{{\mathcal X}}
\def \Y{{\mathcal Y}}
\def \calW{{\mathcal W}}
\def \pf{{\bf Proof: }}
\def \lam{{\lambda}}
\def\lc{\left\lceil}   
\def\rc{\right\rceil}
\def \N{{\mathbb N}}
\def \Ls{{\Lambda}_{m-1}}
\def \Gb{\mathrm{G}}
\def \Hb{\mathrm{H}}
\def \Delta{\triangle}
\def \Rar{\Rightarrow}
\def \p{{\mathsf{p}(\lam; v)}}

\def \D{{\mathbb D}}

\def \tr{\mathrm{Tr}}
\def \cond{\mathrm{cond}}
\def \lam{\lambda}
\def \sig{\sigma}
\def \sign{\mathrm{sign}}

\def \ep{\epsilon}
\def \diag{\mathrm{diag}}
\def \rev{\mathrm{rev}}
\def \vec{\mathrm{vec}}

\def \ham{\mathsf{Ham}}
\def \herm{\mathsf{Herm}}
\def \sym{\mathsf{sym}}
\def \odd{\mathsf{sym}}
\def \en{\mathrm{even}}
\def \rank{\mathrm{rank}}
\def \pf{{\bf Proof: }}
\def \dist{\mathrm{dist}}
\def \rar{\rightarrow}

\def \rank{\mathrm{rank}}
\def \pf{{\bf Proof: }}
\def \dist{\mathrm{dist}}
\def \Re{\mathsf{Re}}
\def \Im{\mathsf{Im}}
\def \re{\mathsf{re}}
\def \im{\mathsf{im}}

\def \sym{\mathsf{sym}}
\def \sksym{\mathsf{skew\mbox{-}sym}}
\def \odd{\mathrm{odd}}
\def \even{\mathrm{even}}
\def \herm{\mathsf{Herm}}
\def \skherm{\mathsf{skew\mbox{-}Herm}}
\def \str{\mathrm{ Struct}}
\def \eproof{$\blacksquare$}

\def \cnot{\mathrm{CNOT}}
\definecolor{darkolivegreen}{rgb}{0.33, 0.42, 0.18}

\def \bS{{\bf S}}
\def \cA{{\cal A}}
\def \E{{\mathcal E}}
\def \X{{\mathcal X}}
\def \F{{\mathcal F}}
\def \cH{\mathcal{H}}
\def \cJ{\mathcal{J}}
\def \tr{\mathrm{Tr}}
\def \range{\mathrm{Range}}
\def \adj{\star}

\pdfstringdefDisableCommands{%
  \def\\{}%
  \def\texttt#1{<#1>}%
}
\def \adj{\star}

\def \pal{\mathrm{palindromic}}
\def \palpen{\mathrm{palindromic~~ pencil}}
\def \palpoly{\mathrm{palindromic~~ polynomial}}
\def \odd{\mathrm{odd}}
\def \even{\mathrm{even}}

\newcommand{\tg}[1]{\textcolor{green}{ #1}}


\title{Weighted Quantum Signal Processing: Low-Depth Polynomial Approximation with Applications to Kolmogorov–Arnold Networks}

\author{Rohit Sarma Sarkar}
\email{rohit15sarkar@yahoo.com}
\affiliation{Universitat Polit\`ecnica de Catalunya, Barcelona, Spain}

\author{Rupayan Bhattacharjee}
\affiliation{Universitat Polit\`ecnica de Catalunya, Barcelona, Spain}

\author{Elias F. Combarro}
\affiliation{Universidad de Oviedo, Spain}

\author{Michele Grossi}
\affiliation{Quantum Computing Infrastructures and Algorithms Competence Centre, European Organization for Nuclear Research (CERN), Geneva 1211, Switzerland}

\author{Lirand\"e Pira}
\affiliation{Centre for Quantum Technologies (CQT), National University of Singapore, Singapore}

\author{Carmen G. Almud\'ever}
\affiliation{Universitat Polit\`ecnica de Val\`encia, Spain}

\author{Sergi Abadal}
\affiliation{Universitat Polit\`ecnica de Catalunya, Barcelona, Spain}

\author{Eduard Alarcon}
\affiliation{Universitat Polit\`ecnica de Catalunya, Barcelona, Spain}



\begin{abstract}
Quantum Signal Processing (QSP) is a powerful quantum framework for generating and approximating univariate polynomials. However, QSP is often limited by circuit-depth bottlenecks and parity constraints on the class of realizable polynomials. In this work, we introduce Weighted Quantum Signal Processing (WQSP), an extension of QSP in which a weight function is assigned to the central rotation operator. This formulation provides a deeper understanding of QSP, which emerges as the special case of WQSP with unit weights. The choice of weights fundamentally determines the structure and expressive capabilities of WQSP circuits. When the weights are natural numbers greater than one, WQSP reduces to a pruned version of QSP, revealing parameter redundancies in the standard framework. Through an appropriate selection of integer weights, WQSP achieves linear-to-exponential reductions in the number of parameters required to realize arbitrary bounded univariate polynomials while preserving approximation quality. For generic weights, we establish corresponding approximation error bounds and show that, in many cases, the approximation is exact. We analyze WQSP from both a deterministic perspective, where polynomial generation is formulated as the solution of a linear system, and a quantum machine learning perspective, where WQSP serves as a structured and expressive quantum learning model. We further employ this learning framework to parameterize learnable activation functions in Kolmogorov--Arnold Networks for multivariate function approximation. Our results show that WQSP provides a compact, flexible, and theoretically grounded framework for realizing arbitrary univariate polynomials while requiring significantly fewer trainable parameters than conventional QSP. This yields expressive and parameter-efficient neural architectures, highlighting the potential of WQSP as a scalable primitive for quantum-enhanced machine learning.
\end{abstract}


%
\maketitle
\onecolumngrid

\vspace{0.3cm}
\twocolumngrid
\vspace{0.6cm}
\hrule 
\vspace{0.4cm}

\tableofcontents



\section{Introduction}
In recent years, developments in quantum algorithm theory have revealed a unifying perspective i.e. many of the most influential quantum algorithms can be understood as instances of implementing a function of a matrix or Hamiltonian, denoted by $f(H)$\cite{gilyen,chuang}. This perspective encompasses a broad class of algorithms including Hamiltonian simulation, quantum search, factoring, quantum walks, quantum linear system solvers, and several modern optimization and machine learning routines. Consequently, the ability to efficiently realize spectral transformations of operators has emerged as one of the central objectives in quantum algorithm design. Over time, several frameworks have been developed for constructing such Hamiltonian transformations. Early approaches relied on phase estimation techniques, which formed the basis of algorithms such as HHL for quantum linear systems. Later developments introduced methods based on linear combinations of unitaries (LCU)\cite{berry2015hamiltonian, Childs2017Quantum}, enabling improved asymptotic scaling for Hamiltonian simulation and matrix function implementation. Among these approaches, Quantum Signal Processing (QSP) has emerged as one of the most versatile and powerful techniques.

The central idea behind QSP is to approximate a target continuous function through polynomial transformations of the eigenvalues of a unitary operator $U$ that encodes the underlying Hamiltonian or matrix. QSP realizes high-degree polynomial transformations using remarkably simple quantum circuits and requiring only a minimal number of ancilla qubits by interleaving signal operators with carefully chosen phase rotations. This framework has led to asymptotically optimal algorithms for Hamiltonian simulation and has become a central primitive in modern quantum algorithm design.

QSP was later generalized through Quantum Singular Value Transformation (QSVT)\cite{gilyen}, where QSP appears as a special case corresponding to polynomial transformations of eigenvalues encoded within singular value transformations. While QSVT considerably broadens the class of realizable matrix transformations, practical implementations typically require larger fault-tolerant quantum circuits and additional encoding overhead. In contrast, QSP operates within a significantly simpler single-qubit signal-processing framework, making it particularly attractive for constructing resource-efficient quantum algorithms. In fact, QSP has already been experimentally demonstrated on trapped-ion and noisy quantum devices \cite{explim,Kikuchi2023}, and has been explored in applications including Hamiltonian simulation, proportional sampling, and related quantum algorithmic tasks \cite{martyn,Dalzell2025Quantum,laneve,chuang}.

Despite its success, QSP still suffers from several practical limitations. Most notably, standard QSP imposes structural constraints on the family of realizable polynomials. To approximate arbitrary polynomial transformations, one typically decomposes the target polynomial into parity sectors and, for complex-valued functions, further separates real and imaginary components before recombining them through procedures such as linear combination of unitaries and amplitude amplification. This process increases circuit depth and implementation overhead. Hence, Generalized Quantum Signal Processing (GQSP)\cite{GQSP} was introduced in the literature \cite{GQSP}, which alleviates these parity constraints by incorporating generic $\mathrm{SU}(2)$ rotations together with an ancillary qubit register.

However, the use of generic $\mathrm{SU}(2)$ rotations instead of optimized, hardware-native gates forces the compiler to decompose each rotation into a longer sequence of elementary operations (such as CNOT gates and single-qubit pulses). This decomposition significantly increases the circuit depth for a fixed polynomial degree $d$, thereby accelerating fidelity loss due to cumulative gate noise and decoherence before the transformation is completed.

In addition, incorporating additional ancilla-based control for these rotations introduces substantial hardware overhead through increased ancilla requirements and controlled operations. In NISQ architectures, every additional controlled interaction between the ancilla and the target qubit increases implementation complexity and error accumulation, reducing the precision with which the target polynomial $P(x)$ can be realized.

A central problem in QSP is the recovery of phase parameters required to realize a target polynomial transformation, and there exists a large body of literature devoted to this problem. For a detailed introduction to QSP, we refer the reader to \cite{lin2025mathematical}. Constructive phase synthesis algorithms have been developed for both complex and real polynomial settings. In particular, \cite{gilyen,chuang} proposed an $O(d^2)$ algorithm for recovering the phase vector corresponding to a degree-$d$ polynomial through repeated matrix operations and phase reconstruction. Subsequent works improved optimization-based phase recovery \cite{linlin2,Chao2020Finding}, characterized the optimization landscape of QSP \cite{linlinenergy}, and introduced iterative approaches including fixed-point methods \cite{Linliniter,dong2024infinite}.

Alternative approaches focused on complementary polynomial construction for real-valued QSP, including root-finding methods over Laurent polynomials \cite{Haah2019product,gilyen}, contour-integral formulations such as the Prony method \cite{Ying2022}, linear-system-based methods such as the halving algorithm \cite{Chao2020Finding}, and nonlinear Fourier analysis techniques \cite{Alexis2024Quantum,linlinenergy}. Extensions to infinite QSP were later introduced for large-degree polynomial and continuous function approximation \cite{dong2024infinite}. More recently, \cite{linlininfinite} proposed a provably stable numerical algorithm based on Riemann--Hilbert factorization, referred to as the Riemann--Hilbert Weiss algorithm, for computing QSP phase factors. The authors showed existence and uniqueness of solutions for almost all admissible QSP representations. However, computing each phase factor requires solving a Riemann--Hilbert factorization problem via linear systems, resulting in an $O(d^3)$ cost per phase factor for a $d$-degree polynomial.

Despite substantial progress in phase synthesis, practical implementations of QSP for high-degree polynomial approximation still require circuits of increasing depth together with a growing number of trainable phase parameters which is suitable for fault-tolerant circuits. Consequently, although existence of phase solutions is guaranteed \cite{chuang} for admissible polynomial transformations, recovering and optimizing these parameters remains a nontrivial computational task in practice, particularly for resource-constrained and near-term quantum architectures.

\noindent\textbf{Contributions.} In this work, we introduce Weighted Quantum Signal Processing (WQSP), a natural extension of the standard QSP framework that addresses a key limitation of existing QSP-based approaches: the increasing circuit depth and parameter count required for high-degree polynomial approximation. WQSP provides a low-depth, parameter-efficient framework for constructing polynomial transformations and extends naturally to the approximation of general continuous functions through a learning-based formulation.

A key feature of WQSP is the incorporation of classical weight encoding directly into the quantum signal-processing circuit (Section \ref{sec:wqsp}). This additional structure enables the contribution of each layer to be modulated while preserving the overall circuit architecture, providing greater flexibility in constructing polynomial transformations. When the weights are restricted to integers, WQSP reduces to a pruned variant of QSP in which repeated or aggregated signal operations effectively lower circuit depth while preserving expressive power. In this regime, WQSP retains the structure of QSP but removes redundant rotations, resulting in more compact implementations. Overall, the strength of WQSP lies in its ability to trade classical weight design for quantum circuit efficiency, thereby enabling flexible control over circuit depth, parameter count, and parity structure while maintaining strong approximation capabilities for polynomial and continuous functions.

In this work, we first show that, under appropriate choices of integer weight vectors, WQSP can match the approximation performance of standard QSP while requiring significantly fewer parameters and reduced circuit depth. This demonstrates that QSP contains a substantial degree of redundancy that can be eliminated through structured weight design without sacrificing expressive power. For more general weight choices, we derive explicit approximation bounds that characterize the trade-off between circuit depth and approximation accuracy. In particular, increasing compression in the weight structure leads to reduced circuit depth at the cost of a controlled increase in approximation error. This provides a quantitative mechanism for balancing expressivity and resource efficiency within the WQSP framework. In particular, the learning-based formulation, which we also incorporate for WQSP, naturally extends beyond polynomial approximation and can be applied to the broader task of approximating continuous functions (Section \ref{polyfit}).

We also employ Weighted QSP to realize the learnable activation functions of classical Kolmogorov-Arnold Networks (KANs) for multivariate function approximation.
Artificial neural networks form the foundation of modern machine learning owing to their ability to approximate highly nonlinear functions. Classical multilayer perceptrons (MLPs) achieve this through compositions of affine transformations and point-wise nonlinear activation functions, with their expressive power guaranteed by universal approximation theorems. However, standard MLPs often suffer from the curse of dimensionality, motivating alternative representations that exploit structure in target functions to reduce parameter complexity (Section \ref{polycirc}, Section \ref{wqspKAN}). KANs \cite{liu2024kan} have recently emerged as a compelling alternative architecture in this context. They are inspired by the Kolmogorov--Arnold Representation Theorem (KART)~\cite{SCHMIDTHIEBER2021119}, also known as the Kolmogorov Superposition Theorem, which originated from Hilbert's thirteenth problem. KART states that every continuous function $f$ defined on a compact subset of $\mathbb{R}^n$ admits the representation
\begin{equation}\label{KART}
f(x_1,\ldots,x_n)
=\sum_{q=0}^{2n}\phi_q\!\left(\sum_{p=1}^{n}\psi_{pq}(x_p)\right),
\end{equation}
where $\phi_q$ and $\psi_{pq}$ are continuous univariate functions, with $\psi_{pq}$ independent of the target function $f$. Motivated by this representation, KANs replace fixed activation functions with learnable univariate functions defined on network edges, while aggregation is performed at the nodes. This architectural shift improves expressivity, parameter efficiency, and interpretability, particularly for functions exhibiting compositional structure. Consequently, KANs have found applications in computer vision, time-series analysis, physics-informed learning, and graph representation learning~\cite{cheon2024demonstrating,cang2024can,dong2024kolmogorov,vaca2024kolmogorov,patra2025physics,Zhang2025physics,carlo2025kolmogorovarnold,kiamari2024gkan}.

In this work, we employ Weighted Quantum Signal Processing (WQSP) to parameterize the learnable edge activation functions of KANs, yielding a quantum-native architecture for multivariate function approximation. Each activation is represented by a WQSP-generated polynomial, allowing expressive univariate functions to be realized using compact, low-depth variational quantum circuits. The resulting WQSP-KAN framework inherits the compression properties of WQSP, significantly reducing the number of trainable parameters and circuit depth required to represent the activation functions, while avoiding the block-encoding overheads present in existing quantum KAN architectures\cite{ivashkov2024qkan}. Rather than treating variational quantum circuits as standalone predictive models, our approach harnesses their expressive power at the level of univariate function approximation, closely aligning with the underlying Kolmogorov--Arnold decomposition. This hybrid framework combines the interpretability and compositional structure of KANs with the expressive polynomial representations of WQSP, providing a compact, parameter-efficient, and quantum-native approach to multivariate function approximation.

A key advantage of WQSP is its hardware-native gate mapping. While GQSP uses general $\mathrm{SU}(2)$ rotations that typically require $2$ to $3$ physical pulses via Euler decomposition, WQSP employs $R_x(\theta)$ rotations that can be implemented with a single modulated pulse. Consequently, this reduces physical gate depth and error accumulation on NISQ hardware.

Overall, WQSP provides a structured, low-parameter compression of QSP while extending its capabilities to flexible function approximation through weighted signal transformations. The resulting framework offers enhanced control over approximation behavior and has the potential to serve as a scalable primitive for quantum machine learning and other quantum algorithmic frameworks.

\noindent\textbf{Organization.} The paper is organized as follows. In Section~\ref{sec:wqsp}, we review the basics of QSP and WQSP. In Section~\ref{polygen}, we analyze the class of functions generated by WQSP and studying their structures in detail. In Section~\ref{polyfit}, we investigate how WQSP with integer weights can be used to approximate univariate polynomials and continuous functions with significantly fewer parameters and reduced circuit depth compared to QSP. We also present both deterministic and learning-based methods for recovering phase angles corresponding to a target function. For integer weights, we further derive explicit bounds characterizing the effect of generic weight choices on approximation accuracy. In Section~\ref{polycirc}, we describe how functions and polynomials generated by WQSP can be extracted from the corresponding quantum circuits. Subsequently, in Section~\ref{wqspKAN}, we integrate WQSP into Kolmogorov--Arnold Network (KAN) architectures for multivariate function approximation, where WQSP serves as a primitive for constructing expressive nonlinear activation functions. Finally, we conclude the paper with a summary of results and directions for future work.

\section{Weighted Quantum Signal Processing} \label{sec:wqsp}
In this section, we provide a mathematical overview of Quantum Signal Processing (QSP) and its extension, Weighted Quantum Signal Processing (WQSP). We highlight the fundamental differences between QSP and WQSP and demonstrate how the additional flexibility introduced by WQSP can be exploited to realize polynomial transformations with substantially fewer resources than standard QSP.


\subsection{Preliminaries on Quantum Signal Processing}
Quantum Signal Processing (QSP) \cite{gilyen} is a quantum framework for realizing polynomials over the interval $[-1,1]$, originally introduced in \cite{chuang}. More generally, the expressive power of QSP extends beyond polynomial generation. Through polynomial approximation theory, QSP can be used to approximate arbitrary continuous functions in the space $C[-1,1]$. In particular, by the Stone--Weierstrass theorem \cite{deBranges1959StoneWeierstrass}, every continuous function defined on a closed interval can be uniformly approximated to arbitrary accuracy by polynomial functions. Consequently, QSP provides a general mechanism for implementing continuous function transformations and has emerged as a fundamental primitive in quantum algorithms, Hamiltonian simulation, singular value transformations \cite{gilyen} and quantum neural networks\cite{Daskin2024QuantumKAN}.

In quantum circuit model of computation,  Quantum signal processing\cite{daskin2024quantum,gilyen,GQSP} is a quantum circuit with a corresponding unitary operator $U^{\mathrm{QSP}}_{\boldsymbol{\Phi}}(x)$. The circuit $U^{\mathrm{QSP}}_{\boldsymbol{\Phi}}(x)$ is constructed using a sequence of phase parameters $\boldsymbol{\Phi} = (\phi_0, \phi_1, \ldots, \phi_d)^T\in \mathbb{R}^{d+1}$ and takes the following form.

\begin{eqnarray}\label{qspdef}
    U^{\mathrm{QSP}}_{\boldsymbol{\Phi}}(x)=\left(\prod_{j=1}^{d}  R_z(\phi_{d-j+1})U(x)\right)R_z(\phi_0)
\end{eqnarray} where $R_z(\theta)=\exp{(i\theta\sigma_z)}=\bmatrix{\exp{(i\theta)}&0\\0&\exp{(-i\theta)}}$ and $U(x)$ is a rotation operator \cite{gilyen} with the following form.
\[
U(x)
=
\bmatrix{x & i\sqrt{1 - x^2} \\
i\sqrt{1 - x^2} & x}
=
\exp\!\left(i \arccos(x)\, \sigma_x\right),
\]
where $x \in [-1,1]$.  The resulting circuit generates a complex polynomial $P(x)\in \mathbb{C}[x]$ which yields the following form.\begin{eqnarray}\label{qsp}
     U^{\mathrm{QSP}}_{\boldsymbol{\Phi}}(x)=\bmatrix{P(x)&iQ(x)\sqrt{1-x^2}\\ iQ(x)^*\sqrt{1-x^2} & P(x)^*}
\end{eqnarray} where $*$ is the Hermitian operator. Clearly, $U_{\boldsymbol{\Phi}}^{\mathrm{QSP}}(x)$ acts as a block-encoding \cite{chakraborty_et_al:LIPIcs.ICALP.2019.33} of a polynomial $P(x)\in \mathbb{C}[x]$. In other words $
    \|P(x)-\bra{0} U^{\mathrm{QSP}}_{\boldsymbol{\Phi}}(x)\ket{0}\|=0$.  Further, the polynomials $P,Q\in \mathbb{C}[x]$ holds the following properties. 
\begin{itemize}
    \item[(i)] $|P(x)|\leq 1\forall x\in [-1,1]$.
    \item[(ii)]$P(x)$ is a polynomial with degree at most $d$ with $\mathrm{deg}(Q)\leq d-1$.
    \item[(iii)]$P$ and $Q$ have fixed but opposite parity.
    \item[(iv)] $|P(x)|^2+(1-x^2)|Q(x)|^2=1 \forall x\in [-1,1]$
\end{itemize}
Indeed, the coefficients of $P$ and $Q$ vary as components of $\boldsymbol{\Phi}$ are changed.  An alternative formulation\cite{gilyen} also employs
\begin{eqnarray}\label{altdefqsp}
U(x)&& =\bmatrix{\cos{\theta}&i\sin{\theta}\\i\sin{\theta}&\cos{\theta}}\\\nonumber &&=\bmatrix{x & i\sqrt{1 - x^2} \\
i\sqrt{1 - x^2} & x}\\\nonumber && =\exp(i \theta\, \sigma_x)=R_x(\theta)    
\end{eqnarray}

which is valid for real-valued $\theta$ and leads to a linear combination of trigonometric polynomials in $x$ where $x=\cos{\theta}$. In both cases, $U(x)$ can be implemented using a single-qubit rotations about the $x$ axis. Without loss of generality, we shall consider this formulation mentioned in Equation \ref{altdefqsp} throughout the paper.  The quantum circuit corresponding to a $d$-degree polynomial represented by QSP is as follows \begin{align}\label{QSP}\hspace{-0.8cm}
    {\Qcircuit @C=1em @R=.7em {
 &\lstick{}&\gate{R_z(\phi_0)}&\gate{R_x(\theta)} &\gate{R_z(\phi_1)}&\hdots&\hdots&\gate{R_x(\theta)}&\gate{R_z(\phi_d)}&\qw\\}}
\end{align}.

Given a fixed-parity polynomial $P(x)$ of degree $d$ satisfying the conditions discussed above, the objective of QSP is to determine the components of the phase vector $\mathbf{\Phi}\in \mathbb{R}^{d+1}$. In the literature, there exists an algorithm with complexity $O(d^2)$ for computing the corresponding phase angles, based on the leading coefficient of the polynomial and successive polynomial reductions \cite{gilyen}.

Alternatively, one may determine the phase vector through optimization. In practice, many implementations \cite{martyn, linlinenergy} adopt a heuristic or optimization-based approach to search for a suitable phase vector $\mathbf{\Phi}$ by solving
\begin{eqnarray*}
&&\min_{\mathbf{\Phi}} \|P(x)-\mathrm{Re}(\bra{0}U^{\mathrm{QSP}}{\Phi}(x)\ket{0})\|{L_2}\\
&&=\min_{\mathbf{\Phi}} \sqrt{\int_{-1}^{1}\left(P(x)-\mathrm{Re}(\bra{0}U^{\mathrm{QSP}}_{\Phi}(x)\ket{0})\right)^2dx}.
\end{eqnarray*}

This approach avoids repeatedly performing explicit phase constructions based on leading coefficients and recursive matrix operations by instead presenting the determination of $\mathbf{\Phi}$ as an optimization problem. Standard optimization techniques such as Nelder--Mead, gradient descent, and Levenberg--Marquardt \cite{kochenderfer2019algorithms,LM_method} can then be employed to obtain suitable phase parameters.

Such a formulation establishes a learning-based framework in which the underlying polynomial transformation is obtained through optimization rather than constructed solely through deterministic procedures. This perspective is particularly useful when the polynomial representation of a target continuous function is unknown, while still allowing the fitting of predefined polynomial functions. Consequently, the framework is well suited for constructing activation functions, which play a central role in machine learning frameworks like quantum assisted Kolmogorov-Arnold Networks\cite{liu2024kan,ivashkov2024qkan}.

A fundamental limitation of QSP is its inherent parity structure. Standard QSP generates polynomials of fixed parity, implying that approximating general functions containing both even and odd components often requires decomposing the target into separate parity contributions and implementing multiple QSP constructions. For high-degree polynomial approximations and continuous function representations, this separation can significantly increase circuit complexity and further aggravate error accumulation. In this regard, research has been carried out on Generalized Quantum Signal Processing (GQSP) \cite{GQSP}, where the parity constraints on polynomial realization are alleviated by employing a quantum circuit based on general $\mathrm{SU}(2)$ rotations as signal processing operators, rather than restricting to rotations in a single basis. This construction also utilizes additional ancilla qubits and is useful for approximating functions of a Hamiltonian $H$ without requiring the Linear Combination of Unitaries (LCU) framework that appears in certain applications of QSP\cite{chuang}.

Another significant limitation of standard QSP is that the number of trainable parameters scales linearly with the degree of the target polynomial. In particular, it has been shown in \cite{gilyen} that realizing a polynomial of degree $d$ requires at least $d+1$ phase parameters. As the polynomial degree increases, this directly translates into deeper quantum circuits and greater resource requirements. Increased circuit depth not only raises execution time but also amplifies the accumulation of hardware noise and gate imperfections, making implementations more susceptible to errors and reducing overall reliability on near-term quantum devices\cite{wilson2021empirical,explim}. These observations naturally motivate the question of whether QSP admits a more compact representation. In particular, one may ask whether QSP circuits contain structural redundancies that can be pruned to produce shallower circuits with fewer parameters while preserving the same expressive power. The pursuit of this perceptible goal motivates the formulation of the extension introduced in this work.



\subsection{Fundamentals of Weighted Quantum Signal Processing}

We consider an extension of QSP, coined as Weighted Quantum Signal Processing by integrating weight parameters into the former framework. Let us consider $\mathbf{w}=(w_1,\hdots,w_k)^T\in \mathbb{R}^{k}$. Then we define WQSP as \begin{eqnarray}\label{WQSP}
    U^{\mathrm{WQSP}}_{\boldsymbol{\Phi}_\mathbf{w}}(x)=\left(\prod_{j=1}^{k} R_z(\phi_j)U(x)^{w_{j}} \right)R_z(\phi_0)
\end{eqnarray}. 
WQSP as formalized here operates in the classical-input, function-generation regime where $x$ is a classical scalar. This allows the weight $w$ to be absorbed directly into a single gate rotation. The problem of extending this framework to a weighted black-box oracle access inside full QSVT remains an {open}. Clearly, when all weights satisfy $\mathbf{w}=\mathbf{1}$, that is, when each weight equals $1$, the formulation reduces to standard Quantum Signal Processing as defined in Equation~\ref{qspdef}. A simple observation further shows that restricting $\mathbf{w}$ to $\mathbb{N}^k$ leads to an interesting interpretation of WQSP as a structured and pruned variant of QSP. Specifically, from Equation~\ref{qspdef}, setting selected phase angles $\phi_j$ to zero causes adjacent applications of the operator $U(x)$ to combine, effectively producing powers of the form $U(x)^{w_j}$ for some $w_j\in\mathbb{N}$. Consequently, the corresponding QSP circuit can be viewed as undergoing a pruning procedure in which certain intermediate $R_z$ rotations, and in turn, the phase angles $\phi_j\in \mathbf{\Phi}$ are removed, yielding an equivalent WQSP circuit with natural-valued weights. 


It may initially appear that eliminating $R_z$ rotation gates has a significant consequence viz. a reduction in the available degrees of freedom required to realize arbitrary polynomials. In other words, this restriction suggests that certain polynomials may no longer be representable within this regime due to the loss of intermediate rotation parameters. However, we show that this perspective reveals a more nuanced and structurally rich phenomenon.

First, we demonstrate using WQSP that standard QSP contains a substantial amount of redundant parameters, where not all phase rotations are essential for polynomial generation. By exploiting this redundancy, we demonstrate that the proposed framework achieves reductions in both parameter count and circuit depth without compromising the approximation accuracy of the resulting polynomial transformations.

Second, for sufficiently high-degree polynomials $P(x)$, we show that, depending on the choice of weight functions, one can construct a polynomial $P'(x)$ such that $P'(x)$ lies close to $P(x)$ in the $L_2([-1,1])$ norm. In particular, we quantify their separation using $\|P - P'\|_{L_2([-1,1])} $. Furthermore, we derive explicit upper bounds governing this norm, thereby characterizing the effective approximation accuracy achievable within the WQSP framework. We also identify classes of weight functions for which this norm converges to zero. Consequently, for large classes of high-degree polynomials, we achieve linear, quadratic, and even exponential reductions in parameter count and circuit depth. Consequently, WQSP serves as a natural extension of standard QSP, extending its expressive capability from approximating polynomials to continuous function classes.


In this work, we show how, for a given polynomial, one can recover the phase vector $\mathbf{\Phi}$ in WQSP regime. In particular, we present a deterministic framework as well as a learning-based framework for realizing one-dimensional polynomials and continuous functions.

\section{Weighted Quantum Signal Processing: Function Generation}\label{polygen}
We now proceed to analyze the structural form of functions generated within the WQSP regime. However, we first establish several preliminaries required for the ensuing analysis.

Without loss of generality, we consider for some $k\in \mathbb{N}$, $\mathbf{w}=(w_1,\hdots,w_k)^T\in \mathbb{R}^k$ such that $w_1\geq \hdots \geq w_k$.  Let us denote the Hadamard gate $H=\frac{1}{\sqrt{2}}\bmatrix{1 & 1\\1&-1}$ and let us also consider the matrix 
\begin{eqnarray}\label{M_k}
M_k=\bmatrix{M_{k-1}&\mathbf{1}\\ M_{k-1}&-\mathbf{1}}    
\end{eqnarray}
 where $k\in \mathbb{N}$, $\mathbf{1}$ is the all one vector of size $2^{k-1}\times 1$ and $M_1=\bmatrix{1}, M_{2}=\sqrt{2}H=\bmatrix{1&1\\1&-1}$. Thus, $M_k$ is of size $2^{k-1}\times k$. 
\begin{definition}\label{Reach}
    We define the \textit{Reach} of $\mathbf{w}\in \mathbb{R}^k$ to be the set $\mathcal{R}_{\mathbf{w}}=\{|\sum_{j=1}^k\eta_jw_j|:\eta_p\in\{1,-1\}\forall j\in\{1,\hdots,L\}\}$.
\end{definition}
Clearly, $|\mathcal{R}_{\mathbf{w}}|\leq 2^{k-1}$ and all elements of $\mathcal{R}_{\mathbf{w}}$ have the same parity when $\mathbf{w}\in\mathbb{N}^k$. Then we have the following lemma.
\begin{lemma}\label{basic}
Let $k \in \mathbb{N}$ and let $\mathbf{w} \in \mathbb{R}^k$ be arranged as $w_1 \geq \cdots \geq w_k$. Let $\mathcal{R}_{\mathbf{w}}$ denote the Reach of $\mathbf{w}$ as defined in Definition~\ref{Reach}. Then the elements of $\mathcal{R}_{\mathbf{w}}$ are given by the absolute values of the components of the vector $M_k \mathbf{w}$ where $M_k$ is defined in Equation \ref{M_k}.
\end{lemma}
\pf For $k=1$, the proof is trivial. Let, for $k=m$, the hypothesis holds for $\mathbf{w^{(m)}}\in \mathbb{R}^{m}$. Then for $k=m+1$, we have $M_{m+1}\mathbf{w^{(m+1)}}=\bmatrix{M_{m}&\mathbf{1}\\ M_{m}&-\mathbf{1}}\mathbf{w^{(m+1)}}$. 
Then, $M_{m+1}\mathbf{w}$ leads us to 
\begin{eqnarray}\label{mkw}
M_{m+1}\mathbf{w^{(m+1)}}=\bmatrix{M_m\mathbf{w^{(m)}}+\mathbf{1}w_{m+1}\\M_m\mathbf{w^{(m)}}-\mathbf{1}w_{m+1}}    
\end{eqnarray}
So for all the elements of $\mathcal{R}_{\mathbf{w^{(m)}}}$ obtained from $M_m\mathbf{w^{(m)}}$, an additional $w_{m+1}$ is added and then subtracted. The rest of the proof follows immediately. $\hfill\square$

Thus elements of the reach are obtained from the weights and the $M_k$ matrix. We also define the following

\begin{definition}\label{expv}
    Let $f:\mathbb{F}\rightarrow\mathbb{F},\mathbb{F}\in \{\mathbb{R},\mathbb{C}\}$ be a real or complex analytic function and $\mathbf{v}=\bmatrix{v_1&\hdots&v_k}^T\in \mathbb{F}^k$. Then, define $f_\mathbf{v}:\mathbb{F}^k\rightarrow\mathbb{F}^k, f_{\mathbf{v}}:=[f(v_1), \hdots, f(v_k)]^T$
\end{definition} For example, take $\exp_{(v)}=\bmatrix{\exp{(v_1)}&\hdots&\exp{(v_l)}}^T$ and similarly we define $\cos_{(v)}, \sin_{(v)},$ etc.

 Furthermore, starting from $N_1=1=L_1$, let us define \begin{eqnarray}\label{NkLk}
    N_k&=&\frac{1}{2}\bmatrix{N_{k-1}&L_{k-1}\\N_{k-1}&-L_{k-1}}\\\nonumber
    L_k&=&\frac{1}{2}\bmatrix{N_{k-1}&L_{k-1}\\-N_{k-1}& L_{k-1}}
\end{eqnarray}. In such cases $N_2=\frac{1}{2}\bmatrix{1&1\\1&-1},L_2=\frac{1}{2}\bmatrix{1&1\\-1&1}$ and thus $N_k$ and $L_k$ are matrices of order $2^{k-1}\times 2^{k-1}$. We also have the following proposition about the structure of $N_k$ and $L_k$. For $k\geq 2$, the matrices $N_k$ and $L_k$ are of the form $\tilde{P}(\frac{H}{\sqrt{2}})^{\otimes k-1}$ where $\tilde{P}$ is a permutation matrix. This follows from simple induction by using the definition of $L_k$ and $N_k$. Although, for $k=2$, $N_k=\frac{1}{\sqrt{2}}H$ i.e. $\tilde{P}=I$, for $k\geq 3$, however, $\tilde{P}$ can be calculated in the following way. 
\begin{proposition}\label{permutation}
 Let $k\geq 3$ and $N_k$ be the matrix defined in Equation \ref{NkLk}. Let us also define a cycle for $\alpha\in \mathbb{N},\alpha\neq1$, $\mathcal{C}_l(\alpha,2^{k-1})=\{\alpha+l\beta\pmod{2^{k-1}}|\beta\in\{0,1\hdots,\}\}$ and further let for some $m\in \mathbb{N}$, the set $\{\alpha_1,\alpha_2,...,\alpha_m\}$ be such that $2=\alpha_1<\hdots<\alpha_m<2^{k-1}$ and $\forall j\neq i, j\in\{1,\hdots,m\}, \alpha_j\not\in\mathcal{C}_{2^i}(\alpha_i,2^{k-1})$. Then $N_k=\tilde{P}(\frac{1}{\sqrt{2}}H)^{\otimes k-1}$ where $\tilde{P}$ is a permutation matrix such that $\tilde{P}=\prod_{j=1}^{k-2}P_{(\mathcal{C}_{2^j}(\alpha_j,2^{k-1}))}$   
\end{proposition}
\pf Follows from simple induction and definition of $N_k$. $\hfill\square$.\\

\textbf{Examples:}\begin{enumerate}
    \item When $k=3$, then $\tilde{P}=P_{(2,4)}$.(Verify that $N_3=\frac{1}{4}\bmatrix{1&1&1&1\\1&-1&-1&1\\1&1&-1&-1\\1&-1&1&-1}$)
    \item When $k=4$, then $\tilde{P}=P_{(2,4,6,8)}P_{(3,7)}$.
 \item When $k=5$, then $\tilde{P}=P_{(2,4,6,8,10,12,14,16)}P_{(3,7,11,15)}P_{(5,13)}$ and so on where $P_{(a_1,\hdots,a_m)}$ denotes a permutation cycle.
\end{enumerate}

Now, consider $\mathcal{S}(M_k\mathbf{w},l)=\{j \mid |e_j^T M_k\mathbf{w}|=l\}, \; l\in\mathcal{R}_{\mathbf{w}}$, where $e_j$ is the standard basis vector with $1$ in the $j$-th component and $0$ elsewhere. $\mathcal{S}(M_k\mathbf{w},l)$ can be considered inducing a value-based equivalence relation on the index set ${1,\dots,2^{k-1}}$, which partitions indices according to equal entries of $M_k\mathbf{w}$.

We arrange the elements of $\mathcal{R}_{\mathbf{w}}$ in decreasing order as $\{r_1,r_2,\ldots\}$. Denote $\mathcal{R}_{\mathbf{w}}(m)$ to be the $m$-th element of $\mathcal{R}_{\mathbf{w}}$. Further, consider a matrix $B_{M_k\mathbf{w}}$ of size $p\times 2^{k-1}$, where $p$ is the cardinality of $\mathcal{R}_{\mathbf{w}}$. The rows of $B_{M_k\mathbf{w}}$ are defined as follows:
\begin{eqnarray}\label{bk}
B_{M_k\mathbf{w}}(r,:) = \sum_{j\in \mathcal{S}(M_k\mathbf{w},\mathcal{R}_{\mathbf{w}}(r))} e_j^T
\end{eqnarray}
for $r\in \{1,\ldots,p\}$. Then the matrix $B_{M_k\mathbf{w}}$, when multiplied with $N_k$, sums and aggregates the appropriate rows of $N_k$ corresponding to a specific $T_{r_j}$, $r_j \in \mathcal{R}_{\mathbf{w}}$, collapsing them into a single row for each such $r_j$ appearing in the components of $M_k\mathbf{w}$.

For example, let $k=3, \mathbf{w}=[1,1,1]^T$. Then $M_3\mathbf{w}=[3,1,1,-1]^T$ and $\mathcal{R}_{\mathbf{w}}=\{3,1\}$ and thus, $B_{M_3\mathbf{w}}N_k$ is a $2\times 4$ matrix where \begin{eqnarray}
    B_{M_3\mathbf{w}}N_k&&=\bmatrix{1&0&0&0\\0&1&1&1}\frac{1}{4}\bmatrix{1&1&1&1\\1&-1&-1&1\\1&1&-1&-1\\1&-1&1&-1}\\\nonumber
    &&=\frac{1}{4}\bmatrix{1&1&1&1\\3&-1&-1&-1}
\end{eqnarray} When all of the elements of $M_k\mathbf{w}$ are distinct, $B_{M_k\mathbf{w}}=I$ (Identity matrix). This happens when all components of $\mathbf{w}$ are distinct and the proof is quite direct. We also have the following lemma.

\begin{lemma}\label{bklemma}
 Let $B_{M_k\mathbf{w}}$ be the matrix defined in Equation \ref{bk}. Then $B_{M_k\mathbf{w}}$ has a full row rank and singular values of $B_{M_k\mathbf{w}}$ are $\mathcal{S}(M_k\mathbf{w},\mathcal{R}_{\mathbf{w}}{(r)})$ where $1\leq r\leq |\mathcal{S}(M_k\mathbf{w},\mathcal{R}_{\mathbf{w}})|$ with the minimum singular value equal to $1$. 
\end{lemma}
\pf From the definition of $B_{M_k\mathbf{w}}$, it is immediately clear that all the rows are orthogonal to each other. Hence, it has a full row-rank. Thus, $BB^T$ yields a diagonal matrix with entries as sum of the row elements. Since $M_k$ has only one column and row with all $1$'s, the first row of $B_{M_k\mathbf{w}}$, from definition is always $e_1^T$. Thus the minimum singular value of $B_{M_k\mathbf{w}}$ is 1. $\hfill\square$

Furthermore, for any $\alpha\in \mathbb{R}$, let
\begin{eqnarray}\label{cheb1}
    T_\alpha(\cos{\theta})&&=\cos{(\alpha\theta)}
\end{eqnarray}
and similarly,

\begin{eqnarray}\label{cheb2}
    S_{\alpha}(\cos{\theta }) =\frac{\sin((\alpha+1)\theta)}{\sin\theta} 
\end{eqnarray}

For $\alpha\in \mathbb{N}$, $T$ and $S$ become Chebyshev polynomials of first and second kind respectively with degree $\alpha$.  Then we have the following result.



\begin{theorem}\label{WQSPsum}
    Let $k\in \mathbb{N}$ and $\mathbf{w} = (w_1, \dots, w_k)^T \in \mathbb{N}^k$ such that $w_1\geq \hdots \geq w_k$and \begin{eqnarray*}
    U^{\mathrm{WQSP}}_{\boldsymbol{\Phi}_\mathbf{w}}(x)=\left(\prod_{j=1}^{k}  R_z(\phi_{j})U(x)^{w_j}\right)R_z(\phi_0)
\end{eqnarray*}. Let $P(x)=\bra{0} U^{\mathrm{WQSP}}_{\boldsymbol{\Phi}_\mathbf{w}}(x)\ket{0}, Q(x)=\bra{0} U^{\mathrm{WQSP}}_{\boldsymbol{\Phi}_\mathbf{w}}(x)\ket{1}$. Further, let $\mathcal{R}_{\mathbf{w}}=\{r_1>r_2>\cdots >r_p\},\text{such that }p\leq 2^{k-1}$ and denote $\mathcal{R}_{\mathbf{w}}(m):=r_m$ denotes the $m$-th largest distinct element of $\mathcal{R}_{\mathbf{w}}$ Then \begin{enumerate}
   \item $P(x)\in \mathbb{C}[x]$ such that $P=\sum_{j=1}^{p}c_{r_j}\exp{(i\phi_0)}T_{r_j}$ where $r_j\in \mathcal{R}_{\mathbf{w}}$ as described in Definition (\ref{Reach}) and $T_{r_j}$ is the Chebyshev polynomial of the first kind with degree $r_j$. 
        \item $Q(x)=i\sqrt{1-x^2}\sum_{j=1}^{p}d_{r_j}\exp{(-i\phi_0)}S_{r_j-1}\in \mathbb{C}[x]$ where $S_{r_j-1}$ is the Chebyshev polynomial of the second kind with degree $r_j-1$, 
    \end{enumerate} where $c_{r_j},d_{r_j}\in \mathbb{C}\forall r_j\in \mathcal{R}_{\mathbf{w}} $. 
   Moreover, \begin{enumerate}
        \item $\mathbf{c}(\mathcal{R}_{\mathbf{1}})=(c_{r_1},\hdots,c_{r_{p}})^T$,
        \item $\mathbf{d}(\mathcal{R}_{\mathbf{1}})=(d_{r_1},\hdots,d_{r_{p}})^T$
        \item $\Phi=\bmatrix{\phi_1&\phi_2&\hdots&\phi_{k-1}&\phi_{k}}^T\in \mathbb{R}^k$
    \end{enumerate}  . Then, \begin{eqnarray}\label{cd2}
    \mathbf{c}(\mathcal{R}_{\mathbf{w}})&=&B_{M_k\mathbf{w}}N_k\exp_{(M_k(i\mathbf{\Phi}_{\mathbf{w}}))}\\\nonumber
    \mathbf{d}(\mathcal{R}_{\mathbf{w}})&=&B_{M_k\mathbf{w}}L_k\exp_{(M_k(i\mathbf{\Phi}_{\mathbf{w}}))}
\end{eqnarray} where $N_k,L_k$ is defined in Equation \ref{NkLk}, $\exp_{(M_k(i\mathbf{\Phi}_{\mathbf{w}}))}$ is defined in Definition \ref{expv} and $B_{M_k\mathbf{w}}$ is a $p\times 2^{k-1}$ matrix is defined in Equation \ref{bk}.
\end{theorem}
\pf We start off with the case where all $w_j$'s are distinct. i.e. $w_1>\hdots >w_k$ and thus $B_{M_k\mathbf{w}}=I, p=2^{k-1}$ and $M_k\mathbf{w}$ possess no repeated components. Let $k=1$. Then $ U^{\mathrm{WQSP}}_{\boldsymbol{\Phi}_\mathbf{w}}(x)= R_z(\phi_1)U(x)^{w_{1}} R_z(\phi_0)$. $U(x)=\exp{(i\theta\sigma_x)}\implies U^{w_1}(x)=\exp{(iw_1\theta\sigma_x)}$ where $x=\cos{\theta}$. Thus, a simple calculation gives us the following. \begin{eqnarray*}
    U(x)^{w_1}&&=\bmatrix{\cos{(w_1\theta)}&i\sin{(w_1\theta)}\\i\sin{(w_1\theta)}&\cos{(w_1\theta)}}\\&&=\bmatrix{T_{w_1}(x)&i\sqrt{1-x^2}S_{w_1-1}(x)\\i\sqrt{1-x^2}S_{w_1-1}(x)&T_{w_1}(x)}
\end{eqnarray*} where $T$ is the Chebyshev polynomial of first kind where $T$ and $S$ are defined in Equation  \ref{cheb1} and \ref{cheb2} respectively. Clearly, $\bra{0}R_z(\phi_1)U(x)^{w_1}R_z(\phi_0)\ket{0}$ yields $T_{w_1}\exp{i(\phi_1+\phi_0)}$ and  $\bra{0}R_z(\phi_1)U(x)^{w_1}R_z(\phi_0)\ket{1}$ yields $i\sqrt{1-x^2}S_{w_1-1}\exp{i(\phi_1-\phi_0)}$. $c_{r_1}=N_1\exp_{(iM_1\phi_1)},d_{r_1}=L_1\exp_{(iM_1\phi_1)}$. So the all of the induction conditions are satisfied. \\

Now let for $k=m$, the result holds. i.e. $P=\sum_{j=1}^{2^{m-1}}c_{r_j}\exp{(i\phi_0)}T_{r_j}=\sum_{j=1}^{2^{m-1}}c_{r_j}\cos{({r_j}\theta)}$ and $Q=i\sqrt{1-x^2}\sum_{j=1}^{2^{m-1}}d_{r_j}\exp{(-i\phi_0)}S_{r_j-1}=i\sum_{j=1}^{2^{m-1}}d_{r_j}\sin{({r_j}\theta)}$ for the ordered weights $(w_1,\hdots,w_m)$, such that the coefficients $\mathbf{c}(\mathcal{R}_{\mathbf{w}})$ and $\mathbf{d}(\mathcal{R}_{\mathbf{w}})$ satisfies Equation (\ref{cd2}). Clearly, without the right multiplication by $R_z(\phi_0)$, we have
\begin{eqnarray*}
    &&P=\sum_{j=1}^{2^{m-1}}c_{r_j}\cos{({r_j}\theta)}\\
    &&Q=i\sum_{j=1}^{2^{m-1}}d_{r_j}\sin{({r_j}\theta)}
\end{eqnarray*}

Now for $\mathbf{w}=(w_1,\hdots,w_{m+1})^T$, the first row of the product $\left(\prod_{j=1}^{m} R_z(\phi_j)U(x)^{w_{j}} \right)$ is $\bmatrix{P(x) & Q(x)}$. Multiplying this with the first column of $R_z(\phi_{m+1})U^{w_{m+1}}R_z(\phi_0)$ leads to the following\begin{widetext}
    \begin{eqnarray*}
\bra{0}\left(\prod_{j=1}^{m+1} R_z(\phi_j)U(x)^{w_{j}} \right)R_z(\phi_0)\ket{0}=\Biggl(P\cos{(w_{m+1}\theta)}\exp{(i\phi_{m+1})}+iQ\sin{(w_{m+1}\theta)}\exp{(-i\phi_{m+1})}\Biggr)\exp(i\phi_0)
\end{eqnarray*} 

We know \begin{eqnarray*}
    P\cos(w_{m+1}\theta)\exp{(i\phi_{m+1})}=\sum_{j=1}^{2^{m-1}}c_{r_j}\exp{(i\phi_{m+1})}\cos(r_j\theta)\cos{(w_{m+1}\theta)}\\
    iQ\sin(w_{m+1}\theta)\exp{(-i\phi_{m+1})}=-\sum_{j=1}^{2^{m-1}}d_{r_j}\exp{(-i\phi_{m+1})}\sin(r_j\theta)\sin{(w_{m+1}\theta)}
\end{eqnarray*}. 
\end{widetext} By using the simple trigonometric formula of $\cos{(A+B)}+\cos{(A-B)}=2\cos{(A)}\cos{(B)}$ and $\cos{(A-B)}-\cos{(A+B)}=2\sin{(A)}\sin{(B)}$, we see $\mathcal{R}_{\mathbf{w}}$ becomes $\{\mathcal{R}_{\mathbf{w}}+w_{m+1}\}\cup\{\mathcal{R}_{\mathbf{w}}-w_{m+1}\}$ and thus the cardinality i.e. $|\mathcal{R}_{\mathbf{w}}|$ becomes $2^m$. Hence, we get \begin{widetext}
\begin{eqnarray*}
&&\bra{0}\left(\prod_{j=1}^{m+1} R_z(\phi_j)U(x)^{w_{j}} \right)R_z(\phi_0)\ket{0}\\&=& \frac{1}{2}\exp(i\phi_0)\Biggl(\sum_{j=1}^{2^{m-1}}\Biggl(c_{r_j}\exp{(i\phi_{m+1})}+d_{r_j}\exp{(-i\phi_{m+1})}\Biggr)\cos((r_j+w_{m+1})\theta))+\\&&\Biggl(c_{r_j}\exp{(i\phi_{m+1})}-d_{r_j}\exp{(-i\phi_{m+1})}\Biggr)\cos((r_j-w_{m+1})\theta))\Biggr)
\\&=& \frac{1}{2}\exp(i\phi_0)\Biggl(\sum_{j=1}^{2^{m-1}}(c_{r_j}\exp{(i\phi_{m+1})}+d_{r_j}\exp{(-i\phi_{m+1})})T_{r_j+w_{m+1}}+\\&&(c_{r_j}\exp{(i\phi_{m+1})}-d_{r_j}\exp{(-i\phi_{m+1})})T_{r_j-w_{m+1}}\Biggr)
\end{eqnarray*}. 
\end{widetext}
We also get the following using the identities $\sin{(A+B)}+\sin{(A-B)}=2\sin{(A)}\cos{(B)}$ and $\sin{(A+B)}-\sin{(A-B)}=2\cos{(A)}\sin{(B)}$. 
\begin{widetext}
    \begin{eqnarray*}
&&\bra{0}\left(\prod_{j=1}^{m+1} R_z(\phi_j)U(x)^{w_{j}} \right)R_z(\phi_0)\ket{1}\\&=& i\frac{1}{2}\exp(-i\phi_0)(\sum_{j=1}^{2^{m-1}}(c_{r_j}\exp{(i\phi_{m+1})}+d_{r_j}\exp{(-i\phi_{m+1})})\sin((r_j+w_{m+1})\theta))+\\&&(-c_{r_j}\exp{(i\phi_{m+1})}+d_{r_j}\exp{(-i\phi_{m+1})})\sin((r_j-w_{m+1})\theta)))
\end{eqnarray*}. 
\end{widetext}

Clearly, the updated coefficients $c_{r_j}$ and $d_{r_j}$ are components of $N_k\exp_{(iM_{m+1}\Phi)}$ and $L_k\exp_{(iM_{m+1}\Phi)}$ respectively and further, $T_{r_j-w_{m+1}}$ and $T_{r_j+w_{m+1}}$ have fixed parity. Thus using induction hypothesis, we get our proof.

For the second case, if $w_j=w_l$ for some $j$ and $l$ then $M_k\mathbf{w}$ possess repeated components and the proof follows similar to the previous case. However, the terms $c_{r_j}T_{r_j},c_{r_l}T_{r_l}$ are added up for all such $j$ and $l$. Thus,  coefficients $c_{r_j}$ and $d_{r_j}$ are components of $B_{M_k\mathbf{w}}N_k\exp_{(iM_{m+1}\Phi)}$ and $B_{M_k\mathbf{w}}L_k\exp_{(iM_{m+1}\Phi)}$ since $B_{M_k\mathbf{w}}$, by construction when multiplied, aggregates the rows of $N_k$ where $M_k\mathbf{w}$ have same components and the rest of the proof follows immediately. $\hfill\square$\\




\begin{remark}
  When all components of $\mathbf{w}\in\mathbb{R}^k$ are equal to $1$, the system reduces to QSP and the reach satisfies $\mathcal{R}_{\mathbf{1}}=\{k,k-2,\hdots,(k\pmod 2)\}$. In such cases, it is easy to observe that the vector $M_k\mathbf{1}$ contains repeated components. Hence, from Theorem \ref{WQSPsum}, the coefficients corresponding to $T_{r_j}$ for all repeated values $r_j$, i.e. components belonging to the same value class, are accumulated in QSP.

For example, let $k=3$ and $\mathbf{w}=(1,1,1)^T$. Then, by Theorem~\ref{WQSPsum}, we obtain
\begin{eqnarray*}
P=c_3T_3+c_1T_1+c_1T_1+c_1T_1=c_3T_3+3c_1T_1.
\end{eqnarray*}

This follows because $1+1+1=3$, whereas all other combinations, namely $|1+1-1|$, $|1-1+1|$, and $|1-1-1|$, evaluate to $1$. Consequently, $M_k\mathbf{1}$, and therefore $\mathcal{R}_{\mathbf{1}}=\{3,1\}$.

Further, a direct calculation yields
\begin{eqnarray*}
\bmatrix{c_3\\3c_1}
=
\frac{1}{4}
\bmatrix{
1&1&1&1\\
3&-1&-1&-1
}
\exp_{(iM_k\Phi)}.
\end{eqnarray*}
The matrix $\frac{1}{4}\bmatrix{1&1&1&1\\3&-1&-1&-1}$ is clearly constructed by adding up the last $3$ rows of $N_3$. In particular

\begin{eqnarray*}
    \frac{1}{4}
\bmatrix{
1&1&1&1\\
3&-1&-1&-1
}=\underbrace{\bmatrix{1&0&0&0\\0&1&1&1}}_{B_{M_3\mathbf{1}}}\underbrace{\frac{1}{4}\bmatrix{1&1&1&1\\1&-1&-1&1\\1&1&-1&-1\\1&-1&1&-1}}_{N_3}
\end{eqnarray*}

The construction of $B_{M_3\mathbf{1}}$ easily follows from Equation \ref{bk}. First of all $M_3\mathbf{1}=\bmatrix{3&1&1&-1}^T$. Thus, from Lemma \ref{basic}, elements of the reach are obtained from absolute values of $B_{M_3\mathbf{1}}$ which gives $\mathcal{R}_{\mathbf{1}}=\{3,1\}$.  Thus,  $\mathcal{S}(M_3\mathbf{1},3)=\{1\}, \mathcal{S}(M_3\mathbf{1},1)=\{2,3,4\}$  and thus, $\bmatrix{1&0&0&0\\0&1&1&1}=B_{M_3\mathbf{1}}$ and we get our result.

 For another example, let $\mathbf{w}=[1,1,1,1]^T$. Then $M_k\mathbf{w}=[4,2,2,0,2,0,0,2]^T$. Thus $B_{M_k\mathbf{1}}N_K$ is a $3\times 8$ matrix where the first row $B_{M_k\mathbf{1}}N_K(1,:)= N_k(1,:)$, second row of $B_{M_k\mathbf{1}}N_K$ i.e. $B_{M_k\mathbf{1}}N_K(2,:)=N_k(2,:)+N_k(3,:)+N_k(5,:)+N_k(8,:)$ and $B_{M_k\mathbf{1}}N_K(3,:)= N_k(4,:)+N_k(6,:)+N_k(7,:)$. From this fact, we can find out what kind of polynomial a QSP circuit generates.
\end{remark}

\begin{corollary}\label{QSPsum}
Let $k\in \mathbb{N}$ such that \begin{eqnarray*}
    U^{\mathrm{QSP}}_{\boldsymbol{\Phi}}(x)=\left(\prod_{j=1}^{k}  R_z(\phi_{j})U(x)\right)R_z(\phi_0)
\end{eqnarray*}.Then, $\mathcal{R}_{\mathbf{1}}=\{r_1>r_2>\cdots>r_p\}=\{k,k-2,\ldots\},p\leq 2^{k-1}$. Denote $\mathcal{R}_{\mathbf{1}}(m):=r_m$ denotes the $m$-th largest distinct element of $\mathcal{R}_{\mathbf{1}}$ and let $P(x)=\bra{0} U^{\mathrm{QSP}}_{\boldsymbol{\Phi}}(x)\ket{0}, Q(x)=\bra{0} U^{\mathrm{QSP}}_{\boldsymbol{\Phi}}(x)\ket{1}$. Then \begin{enumerate}
        \item $P(x)\in \mathbb{C}[x]$ such that $P=\sum_{j=1}^{p}c_{r_j}\exp{(i\phi_0)}T_{r_j}$ where $r_j\in \mathcal{R}_{\mathbf{1}}$ and $p$ is the cardinality of $\mathcal{R}_{\mathbf{1}}$.
        \item $Q(x)=i\sqrt{1-x^2}\sum_{j=1}^{p}d_{r_j}\exp{(-i\phi_0)}U_{r_j-1}$ where $U_{r_j-1}$.
    \end{enumerate} where $c_{r_j},d_{r_j}\in \mathbb{C}\forall r_j\in \mathcal{R}_{\mathbf{1}} $. 
    Further, let \begin{enumerate}
        \item $\mathbf{c}(\mathcal{R}_{\mathbf{1}})=(c_{r_1},\hdots,c_{r_{p}})^T$,
        \item $\mathbf{d}(\mathcal{R}_{\mathbf{1}})=(d_{r_1},\hdots,d_{r_{p}})^T$
        \item $\Phi=\bmatrix{\phi_1&\phi_2&\hdots&\phi_{k-1}&\phi_{k}}^T\in \mathbb{R}^k$
    \end{enumerate}  . Then, \begin{eqnarray}\label{cd}
    \mathbf{c}(\mathcal{R}_{\mathbf{1}})&=&B_{M_k\mathbf{1}}N_k\exp_{(M_k(i\mathbf{\Phi}_{\mathbf{w}}))}\\\nonumber
    \mathbf{d}(\mathcal{R}_{\mathbf{1}})&=&B_{M_k\mathbf{1}}L_k\exp_{(M_k(i\mathbf{\Phi}_{\mathbf{w}}))}
\end{eqnarray} where $B_{M_k\mathbf{1}}$ is defined in Equation \ref{bk} and $N_k,L_k$ are defined in Equation \ref{NkLk}.
\end{corollary}

\pf Follows immediately from the discussions above and Theorem \ref{WQSPsum}. $\hfill\square$

\begin{remark}
\begin{enumerate}
   Indeed, for QSP the maximum element of $\mathcal{R}_{\mathbf{1}}=k$. Thus,  \begin{eqnarray*}
    U^{\mathrm{QSP}}_{\boldsymbol{\Phi}}(x)=\left(\prod_{j=1}^{k}  R_z(\phi_{1})U(x)\right)R_z(\phi_0)
\end{eqnarray*} generates at most $k$-degree polynomials. 
    
\item In Theorem~\ref{WQSPsum}, and Corollary~\ref{QSPsum}, the coefficients $c_{r_j}$ and $d_{r_j}$ are symbolic in nature. In other words, one may replace $c_{r_j}\rightarrow c_{r_j}\exp(-i\phi_0)$ and $d_{r_j}\rightarrow d_{r_j}\exp(i\phi_0)$ without altering the resulting polynomial. Hence, without loss of generality, one may set $\phi_0=0$ whenever $k\geq 1$.

This observation is particularly interesting since it shows that, for WQSP and consequently QSP, one parameter is already redundant. Thus, for $k\geq 1$,
\begin{eqnarray}\label{QSPnew}
U^{\mathrm{QSP}}_{\boldsymbol{\Phi}}(x)
=
\left(
\prod_{j=1}^{k}
R_z(\phi_j)U(x)
\right)
\end{eqnarray}
and
\begin{eqnarray}\label{WQSPnew}
U^{\mathrm{WQSP}}_{\boldsymbol{\Phi}_{\mathbf{w}}}(x)
=
\left(
\prod_{j=1}^{k}
R_z(\phi_j)U(x)^{w_j}
\right).
\end{eqnarray}

Therefore, only $d$ parameters (rather than $d+1$) are required for a QSP implementation of a polynomial of degree at least $d>0$. However, for the trivial case $d=0$, the additional rotation $R_z(\phi_0)$ remains necessary. In the later sections of this work, we show how WQSP allows this parameter count to be reduced further.

    \item For any weight vector $\mathbf{w}$, whenever $M_k\mathbf{w}$ or $\mathcal{R}_{\mathbf{w}}$ don't produce distinct elements, the main idea is to always find the rows of $N_k$ corresponding to the similar elements $c_{r_j}$ and collapse them into a single row vector via addition.
    \item It is also of note that the polynomials in Theorem \ref{WQSPsum} and as a consequence polynomials in Corollary \ref{QSPsum} harbors some interesting properties already observed in \cite{gilyen}.
\end{enumerate}

\end{remark}

\begin{corollary}\label{wqspcoro}
    Let $\mathbf{w} = (w_1, \dots, w_k) \in \mathbb{N}^k$ be a fixed weight vector and $N = \sum_{j=1}^k w_j$. Let $P(x)=\bra{0}U^{\mathrm{WQSP}}_{\boldsymbol{\Phi},\mathbf{w}}(x)\ket{0}$ and $R(x)=\frac{1}{i\sqrt{1-x^2}}\bra{0}U^{\mathrm{WQSP}}_{\boldsymbol{\Phi},\mathbf{w}}(x)\ket{1}$. Then, 
\begin{enumerate}
    \item $P(x)\in \mathbb{C}[x]$ and $|P(x)| \le 1$ for all $x \in [-1, 1]$.
    \item $P(x)$ has parity $N \pmod 2$ and $R(x)\in \mathbb{C}[x]$ has parity $N-1 \pmod 2$.
    \item  $|P(x)|^2+(1-x^2)|R(x)|^2=1$.
\end{enumerate}
\end{corollary}
\pf Follows directly from Theorem \ref{WQSPsum} and using the fact that $P$ and $R$ are row elements of a unitary matrix. $\hfill\square$

So far, we have characterized the class of polynomials that can be generated by WQSP circuits. However, a natural question arises: given a bounded polynomial, can one construct a suitable WQSP framework that realizes it? For standard QSP, the answer is affirmative as established in \cite{gilyen}.

For WQSP, however, this question remains open. In the following, we address this question and investigate the extent to which QSP admits parameter redundancies, demonstrating how WQSP can provide a more efficient realization framework with advantages in parameter count, and circuit depth over standard QSP. 
\section{Weighted Quantum Signal Processing: Polynomial Fitting}\label{polyfit}
In this section, we address the following question: given a polynomial $P(x)$ satisfying conditions of Corollary~\ref{wqspcoro}, can one construct a corresponding WQSP circuit with suitable parameters? Equivalently, the objective is to recover the phase vector $\mathbf{\Phi}_{\mathbf{w}}$ for a given polynomial. Clearly, we have observed so far that WQSP produces parity-fixed polynomials. Thus, we shall analyze how WQSP can realize parity-fixed polynomials with less resources. In this regime, we present both deterministic and learning-based approaches for constructing the corresponding circuits. Particularly, we see that the learning-based approaches naturally extend beyond polynomial realization and can be applied to approximate univariate bounded continuous functions.

\subsubsection{Weighted Quantum Signal Processing: Polynomial fitting Algorithm}

In this section, we shall deal with realizing polynomials with a fixed parity. This is because WQSP generates fixed parity polynomials from Theorem \ref{WQSPsum}. Further, any arbitrary polynomial $P$ can be written as $P(x)=P_{\mathrm{even}}(x)+P_{odd}(x)$ where $P_{\mathrm{even}}(x)=\frac{1}{2}(P(x)+P(-x))$ and $P_{\mathrm{odd}}(x)=\frac{1}{2}(P(x)-P(-x))$ and thus, for an arbitrary continuous function or a polynomial, we need $2$- quantum circuits to approximate them through the WQSP regime.  

Let us consider a fixed-parity polynomial $P(x)$ of maximum degree $d$. Without loss of generality, we assume that $d$ is even. The case where $d$ is odd follows analogously, and all subsequent constructions and results remain unchanged.

\begin{enumerate}
    \item Let, \begin{eqnarray*}
    P(x)=a_dx^d+a_{d-2}x^{d-2}+\hdots+a_0\in \mathbb{C}[x]
\end{eqnarray*} such that $|P(x)|\leq 1$. We would like to change the basis from $\{1,x,x^2,\hdots\}$ to $\{T_0,T_1,T_2,\hdots,\}$ such that \begin{eqnarray}
    P(x)=c_{d}T_{d}(x)+c_{d-2}T_{d-2}(x)+\hdots + c_0T_0(x)
\end{eqnarray}.
\item  Let, $\mathrm{nz}(\mathbf{c})=\{j|c_j\neq 0,j\in\{d,d-2,\hdots,0\}\}$ (For odd $d$, $\mathrm{nz}(\mathbf{c})=\{j|c_j\neq 0,j\in\{d,d-2,\hdots,1\}\}$). Without loss of generality let $c_d\neq 0$.  Clearly $|c_j|\leq \frac{1}{\mathrm{nz}(\mathbf{c})}\forall j$ since $|P(x)|\leq 1$. 

Let us order $\mathrm{nz}(\mathbf{c})=\{j_1,\hdots,j_{m}\}$ such that $j_1=d\geq\hdots\geq j_m$. We choose a partition of $d$ say $\mathbf{w}=(w_1,\hdots,w_k)^T$ such that $w_1\geq\hdots\geq w_k,\sum_{j=1}^kw_j=d$ and $\mathcal{R}_{\mathbf{w}}=\mathrm{nz}(\mathbf{c})$. If such a partition can't be found or it doesn't exist, choose $\mathbf{w}$ such that $w_1\geq\hdots\geq w_k,\sum_{j=1}^kw_j=d$ and $\mathrm{nz}(\mathbf{c})\subset\mathcal{R}_{\mathbf{w}}$. 
\item Denote $\mathcal{S}(M_k\mathbf{w},l)=\{j||e_j^TM_k\mathbf{w}|=l\}, l\in\mathrm{nz}(\mathbf{c})(\text{ or }\mathcal{R}_{\mathbf{w}})$ and construct $B_{M_k\mathbf{w}}$, a $p\times 2^{k-1}$ matrix from Equation \ref{bk}, such that \begin{eqnarray}\label{bk2}
    B_{M_k\mathbf{w}}(r,:)=\sum_{j\in \mathcal{S}(M_k\mathbf{w},(\mathrm{nz}(\mathbf{c}))_{(r)})}e_j^T
\end{eqnarray} where $r\in\{1,\hdots,p\}$, $p=|\mathrm{nz}(\mathbf{c})|$, and $(\mathrm{nz}(\mathbf{c}))_{(r)}$ is the $r$-th element of the set $\mathrm{nz}(\mathbf{c})$. If no suitable partition is found or it doesn't exist, then \begin{eqnarray}\label{bk3}
    B_{M_k\mathbf{w}}(r,:)=\sum_{j\in \mathcal{S}(M_k\mathbf{w},\mathcal{R}_{\mathbf{w}}{(r)})}e_j^T
\end{eqnarray} where $p=|\mathcal{R}_{\mathbf{w}}|$, and $\mathcal{R}_{\mathbf{w}}{(r)}$ is the $r$-th element of the set $\mathcal{R}_{\mathbf{w}}$. 
\item Solve the linear system \begin{eqnarray}\label{syst}
    \mathbf{c}&=& B_{M_k\mathbf{w}}N_k\mathbf{y}
\end{eqnarray} where $\mathbf{y}=\exp_{(M_k(i\mathbf{\Phi}_{\mathbf{w}}))}$. It is of note that $B_{M_k\mathbf{w}}N_k$ is a matrix of size $p\times 2^{k-1}$. Clearly, $p\leq 2^{k-1}$ and when all components of $\mathbf{w}$ are distinct, $p=2^{k-1}$ and $B_{M_k\mathbf{w}}=I$ (Identity matrix because no rows can be collapsed into other as all elements of $M_k\mathbf{w}$ are distinct).


\begin{enumerate}
    \item \textbf{Case-I}: $B_{M_k\mathbf{w}}=I$
    In  this case, $B_{M_k\mathbf{w}}N_k=N_k$ and $N_k=\tilde{P}(\frac{1}{\sqrt{2}}H)^{\otimes k-1}\implies N_k^{-1}=({\sqrt{2}}H)^{\otimes k-1}\tilde{P}^T$ where $\tilde{P}$ is a permutation matrix defined in Proposition \ref{permutation}. This follows directly from the orthogonality and symmetric property of rows of $\frac{1}{\sqrt{2}}H$.
    \item \textbf{Case-II}: $B_{M_k\mathbf{w}}\neq I$
    In  this case, $B_{M_k\mathbf{w}}N_k$ is an under-determined system and there are infinitely many solutions. 
\end{enumerate}
In any case, a solution exists and let us denote it by $\mathbf{b}=\exp_{(M_k(i\mathbf{\Phi}_{\mathbf{w}}))}$. If the polynomial is real, then $\mathbf{b}=\cos_{(M_k(\mathbf{\Phi}_{\mathbf{w}}))}$ (This follows from taking out the real parts in Theorem \ref{WQSPsum}).  If the polynomial is purely imaginary, 
$\mathbf{b}=i\sin_{(M_k(\mathbf{\Phi}_{\mathbf{w}}))}$. For the sake of simplicity, we stick to real polynomials since the complex case will follow similarly. 
\begin{algorithm}[H]
\caption{Approximating a real polynomial of degree $d$ with WQSP}\label{algo1}
\textbf{Provided:} Polynomial $P(x)\in \mathbb{R}[x]$ such that $P(x)$ satisfies conditions of Corollary \ref{wqspcoro} and $\mathrm{deg}P(x)=d$ (WLOG, d is even), $P(x)=c_dT_d+c_{d-2}T_{d-2}+\hdots+c_0T_0$ and precision parameter $\epsilon$

\textbf{Input:} $k\in \mathbb{N},$ and weight vector $\mathbf{w}\in\mathbb{N}^k$ such that $w_1\geq\hdots\geq w_k$ and $\sum_{j=1}^kw_j=d$.

\textbf{Output:} Phase vector $\Phi_\mathbf{w}\in \mathbb{R}^k$ such that $U^{\mathrm{WQSP}}_{\boldsymbol{\Phi}_{\mathbf{w}}}(x) = \left( \prod_{j=1}^{k} R_z(\phi_j)U(x)^{w_j} \right)$

\begin{algorithmic}[1]
\Procedure{}{WQSP Polynomial Fitting}    
\State $\mathbf{c}\rightarrow (c_d,c_{d-2},\hdots,c_0)^T$.
\State $\mathrm{nz}(\mathbf{c})\rightarrow\emptyset$
\For{$j=d;d=d-2;d\geq 0$}
\If{$c_j\neq 0$}
\State $\mathrm{nz}(\mathbf{c})\rightarrow \mathrm{nz}(\mathbf{c})\cup \{j\}$
\EndIf
\EndFor
\State $\mathcal{R}_{\mathbf{w}}\rightarrow\emptyset$
\State Construct $M_k,N_k,L_k$
\State Calculate $M_k\mathbf{w}$
\For{$j=1;j++;j\leq 2^{k-1}$}
\State $\mathcal{R}_{\mathbf{w}}\rightarrow \mathcal{R}_{\mathbf{w}}\cup \{|e_j^TM_k\mathbf{w}|\}$
\EndFor
\If{$\mathrm{nz}(\mathbf{c})=\mathcal{R}_{\mathbf{w}}$}
\State $p\rightarrow|\mathrm{nz}(\mathbf{c})|$
\State $B_{M_k\mathbf{w}}\rightarrow O$
\For{$r=1,r++,r\leq p$}
\State Construct $\mathcal{S}(M_k\mathbf{w},\mathrm{nz}(\mathbf{c})_{(r)})$
\State Construct $B_{M_k\mathbf{w}}(r,:)$ from Eq. \ref{bk2}
\EndFor
\Else
\State $p\rightarrow|\mathcal{R}_\mathbf{w}|$
\State $B_{M_k\mathbf{w}}\rightarrow O$
\For{$r=1,r++,r\leq p$}
\State Construct $\mathcal{S}(M_k\mathbf{w},\mathcal{R}_\mathbf{w}{(r)})$
\State Construct $B_{M_k\mathbf{w}}(r,:)$ from Eq. \ref{bk3}
\EndFor
\EndIf
\State Solve for $\mathbf{y}$ the system $\mathbf{c}= B_{M_k\mathbf{w}}N_k\mathbf{y}$ in Eq. \ref{syst}.
\State $\mathbf{y}\rightarrow \mathbf{b}$
\State Solve for $\mathbf{\Phi_w}$ the system ${M_k\mathbf{\Phi}_{\mathbf{w}}}=\arccos_{\mathbf{b}}$ in Eq. \ref{important}.
\If{Equation \ref{important} is not consistent}
\State Use Least Squares to Solve Eq. \ref{important}
\State $\mathbf{\mathbf{\Phi_w}}\rightarrow \mathbf{\mathbf{\Phi_w}}^{\mathrm{LS}}$
\EndIf
\State Construct $\mathbf{c}^{\mathrm{LS}}$ from Eq. \ref{syst} and $P_{\mathrm{LS}}(x)$.
\If{$\|P(x)-P_{\mathrm{LS}}(x)\|_{L^2}\leq \epsilon$}
\State Break
\Else
\State Choose another $k,\mathbf{w}$
\EndIf
\EndProcedure
\end{algorithmic}
\end{algorithm}
\item From $\mathbf{b}=\cos_{(M_k(i\mathbf{\Phi}_{\mathbf{w}}))}$ we get \begin{eqnarray}\label{important}
{M_k\mathbf{\Phi}_{\mathbf{w}}}=\arccos_{\mathbf{b}}
\end{eqnarray} which needs to be solved. 

\end{enumerate}

However, it is easy to observe that Equation~\ref{important} defines an overdetermined system since $M_k$ has size $2^{k-1}\times k$. Consequently, an exact solution need not exist unless Equation~\ref{important} is consistent. At first glance, this may suggest that WQSP is restrictive. However, one may instead solve Equation~\ref{important} in the least-squares sense to obtain an approximate solution and reconstruct a polynomial $P_{\mathrm{LS}}(x)$. We illustrate this through Algorithm \ref{algo1}

It is worth noting that Algorithm~\ref{algo1} applies analogously for odd $d$, in which case the iteration terminates at $d=1$. Further, it is interesting to notice that Equation~\ref{important} becomes consistent in the case of QSP. We show that, by appropriately choosing integer weights, one can recover behavior analogous to QSP while achieving the same polynomial realization using significantly shallower circuits. Further, we derive an upper bound on the distance between $P(x)$ and $P_{\mathrm{LS}}(x)$.
\subsubsection{WQSP Error Bounds for Polynomial Fitting Algorithm}
In this section, we derive error bounds for deterministic algorithm to extract the phase vectors in WQSP used in the realization of real-valued polynomials. The complex-valued case follows analogously. We further show that, by choosing appropriate weight functions in WQSP, exact polynomial realization can be achieved in certain cases with fewer parameters compared to QSP. For arbitrary weights, we show that the approximation error becomes negligible whenever the polynomial coefficients are sufficiently small.

Let us consider $k\geq 2\in \mathbb{N}$, and denote $\nu=2^{k-1}$. From Equation \ref{syst}, consider \begin{eqnarray}\label{W}
    W=N_k^{-1}B_{M_k\mathbf{w}}^{\dagger}
\end{eqnarray} where $B_{M_k\mathbf{w}}^{\dagger}$ is the Moore-Penrose pseudo-inverse \cite{datta2010numerical} of $B_{M_k\mathbf{w}}$ defined as \begin{eqnarray*}
    B_{M_k\mathbf{w}}^{\dagger}=B_{M_k\mathbf{w}}^{T}(B_{M_k\mathbf{w}}B_{M_k\mathbf{w}}^T)^{-1}
\end{eqnarray*} 
Indeed, the pseudo-inverse provides a least square solution of the Equation \ref{syst}. Further we know from Proposition \ref{permutation} that $N_k=\tilde{P}(\frac{1}{\sqrt{2}}H)^{\otimes k-1}$. Hence, $N_k^{-1}=({\sqrt{2}}H)^{\otimes k-1}\tilde{P}^T$. Clearly, from the discussions so far $\|\|_{op}$ be the operator norm. Then
\begin{eqnarray*}
\|W\|_{op}\leq \|({\sqrt{2}}H^{\otimes k-1})\|_{op}\|\tilde{P}^T\|_{op}\|B_{M_k\mathbf{w}}^{\dagger}\|_{op}\leq\sqrt{\nu}    
\end{eqnarray*}
 
Moreover, since, we are dealing with real polynomials, 
\begin{eqnarray*}
    \mathbf{c}&=& B_{M_k\mathbf{w}}N_k\mathbf{y}
\end{eqnarray*} where $\mathbf{y}=\cos_{M_k\mathbf{\Phi}_{\mathbf{w}}}$ leads to  
\begin{eqnarray}\label{WW}
   M_k\mathbf{\Phi_w}=\arccos_{W\mathbf{c}}
\end{eqnarray} 

The phases of WQSP is given by $ \mathbf{y} = M_k\mathbf{\Phi}_{\mathbf{w}}$ and we know from Lemma \ref{basic} that $M_k$ is of size $\nu\times k$ and typically \(k \ll \nu\) and $\mathbf{\Phi}_{\mathbf{w}}\in \mathbb{R}^k$ is the underlying phase vector. By definition, a phase vector is admissible (i.e. the system is consistent) if $\exists \mathbf{\Phi}_{\mathbf{w}}\text{ such that }  \mathbf{y} = M_k\mathbf{\Phi}_{\mathbf{w}}$. But the set of all vectors of the form \(M_k\mathbf{\Phi}_{\mathbf{w}}\) is precisely the column space of \(M_k\). Hence the admissible phase vectors are exactly lie $\mathrm{col}(M_k).$ Thus, given a vector $\mathbf{y}\in\mathbb R^\nu$, the reconstruction problem consists of finding the admissible vector closest to \(\mathbf{y}\) in the Euclidean norm. This is done through solving the least--squares problem \small
\begin{align}\label{Lsprob}
\min_{\mathbf{\Phi}_{\mathbf{w}}\in\mathbb R^k}
\|\mathbf{y}-M_k\mathbf{\Phi}_{\mathbf{w}}\|_2^2=\min_{\mathbf{\Phi}_{\mathbf{w}}\in\mathbb R^k}
\|\arccos_{\mathbf{Wc}}-M_k\mathbf{\Phi}_{\mathbf{w}}\|_2^2.    
\end{align}\normalsize
where the unique minimizer is
\begin{eqnarray*}
\mathbf{\Phi}_{\mathbf{w}}^\mathrm{LS}
=
(M_k^TM_k)^{-1}M_k^T\mathbf{y}=M_K^\dagger\mathbf{y},    
\end{eqnarray*}

and the corresponding reconstructed vector is
\begin{eqnarray}\label{haty}
{\mathbf{y}}^\mathrm{LS}
=
M_k\mathbf{\Phi}_{\mathbf{w}}^\mathrm{LS}
=
P_{M_k}\mathbf{y},    
\end{eqnarray}

where $P_{M_k}=M_k(M_k^TM_k)^{-1}M_k^T$ is the orthogonal projector onto $\mathrm{col}(M_k).$ This leads to the corresponding polynomial $P_\mathrm{LS}(x)$. If $\mathbf{y}\in\operatorname{col}(M_k)$, then the system becomes consistent and there exists \(\mathbf{\Phi}_{\mathbf{w}}\) such that $\mathbf{y}=M_k\mathbf{\Phi_w}$. Hence, the least--squares residual vanishes and we get $P_{M_k}\mathbf{y}=\mathbf{y}$. Consequently, $W\mathbf{c}=\cos_{(M_k\mathbf{\Phi}_{\mathbf{w}})}$ exactly, and therefore the error is zero. This leads to the following theorem.

\begin{theorem}\label{bounds}
   Let WLOG, $P(x)\in \mathbb{R}[x], P(x)=c_dT_d(x)+c_{d-2}T_{d_2}(x)+\hdots+c_0T_0, |P(x)|\leq 1 $ be an even parity polynomial which is also a linear sum of Chebyshev polynomials with coefficients $\mathbf{c}$. Let $\mathbf{w}\in\mathbb{N}^k$ be a suitable weight vector such that $\mathrm{nz}(\mathbf{c})\subseteq\mathcal{R}_{\mathbf{w}}$ and $\nu=2^{k-1}$. Further, let $\mathbf{\Phi}_{\mathbf{w}}^\mathrm{LS}$ be the least square solution to the problem defined in Equation \ref{Lsprob} such that ${\mathbf{y}}^\mathrm{LS}=M_k\mathbf{\Phi}_{\mathbf{w}}^\mathrm{LS}$ and ${\mathbf{c}}^\mathrm{LS}=B_{M_k\mathbf{w}}N_k{\mathbf{y}^\mathrm{LS}}$ and denote $P_{\mathrm{LS}}(x)={c}^\mathrm{LS}_dT_d+{c}^\mathrm{LS}_{d-2}T_{d-2}+\hdots +{c}^\mathrm{LS}_0T_0$ to be the constructed polynomial through the Least-Squares problem in Algorithm \ref{algo1}. If $D_{PP_{\mathrm{LS}}}$ be the distance between $P$ and $P_{\mathrm{LS}}$, then \begin{eqnarray*}
       D_{PP_{\mathrm{LS}}}^2\le2\pi\|B_{M_k\mathbf{w}}\|_{op}^2(\|c\|_2^2+ C(\rho)^2).
   \end{eqnarray*} 
   where, $C(\rho):=
\sup_{|x|\le\rho}|r(x)|$ such that $r(x):=\arccos(x)-\frac{\pi}{2}+x,$ and $\rho:=\|W\mathbf{c}\|_\infty$ such that $W$ is defined in Equation \ref{W}. 
\end{theorem}

\pf From standard theory of Chebyshev polynomials, \begin{eqnarray}\label{distanc}
    D_{PP_{\mathrm{LS}}}^2&&=\|P(x)-P_{\mathrm{LS}}(x)\|^2_{L_{\omega}^2}\\\nonumber&&=\int_{-1}^1\frac{1}{\sqrt{1-x^2}}(P(x)-P_{\mathrm{LS}}(x))^2dx\\\nonumber
    &&={\pi}(c_0-{c}^\mathrm{LS}_0)^2+\frac{\pi}{2}\sum_{j\geq1}(c_j-{c}^\mathrm{LS}_j)^2\\\nonumber&&=\|\mathbf{c}-{\mathbf{c}}^\mathrm{LS}\|_{L_{\omega}^2}^2
\end{eqnarray} where $\|\|_{L_{\omega}^2}$ is the weighted $L^2$ norm\cite{kreyszig1978introductory}. Clearly, $\|\mathbf{c}-{\mathbf{c}}^\mathrm{LS}\|_{L_{\omega}^2}^2\leq \pi\|\mathbf{c}-{\mathbf{c}}^\mathrm{LS}\|_2^2$  and thus \begin{eqnarray}
    D_{PP_{\mathrm{LS}}}^2\leq \pi\|\mathbf{c}-{\mathbf{c}}^\mathrm{LS}\|_2^2.
\end{eqnarray}. Now, from Equation \ref{syst},

\begin{eqnarray*}
    \mathbf{c}-{\mathbf{c}^\mathrm{LS}}=B_{M_k\mathbf{w}}N_k(\cos_{\mathbf{y}}-\cos_{{\mathbf{y}}^\mathrm{LS}})
\end{eqnarray*} Now, 
\begin{eqnarray*}\hspace{-0.5cm}
 \|B_{M_k\mathbf{w}}N_k\cos_{\mathbf{y}}\|^2_2&&=\|B_{M_k\mathbf{w}}\Tilde{P}(\frac{1}{\sqrt{2}}H)^{\otimes k-1}\cos_{\mathbf{y}}\|^2_2\\&& \|B_{M_k\mathbf{w}}\|_{op}^2\|\Tilde{P}\|_{op}^2\|(\frac{1}{\sqrt{2}}H)^{\otimes k-1}\cos_{\mathbf{y}}\|^2_2\\&&\leq\frac{1}{\nu} \|B_{M_k\mathbf{w}}\|_{op}^2 \|\cos_{\mathbf{y}}\|_2^2 
\end{eqnarray*} where $\Tilde{P}$ is a permutation defined in Proposition \ref{permutation}.
Thus, 
\begin{eqnarray*}
    D_{PP_{\mathrm{LS}}}^2\leq \pi\|\mathbf{c}-{\mathbf{c}}^\mathrm{LS}\|_2^2\leq \frac{\pi}{\nu}\|B_{M_k\mathbf{w}}\|_{op}^2\|\cos_\mathbf{y}-\cos_{{\mathbf{y}}^\mathrm{LS}}\|_2^2.
\end{eqnarray*}. Now cosine is Lipschitz-$1$ continuous and therefore,
\begin{eqnarray*}
  \|\cos_\mathbf{y}-\cos_{{\mathbf{y}}^\mathrm{LS}}\|_2^2\leq  \|\mathbf{y}-{{\mathbf{y}}^\mathrm{LS}}\|_2^2. 
\end{eqnarray*}
Now, from definition of ${\mathbf{y}}^\mathrm{LS}$ in Equation \ref{haty}, we have 
\begin{eqnarray}\label{Dsquare}
 &&\mathbf{y}-{{\mathbf{y}}^\mathrm{LS}}=(I-P_{M_k\mathbf{w}})\mathbf{y}\\\nonumber
 &&\implies  D_{PP_{\mathrm{LS}}}^2\leq \frac{\pi}{\nu}\|B_{M_k\mathbf{w}}\|_{op}^2\|(I-P_{M_k})\mathbf{y}\|_2^2
\end{eqnarray} 
Further, from Equation \ref{WW},\begin{eqnarray*}
    \mathbf{y}=\arccos_{W\mathbf{c}}=\frac{\pi}{2}\mathbf{1}-W\mathbf{c}+r(W\mathbf{c})
\end{eqnarray*} where $W=N_k^{-1}B_{M_k\mathbf{w}}^\dagger$ as defined in Equation \ref{W} and $\mathbf{1}$ is the all one vector. From the construction of $M_k$, clearly, $\mathbf{1}\in \mathrm{col}(M_k)$. Thus, $(I-P_{M_k\mathbf{w}})\mathbf{1}=0$ as $P_{M_k}=M_k(M_k^TM_k)^{-1}M_k^T$ is the orthogonal projector onto the column space $\mathrm{col}(M_k).$ Furthemore, Since \(P_M\) is an orthogonal projection, \(I-P_{M_k}\) is also an orthogonal projection (onto \(\operatorname{col}(M_k)^\perp\)). Therefore,$\|I-P_{M_k}\|_{op}=1$. Hence,
\begin{eqnarray*}
    (I-P_{M_k})\mathbf{y}=-(1-P_{M_k})W\mathbf{c}+ (I-P_{M_k})r(W\mathbf{c})\end{eqnarray*} Taking norms and using the fact that $\|I-P_{M_k}\|_{op}=1$, we get 
\begin{eqnarray*}
    \|(I-P_{M_k})\mathbf{y}\|_2=\|W\mathbf{c}\|_2+ \|r(W\mathbf{c})\|_2\end{eqnarray*} Squaring both sides and using the fact that $\|a+b\|^2_2\leq 2\|a\|_2^2+2\|b\|_2^2$, we get,

\begin{eqnarray*}
    \|(I-P_{M_k})\mathbf{y}\|^2_2=2\|W\mathbf{c}\|^2_2+ 2\|r(W\mathbf{c})\|^2_2\end{eqnarray*} Now, from the definition of $W$ in Equation \ref{W} and Lemma \ref{bklemma}, we see $\|W\mathbf{c}\|^2_2\leq\nu\|\mathbf{c} \|_2^2$

Also, a simple check shows that $|r((Wc)_j)|\le C(\rho)\forall j$, and thus we obtain
\begin{eqnarray*}
\|r(Wc)\|_2^2=\sum_{i=1}^n |r((Wc)_i)|^2\le\nu C(\rho)^2.  
\end{eqnarray*}

Hence\begin{eqnarray*}
\|(I-P_M)y\|_2^2
\le
2\nu\|c\|_2^2
+
2\nu C(\rho)^2.    
\end{eqnarray*} Plugging this in Equation \ref{Dsquare},
the result follows immediately. $\hfill\square$

\begin{remark}
\begin{enumerate}
\item Theorem \ref{bounds} give the exact same result for odd parity polynomials where $T_j$ starts from $T_1$ instead of $T_0$. Further, for a complex polynomial $P(x)\in\mathbb{C}[x]$, we write $P(x)=A(x)+iB(x)$ where $A(x),B(x)\in \mathbb{R}[x]$. In such cases, we apply Theorem \ref{bounds} to $A$ and $B$ and using triangle inequalty, we get $D^2_{PP_{\mathrm{LS}}}\le D^2_{AA_{\mathrm{LS}}}+D^2_{BB_{\mathrm{LS}}}$.
    \item Assume that $\|\mathbf{c}\|_\infty = \nu^{-1}=\frac{1}{2^{k-1}}$. $\|W\mathbf{c}\|_{\infty}\leq \|W\mathbf{c}\|_{2}\leq \|W\|_{op}\|\mathbf{c}\|_2\leq \sqrt{\nu}\|W\|_{op}\|\mathbf{c}\|_{\infty}=1$. Recall that 
$C(\rho) = \sup_{|x|\le \rho} \left|\arccos(x) - \frac{\pi}{2} + x\right|$ where $r(x) = \arccos(x) - \frac{\pi}{2} + x$ and $\rho:=\|W\mathbf{c}\|_{\infty}$. Since  
$r'(x) = 1 - \frac{1}{\sqrt{1-x^2}} \le 0$,  
the extrema on $[-1,1]$ occur at the endpoints. Thus  
$r(1) = 1 - \frac{\pi}{2}$ and $r(-1) = \frac{\pi}{2} - 1$, so $C(1) = \frac{\pi}{2} - 1$. Therefore  
$C(\|W \mathbf{c}\|_\infty) \le \frac{\pi}{2} - 1$. Thus \begin{eqnarray*}
    D_{PP_{\mathrm{LS}}}^2 \le 2\pi \|B\|_{op}^2 \left(\frac{1}{n} + \left(\frac{\pi}{2}-1\right)^2\right). 
\end{eqnarray*} Assuming, all $\mathbf{w}$ components are unique implies $B_{M_k}=I$ and $D_{PP_{\mathrm{LS}}}\sim O(1)$.  
\item However, if we assume $B_{M_k}=I$ and $\|\mathbf{c}\|_\infty = \nu^{-1-m}$, for some $m>0$. Then , so $\|W\mathbf{c}\|_\infty \le \nu^{-m}$.

Using the expansion
$r(x)=\arccos(x)-\frac{\pi}{2}+x = -\frac{x^3}{6}+O(x^5)$,
we have for small $\rho$ that $C(\rho)=O(\rho^3)$, hence
$C(\|Wc\|_\infty)=O(\nu^{-3m})$ and $C(\|Wc\|_\infty)^2=O(\nu^{-6m})$.

Also,
$\|\mathbf{c}\|_2^2 \le \nu\|\mathbf{c}\|_\infty^2 \le \nu\cdot \nu^{-2-2m}=\nu^{-1-2m}$.

Therefore,
\[
D^2 \le 2\pi(\left(\nu^{-1-2m} + O(\nu^{-6m})\right).
\] Hence, for any coefficient vector $|\mathbf{c}|$ whose components are bounded above by $2^{1-k-m}$ for some $m\geq 0$, and for distinct weight components in $\mathbf{w}\in\mathbb{N}^k$, WQSP provides a sufficiently accurate approximation of the target polynomial.

This observation becomes particularly useful for high-degree polynomials expressed as linear combinations of a large number of Chebyshev polynomials, where finding an appropriate weight-vector partition can become computationally expensive. In such cases, one may choose the weight vector $(1,2,3,\hdots)^T$ for even-parity polynomials. For odd-parity polynomials of large degree, one may instead choose the weights $(1,1,2,3,\hdots)^T$. This construction yields a quadratic reduction in the number of parameters while maintaining relatively small approximation errors.

\item When $k=2$, then $M_k=\bmatrix{1&1\\1&-1}$ i.e. $\mathrm{col}(M_2)=\mathbb{R}^2$ i.e. the column space of $M_2$ spans entire $\mathbb{R}^2$. Thus for any $\mathbf{y}$, $P_{M_k}\mathbf{y}=y$. Hence, $(I-P_{M_2})\mathbf{y}=\mathbf{0}$ and thus $D_{PP_{\mathrm{LS}}}=0$. This is very significant and show cases a primary advantage of WQSP to QSP. We shall demonstrate with an example. \\
\noindent\textbf{Example:}
Let $P(x)\in\mathbb{R}[x]$ be a fixed parity polynomial such that $P(x)=c_{1}T_{n_1+n_2}(x)+c_{2}T_{n_1-n_2}(x)$ for a large $n=n_1+n_2$. In the regime of QSP, as seen from \cite{gilyen}, in order to calculate the phases, an algorithm of $O(4n^2)$ is required. This is because finding out angles of QSP works in normal basis unlike Chebyshev. However, for WQSP, we choose the weight $\mathbf{w}=[n_1,n_2]^T$ because the $\mathcal{R}_{\mathbf{w}}=\{n_1+n_2,n_1-n_2\}$. In such case, $N_2=\frac{1}{2}\bmatrix{1&1\\1&-1}$ $M_2\mathbf{\Phi}_{\mathbf{w}}=\bmatrix{\phi_1+\phi_2\\\phi_1-\phi_2}$. Thus, and we solve for \begin{eqnarray*}
  N_2\bmatrix{\cos{(\phi_1+\phi_2)}\\\cos{(\phi_1-\phi_2)}}=\bmatrix{c_1\\c_2}  
\end{eqnarray*}. This generates the solution $\phi_1=\arccos{(c_1+c_2)}+\arccos{(c_1-c_2)}$ and $\phi_2=\arccos{(c_1+c_2)}-\arccos{(c_1-c_2)}$ and thus WQSP produces an exact solution. The conversion from a standard polynomial to a Chebyshev basis can be done using Fast Chebyshev Transformation \cite{cheb1,cheb2,cheb3} with complexity $O(n\log n)$ and thus WQSP proves advantageous in such cases. If the polynomial presents itself in a Chebyshev basis from the start, this takes $O(1)$ operations.
\item It is of note that in Theorem \ref{bounds}, when we consider the QSP framework i.e when $\mathbf{w}=\mathbf{1}$,  the system of equations described in Equation \ref{WW}, after obtaining from Equation \ref{syst} becomes consistent. This means for a polynomial fitted through QSP regime, $\exists \mathbf{\Phi}_{\mathbf{w}}$ such that \begin{eqnarray*}
    B_{M_k\mathbf{1}}N_k\cos_{M_k\mathbf{\Phi}_{\mathbf{w}}}=\mathbf{c}
\end{eqnarray*}
In such cases, $D_{PP_{\mathrm{LS}}}=0$ since the least square turns zero due to the system being consistent. This is due to the fact that the elements of $\mathbf{\Phi}_{\mathbf{w}}$ can be calculated from the leading terms of $P(x)$ repeatedly using the method mentioned in \cite{gilyen}.
\item Since WQSP's efficacy depends on the choice of weights, a good choice is always to choose a partition or the weight vector $\mathbf{w}$, such that the resultant system  \begin{eqnarray*}
    B_{M_k\mathbf{1}}N_k\cos_{M_k\mathbf{\Phi}_{\mathbf{w}}}=\mathbf{c}
\end{eqnarray*} resembles a QSP system in order to get a perfect fit. In other words, by choosing appropriate weights, one can map a WQSP problem of large degree polynomials into a QSP setting of smaller degree polynomials. We illustrate this with an example. 
\noindent\textbf{Example:} Let us take a real polynomial of degree $8$ such that $P(x)=c_0T_1+c_4T_4+c_8T_8$. Then If we consider the weight vector to be $(2,2,2,2)$ Then
\begin{eqnarray*}
    \underbrace{B_{M_k\mathbf{w}}N_{k}}_{A_{WQSP}}=\frac{1}{8}\bmatrix{1&0&0&0&0&0&0&0\\0&1&1&0&1&0&1&0\\0&0&0&1&0&1&0&1}N_k.
\end{eqnarray*}A simple matrix calculation shows that $B_{M_k\mathbf{w}}N_{k}$ derived for this problem (denoted as $A_{WQSP}$) is permutationally similar to the matrix $B_{M_k\mathbf{w}}N_{k}$ derived for the problem of fitting $Q(x)=c_0T_0+c_2T_2+c_4T_4$ through QSP (denote it as $A_{QSP}$) i.e. $A_{WQSP}=P_1A_{QSP}P_2$ for some permutations $P_1$ and $P_2$. Thus for $\mathbf{w}=(2,2,2,2)$, $P(x)$ can be exactly fitted using WQSP. In Figure \ref{fig:8poly}, we demonstrate exactly this fact by choosing $c_0,c_4,c_8$ from uniform distribution $[-1/8,1/8]$ since $|P(x)|\leq 1$ must be satisfied. This shows that QSP has a lot of redundant parameters. Further, for B-splines as well, the total number of parameters=$G+d$ where $G$ number of internal grid intervals/knot vectors. Hence, WQSP can exactly fit polynomials with lesser number of parameters compared to QSP and B-splines. This results in design of shallower quantum circuits with lower depth and gate complexity which is useful for modern era quantum systems. 
\item It has also been observed that for other choices of $\mathbf{w}$ such that $M_k\mathbf{w}$ contains more than one repeated element (excluding the first component), an exact solution is obtained in most cases. To study this phenomenon, we conducted experiments on polynomials of degree $d$ ranging from $3$ to $15$. For each weight partition of each degree, approximately $500$ samples were generated, where the coefficients were drawn independently from the uniform distribution $U[-1/d,1/d]$.

A possible explanation for this behavior lies in the effective degrees of freedom induced by the non-zero coefficients in the Chebyshev basis expansion. For example, the polynomial $P(x)=c_1T_1+c_3T_3+c_5T_5+c_7T_7$ contains four independent coefficients and therefore requires at least four effective degrees of freedom in the $R_z$ parameters. The weight partition $(3,2,1,1)$ provides exactly this degree of freedom, resulting in zero error (up to floating-point precision). A similar phenomenon is observed for partitions such as $(2,2,2,1)$ and $(3,1,1,1,1)$, since these also satisfy the minimum degree-of-freedom threshold. In Figure~\ref{fig:7poly}, we present an instance of the same phenomenon for a randomly generated polynomial of degree $7$ with different weight partitions.  

In contrast, partitions such as $(4,2,1)$ exhibit relatively larger approximation errors of order $O(10^{-2})$--$O(10^{-1})$, since the available degrees of freedom are insufficient. Nevertheless, the bound established in Theorem~\ref{bounds} remains satisfied in all such cases.

Another perspective on this phenomenon is through a factorized representation of the polynomial. Suppose the polynomial $P(x)$ admits a decomposition of the form
$P(x)=\prod_{j,l}(c_jT_j+d_lT_l)$ for suitable $j,l\in\mathbb{N}$ with same parity. In such cases, exact solutions can be obtained for each factor $(c_jT_j+d_lT_l)$ individually, following the arguments developed above, and the construction can then be extended recursively across all such factors.

For example, standard QSP implicitly exploits the fact that parity-constrained complex polynomials admit recursive constructions built from lower-order Chebyshev components such as $(c_jT_2+d_lT_0)$ or $(c_jT_3+d_lT_1)$. Since exact solutions can be generated for each pair of coefficients, exact realization of the overall polynomial follows.

However, obtaining such factorizations in practice is generally difficult, making it computationally expensive to identify weight partitions other than those arising from standard QSP that also admit exact polynomial realization.

For example, consider a polynomial of the form
$P(x)=(c_{50}T_{50}+c_{32}T_{32})(c_{12}T_{12}+c_2T_2)$.
In this case, the weight vector $(41,9,7,5)$ yields an exact realization. By contrast, implementing the same polynomial using standard QSP would require deeper circuits, more parameters, and consequently longer execution times.

\end{enumerate}
\end{remark}

\begin{figure*}[ht!]
    \centering
 \subfloat[ \label{2222} Error from WQSP based approximation with pruned circuit and weight vector $(2,2,2,2)$]{\includegraphics[width=\columnwidth]{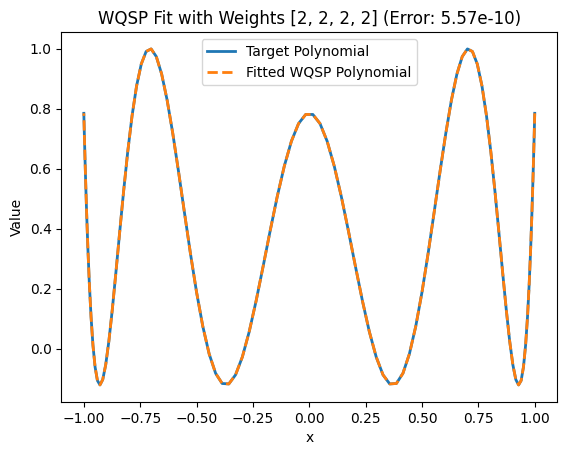}}
 \subfloat[ \label{11111111} Error from QSP based approximation]{\includegraphics[width=\columnwidth]{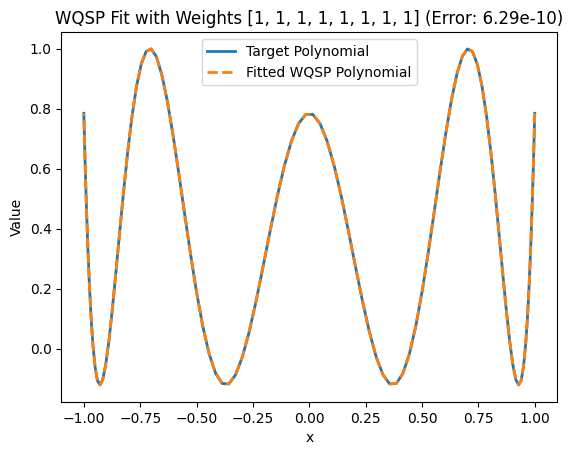}}
    \caption{WQSP fitting of a real $8$ degree target polynomial $P(x)=c_1T_0+c_2T_4+c_3T_8$ with weight partition $(2,2,2,2)$ and whose coefficients are randomly generated. The WQSP fitting in \ref{2222} is compared with the standard QSP fitting as depicted \ref{11111111}.}
    \label{fig:8poly}
\end{figure*}

As stated previously, for very high-degree polynomials, finding a suitable weight partition becomes computationally expensive. Moreover, standard QSP itself becomes difficult to implement due to the increased circuit depth and the $O(d^2)$ complexity required to recover the phase vector $\mathbf{\Phi}$. In such cases, we consider $\mathbf{w}=(1,2,3,\hdots)\in\mathbb{R}^k$ for even-parity polynomials and $\mathbf{w}=(1,1,2,3,\hdots)\in\mathbb{R}^k$ for odd-parity polynomials. Since $\sum_j w_j=d$, it follows that $k\approx\Theta(\sqrt{d})$. From this perspective, WQSP introduces a weight vector as a preprocessing step in the classical encoding stage. Alternatively, if the weights are chosen as $w_j=2^j$ for all $j\in{1,\hdots,k}$ and the degree of $P(x)$ is $d$, then
$d\leq\sum_{j=1}^{k}2^j$,
which implies $k=\Theta(\log_2 d)$. In both cases, the polynomial is converted into a linear combination of Chebyshev polynomials. Further, since $d$ is large, the coefficients satisfy $|\mathbf{c}|{\infty}\leq\frac{1}{2^{k-1}}$, and therefore, by Theorem~\ref{bounds}, WQSP achieves quadratic to exponential reductions in the number of parameters and consequently in circuit depth while incurring only bounded approximation error. For smaller values of $\|\mathbf{c}\|_{\infty}$, the approximation error converges rapidly to $0$, making WQSP highly efficient. Table~\ref{tab:qsp_wqsp_univariate} illustrates these trade-offs. By selecting WQSP weights that reduce the number of trainable parameters quadratically to exponentially relative to QSP, WQSP achieves comparable approximation errors  while requiring significantly fewer parameters, in contrast to QSP, which attains machine-precision accuracy at a substantially higher parameter cost. \begin{table*}[t]
\centering
\caption{Univariate polynomial approximation: Comparison of QSP and WQSP across different polynomial degrees. QSP requires $d+1$ phase parameters per parity component, resulting in $2d+1$ total phases to approximate the even and odd components using separate circuits. In contrast, WQSP uses only $O(2\lceil \log_2(d)) \rceil $ to $O(2\lceil \sqrt{d} \rceil)$ trainable parameters depending on the weights. For each polynomial degree $d$, the results are averaged over $40$ randomly generated polynomials with coefficients sampled independently from the uniform distribution $U[-1/d,\,1/d]$.}
\label{tab:qsp_wqsp_univariate}
\begin{tabular*}{\textwidth}{@{\extracolsep{\fill}} l r r r r l @{}}
\hline
\textbf{Degree ($d$)} & {$\mathbf{RMSE}_{\mathrm{QSP}}$} & {$\mathbf{RMSE}_{\mathrm{WQSP}}$} & {$\mathbf{Parameters}_{\mathrm{QSP}}$} & $\mathbf{Parameters}_{\mathrm{WQSP}}$&\textbf{Partition} \\
\hline
\hline
5  & $3.14\times 10^{-3}$ & $8.97 \times 10^{-3}$ & 11 & {6}& $[2,2,1]_{\mathrm{Odd}};[2,1,1]_{\mathrm{Even}}$  \\
\hline
10 & $1.14 \times 10^{-4}$  & $2.91 \times 10^{-3}$ & 21 & 9& $[4,2,1,1,1]_{\mathrm{Odd}};[4,2,2,2]_{\mathrm{Even}}$ \\
\hline
15 & $1.51 \times 10^{-5}$  & $7.59 \times 10^{-3}$ & 31 & {9}& $[6,4,2,2,1]_{\mathrm{Odd}};[6,4,2,2]_{\mathrm{Even}}$  \\
\hline
20 & $2.16 \times 10^{-3}$  & $7.03 \times 10^{-2}$ & 41 & {11}& $[7,4,4,2,2,1]_{\mathrm{Odd}};[7,4,4,2,2]_{\mathrm{Even}}$ \\
\hline
25 & $6.82 \times 10^{-4}$  & $6.61 \times 10^{-2}$ & 51 & 16 & $[4,4,4,3,3,3,2,2]_{\mathrm{Odd}};[4,4,4,3,3,3,2,1]_{\mathrm{Even}}$ \\
\hline
30 & $1.65 \times 10^{-4}$  & $5.65 \times 10^{-2}$ & 61 & 16& $[6,6,4,4,4,2,2,2]_{\mathrm{Odd}};[6,6,4,4,4,2,2,1]_{\mathrm{Even}}$ \\
\hline
35 & $1.82 \times 10^{-3}$  & $6.09 \times 10^{-2}$ & 71 & 17& $[6,5,5,4,4,4,3,3,1]_{\mathrm{Odd}};[6,5,5,4,4,4,3,3]_{\mathrm{Even}}$ \\
\hline
\end{tabular*}
\end{table*}

The deterministic method described so far provides deeper theoretical insight into QSP and WQSP. However, it becomes computationally infeasible for large systems, as it requires solving two linear systems. Hence, we present a learning-based approach that remains computationally tractable while satisfying the guarantees established in Theorem~\ref{bounds}.

\begin{figure*}[ht!]
    \centering
 \subfloat[ \label{3211} ]{\includegraphics[width=\columnwidth]{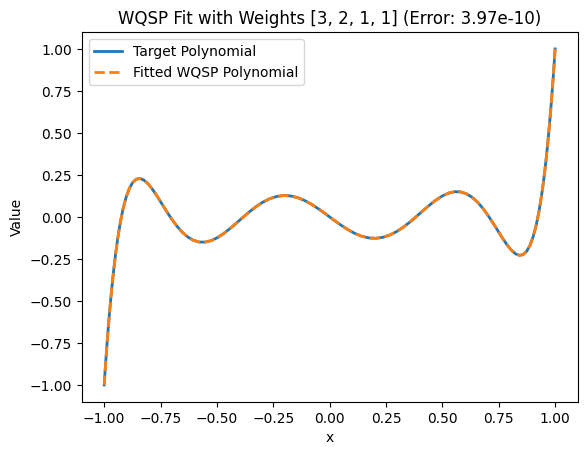}}
 \subfloat[ \label{22111} ]{\includegraphics[width=\columnwidth]{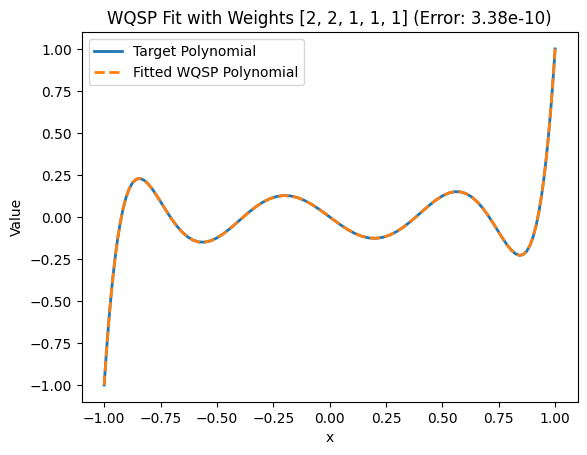}}\\
 \subfloat[ \label{1111111} ]{\includegraphics[width=\columnwidth]{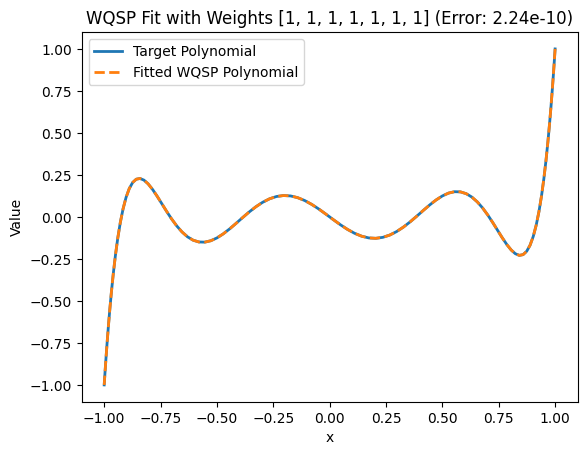}}
 \subfloat[ \label{421} ]{\includegraphics[width=\columnwidth]{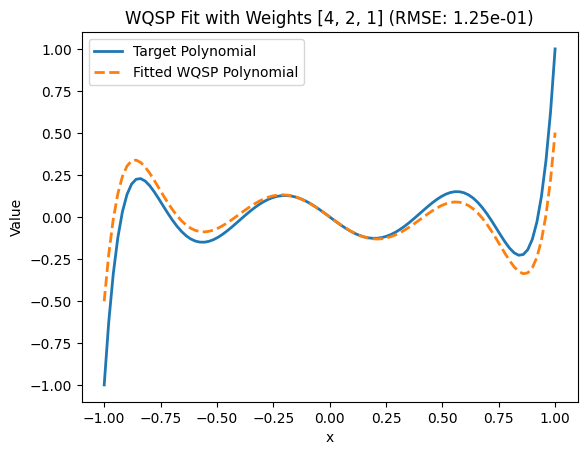}}
   \caption{WQSP fitting of a real $7$ degree target polynomial (fixed odd parity) with different weight partitions whose coefficients are randomly generated. In \ref{3211} and \ref{22111} the weight partitions are respectively $(3,2,1,1)$ and $(2,2,1,1,1)$ whose weight vector lengths are greater than or equal to the coefficient vector length and thus the necessary degrees of freedom are achieved. For \ref{1111111}, we have the QSP case which produces an error of a similar order compared to \ref{3211} and \ref{22111}. In \ref{421}, we see that the weight vector length is less than the coefficient vector length and thus such a setup doesn't cross the degrees of freedom which incurs a higher error.}\label{fig:7poly}
\end{figure*}
\subsubsection{WQSP in a quantum neural network framework}\label{nwqspqnn}
In WQSP, so far we have seen that one can solve two linear systems to get the solutions. In QSP as well, a $O(d^2)$ algorithm is required to construct the vector $\mathbf{\Phi}_{\mathbf{w}}$. Thus, most literature \cite{Soni2022Function,lin2025mathematical,linlinenergy} involving QSP, researchers while implementing the code take a heuristic/ optimization approach in searching for suitable phase vector $\Phi$ for the problem      
\begin{eqnarray}
    &&\mathrm{min}_{\mathbf{\Phi}_{\mathbf{w}}} \|P(x)-\mathrm{Re}(\bra{0}U^{\mathrm{QSP}}_{\Phi}(x)\ket{0})\|_{L_{2}}\\\nonumber&&=\mathrm{min}_{\mathbf{\Phi}_{\mathbf{w}}} \sqrt{\int_{-1}^1 (P(x)-\mathrm{Re}(\bra{0}U^{\mathrm{QSP}}_{\Phi}(x)\ket{0}))^2dx}
\end{eqnarray} i.e. calculating from Root Mean Square Error (RMSE). This alleviates the need to repeatedly perform extensive angle constructions using leading coefficients and matrix multiplications by instead formulating the determination of $\mathbf{\Phi}$ as an optimization problem. Techniques such as Nelder-Mead, Gradient Descent\cite{datta2010numerical},  Levenberg-Marquardt \cite{LM_method} etc., can then be employed to solve it. This establishes a framework in which the underlying algorithm is learned through optimization rather than derived through a purely deterministic procedure. Such an approach is particularly valuable for learning unknown polynomial representations for a given problem, while also enabling the fitting of predefined polynomials. Consequently, this framework becomes highly useful for generating activation functions, which are fundamental components of Kolmogorov-Arnold Networks\cite{liu2024kan}. We note that it is easy to check for $x\in[-1,1]$, $\|.\|_{L_2}^2\leq \|.\|_{L_{\omega}}$ (as defined in Theorem \ref{bounds}), so Theorem \ref{bounds} holds. This is also illustrated in Figure \ref{figmore}, where WQSP approximation error with in $O(10^{-2}/10^{-3})$ has been shown to be achievable with significantly less parameters compared to QSP.
\begin{figure*}[ht!]
    \centering
 \subfloat[ \label{10d} Degree 10  ]{\includegraphics[width=\columnwidth]{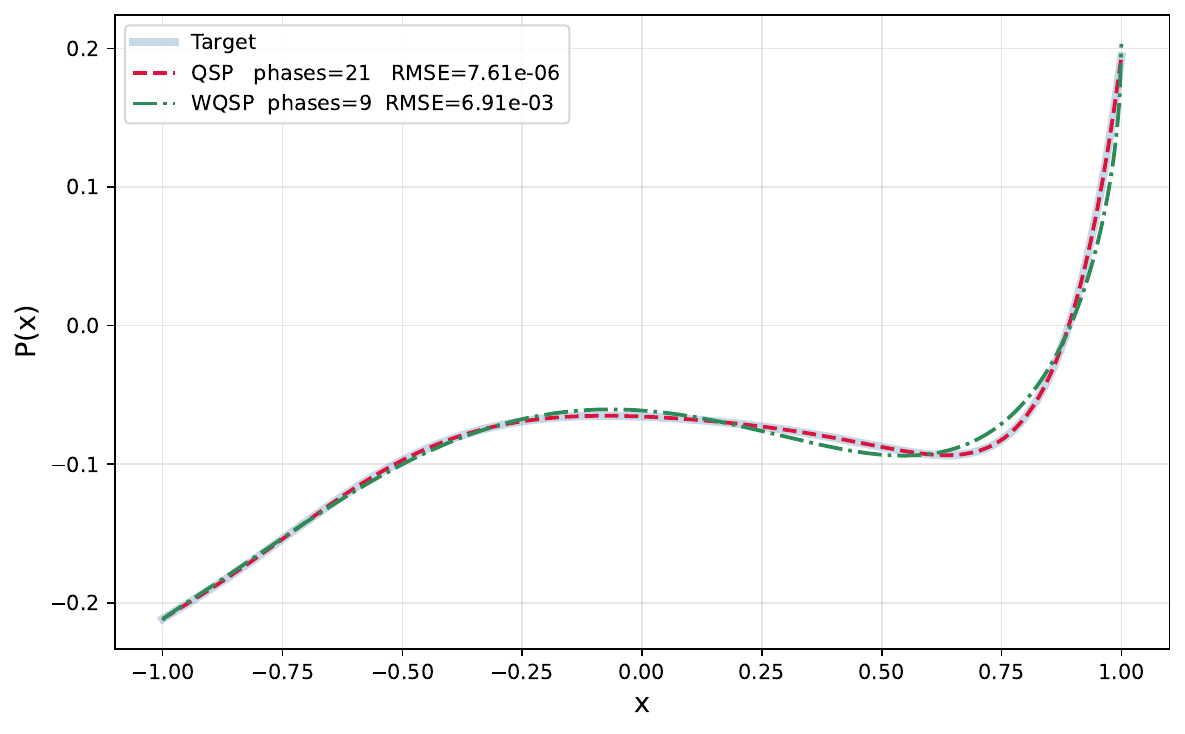}}
 \subfloat[ \label{15d} Degree 15]{\includegraphics[width=\columnwidth]{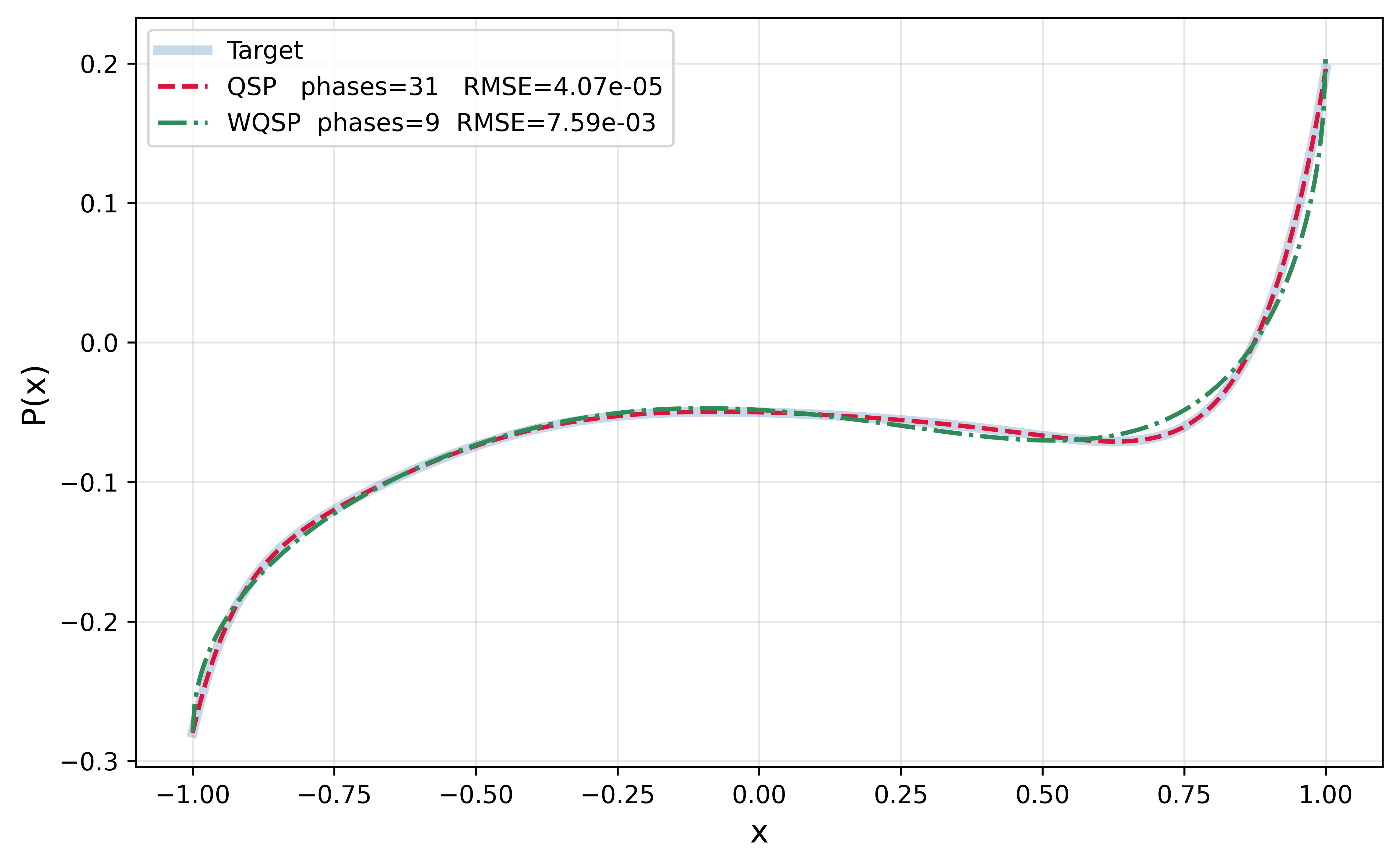}}\\
  \subfloat[ \label{35d} Degree 35 ]{\includegraphics[width=\columnwidth]{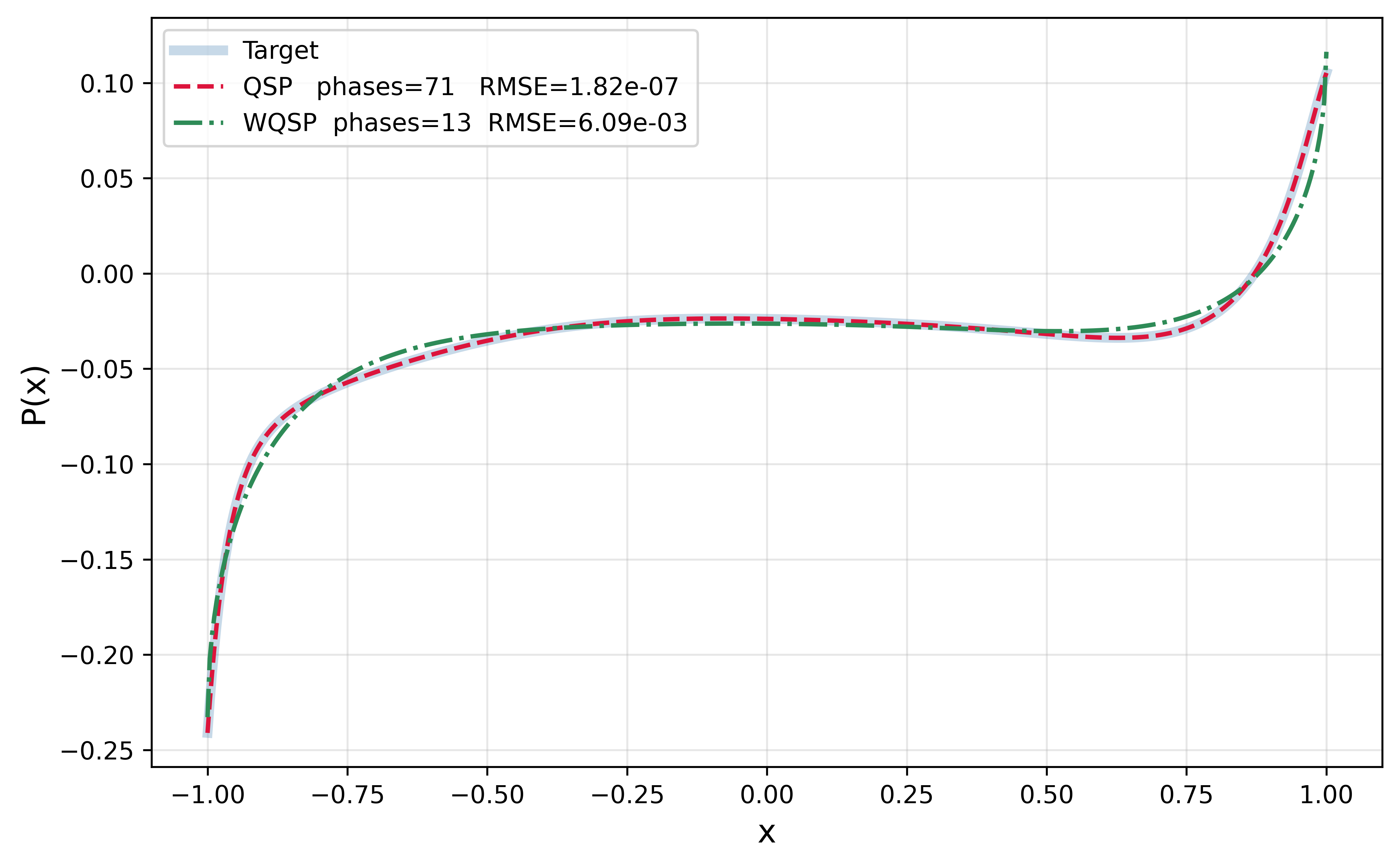}}
   \caption{WQSP fitting results for random real polynomials of degrees $10,15$ and $35$ whose Chebyshev coefficients were independently sampled respectively from the uniform distribution $U[-1/11,1/11],U[-1/16,1/16]$ and $U[-1/36,1/36]$. For both QSP and WQSP, the polynomial was first broken into odd and even components and the corresponding approximation errors were added up. The weight vectors for WQSP are chosen for degree $10: [2,2,2,4]$(even part),$[1,2,2,4]$ (odd part). For degree $15:[1,2,4,8]$(odd part), $[1,1,4,8]$ (even part). For degree $35: [1,2,2,2,4,8,16]$(odd part), $[2,2,2,4,8,16]$(even part).}
    \label{figmore}
\end{figure*}

We incorporate the same approach in the WQSP regime i.e. we find $\mathbf{\Phi}_{\mathbf{w}}$ for a suitable weight partition $\mathbf{w}$ such that we solve the following minimization problem.
\begin{eqnarray}
    \mathrm{min}_{\mathbf{\Phi}_{\mathbf{w}}} \|P(x)-\mathrm{Re}(\bra{0}U^{\mathrm{WQSP}}_{\mathbf{\Phi_w}}(x)\ket{0})\|_{L_{2}}
\end{eqnarray}

This constructs a quantum neural network framework where the parameter vector $\mathbf{\Phi}_\mathbf{w}$, whose components lie interleaved with $U(w_jx)$ or $R_x(w_j\theta)$ in a quantum circuit (see Equation \ref{WQSPcirc}) is optimized repeatedly each round until we achieve a suitable residue or error. In Figure 

\begin{align}\label{WQSPcirc}
    {\Qcircuit @C=1em @R=.7em {
 &\lstick{}&\gate{U(w_kx)} &\gate{R_z(\phi_k)}&\hdots&\hdots&\gate{U(w_1x)}&\gate{R_z(\phi_1)}&\qw\\}}
\end{align}

For complex polynomials, the problem is modified to \begin{eqnarray}\hspace{-0.65cm}
   && \mathrm{min}_{\mathbf{\Phi}_{\mathbf{w}}} \|P(x)-(\bra{0}U^{\mathrm{WQSP}}_{\mathbf{\Phi_w}}(x)\ket{0})\|_{L_{\omega}}=\\\nonumber&&\mathrm{min}_{\mathbf{\Phi}_{\mathbf{w}}}\sqrt{\int_{-1}^1 \frac{1}{\sqrt{1-x^2}}|P(x)-\bra{0}U^{\mathrm{WQSP}}_{\mathbf{\Phi}_{\mathbf{w}}}(x)\ket{0}|^2dx}
\end{eqnarray}.

In Figure \ref{fig:optimizescatter1410}, we employed the method using Levenberg--Marquardt optimization on the parameters for $14$-degree and $10$-degree polynomials whose coefficients were sampled uniformly from $[-1/8,1/8]$ and $[-1/6,1/6]$, respectively, in order to satisfy the boundedness property of the polynomials. We then applied the proposed framework, incorporating the Levenberg--Marquardt method, to approximate these polynomials through WQSP. For each optimization run, $20$ random restarts were performed from different initial points, and the run achieving the least error was selected. The resulting errors were then averaged over each partition. The figure illustrates that WQSP can efficiently approximate polynomials using significantly fewer parameters than standard QSP, indicating that QSP contains substantial parameter redundancy that can be pruned without noticeably affecting the approximation performance of the resulting circuit.
\begin{figure*}[ht!]
    \centering
 \subfloat[ \label{deg14} ]{\includegraphics[width=\columnwidth]{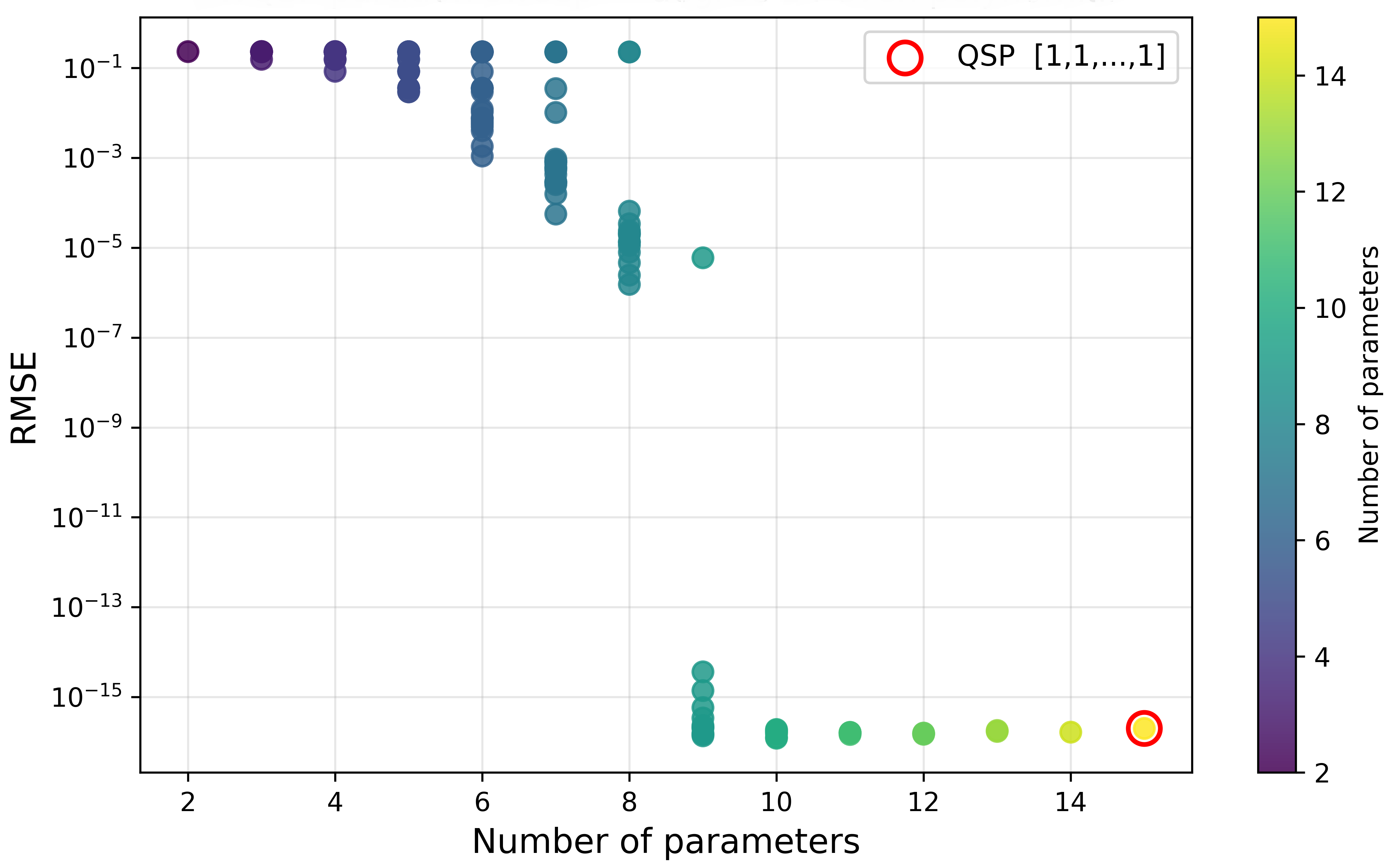}}
 \subfloat[ \label{deg10} ]{\includegraphics[width=\columnwidth]{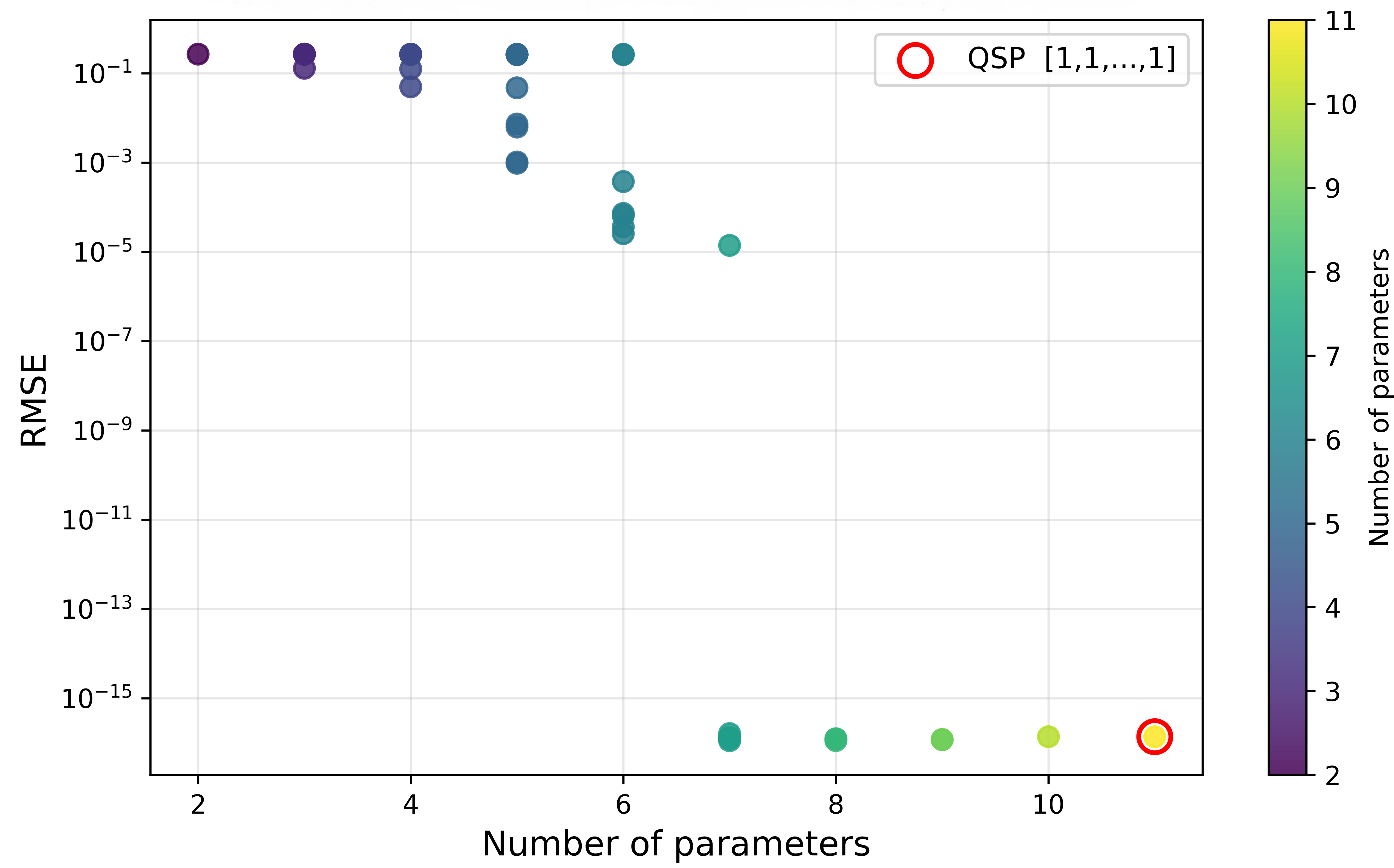}}
   \caption{WQSP fitting error across all partitions for a real degree-$14$ (\ref{deg14}) and degree-$10$ (\ref{deg10}) polynomial. Polynomial coefficients were generated by uniformly sampling from $[-1/8,1/8]$ and $[-1/6,1/6]$, respectively, and approximated using WQSP. Phase parameters were optimized using the Levenberg--Marquardt method with $20$ random restarts, and the solution achieving the lowest approximation error was reported.}\label{fig:optimizescatter1410}
\end{figure*}

We validate the WQSP framework through numerical simulations by training WQSP circuits to approximate univariate polynomial target functions. To evaluate approximation performance across a broad function class, we generate $100$ random target polynomials and train each circuit using $100$ equally spaced samples over $x\in[-1,1]$. For degree-$10$ polynomials, the results in Figure~\ref{tabledeg10} show the approximation error obtained under different weight partitions. We observe that partitions whose lengths exceed the number of nonzero Chebyshev coefficients generally achieve lower approximation error, supporting the theoretical intuition that increased effective degrees of freedom improve approximation quality.

\begin{figure*}[t!]
    \centering
 \subfloat[Weight{=$[1,1,\dots,1]$}, Mean Error= $6.89\times 10^{-5}$ ]{\includegraphics[width=1\columnwidth]{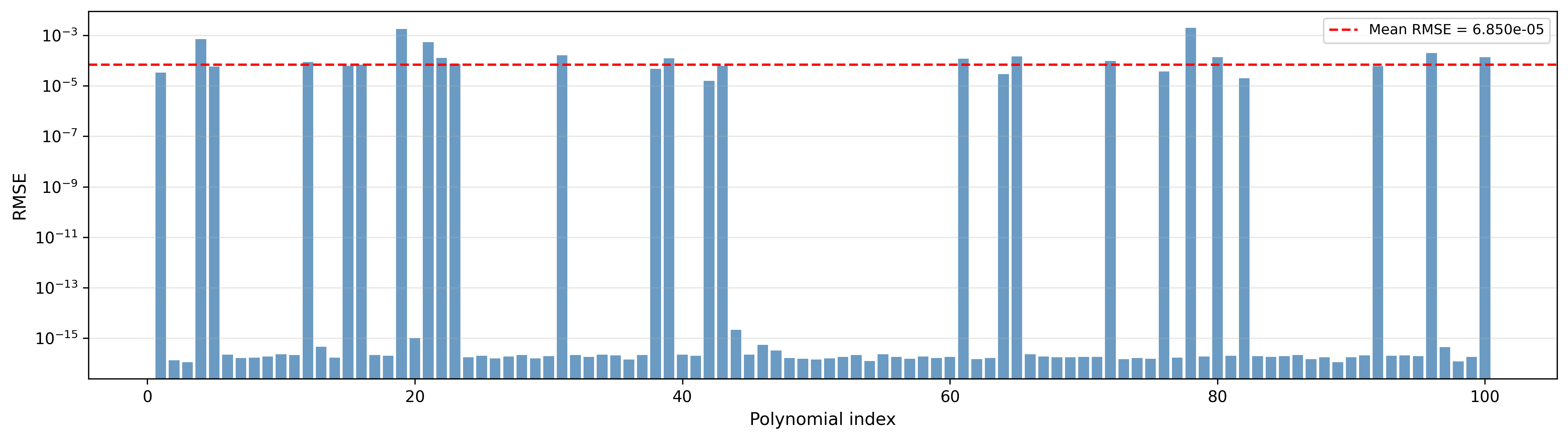}} \hfill
 \subfloat[Weight{=$[2,1,\dots,1]$},  Mean Error= $1.14\times 10^{-4}$]{\includegraphics[width=1\columnwidth]{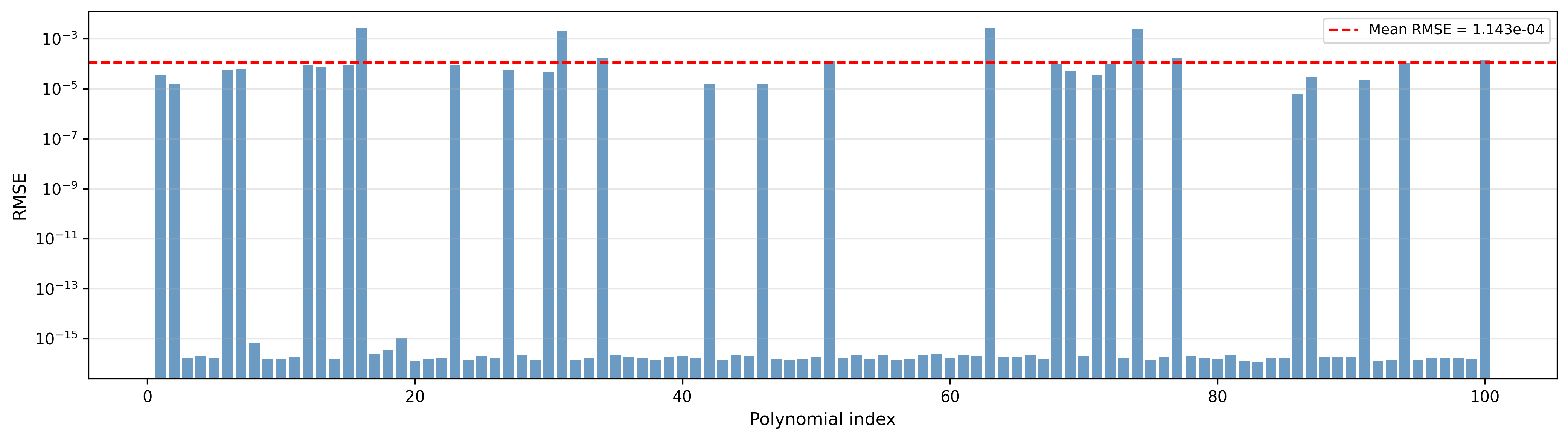}}\\
 \subfloat[Weight{=$[2,2,2,1,\dots,1]$},  Mean Error= $5.08\times 10^{-5}$]{\includegraphics[width=1\columnwidth]{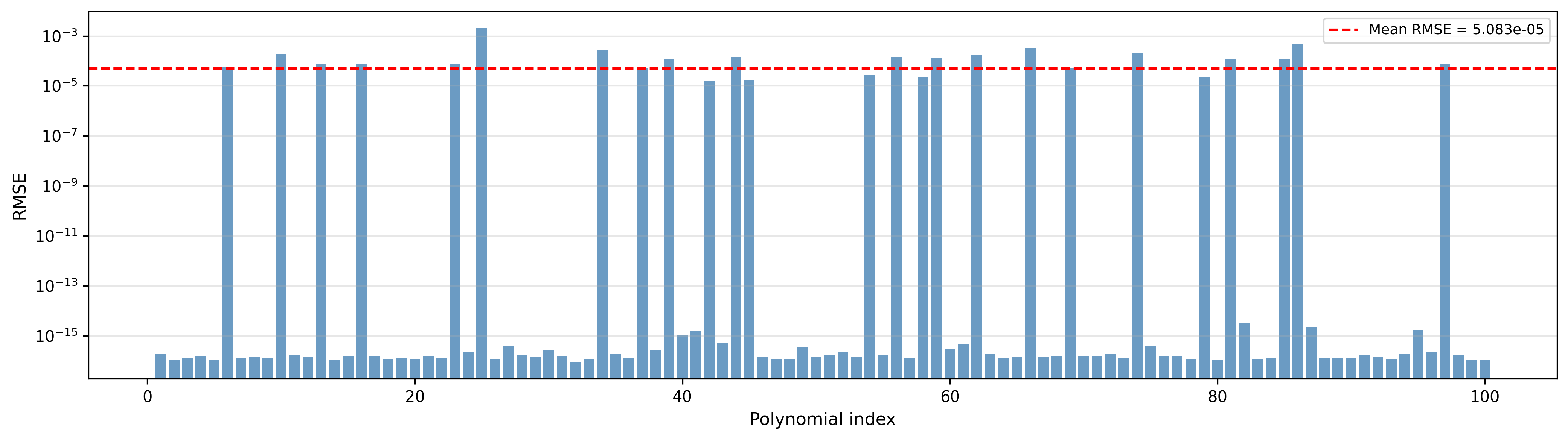}} \hfill
 \subfloat[Weight{=$[3,3,1,\dots,1]$},  Mean Error= $6.81\times 10^{-3}$]{\includegraphics[width=1\columnwidth]{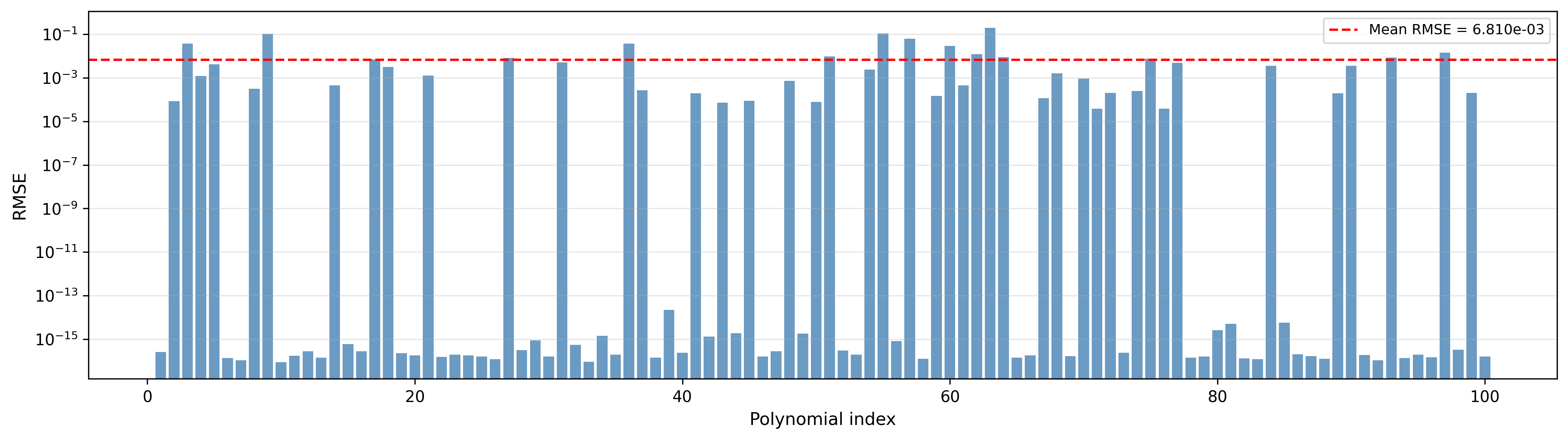}}\\
 \subfloat[Weight{=$[4,2,2,2]$},  Mean Error= $2.05\times 10^{-1}$]{\includegraphics[width=1\columnwidth]{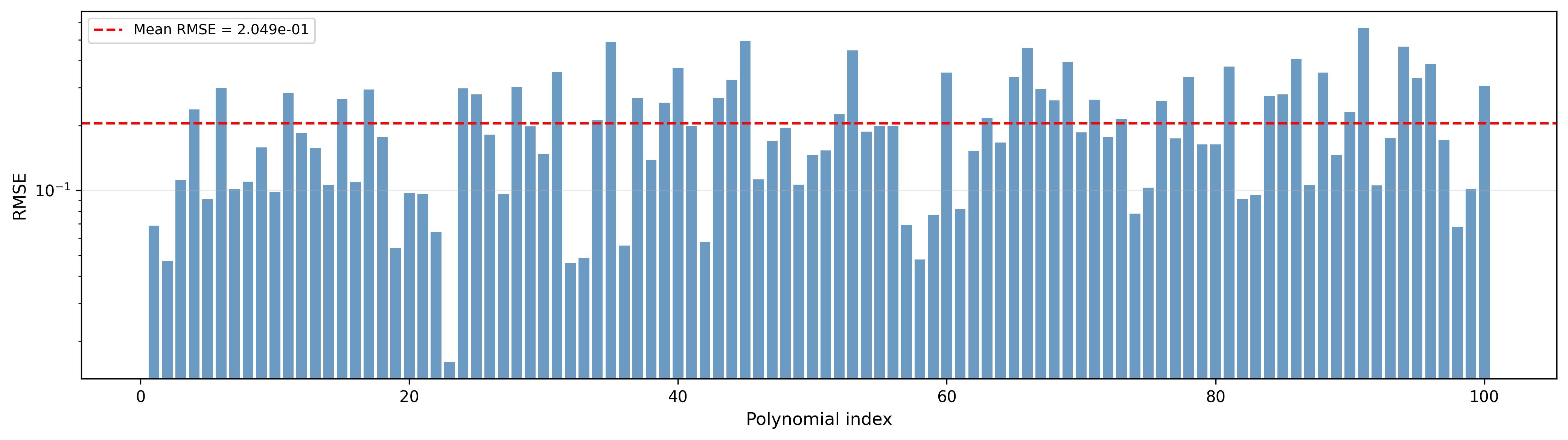}} \hfill
 \subfloat[Weight{=$[4,2,2,1,1]$},  Mean Error= $5.21\times 10^{-2}$]{\includegraphics[width=1\columnwidth]{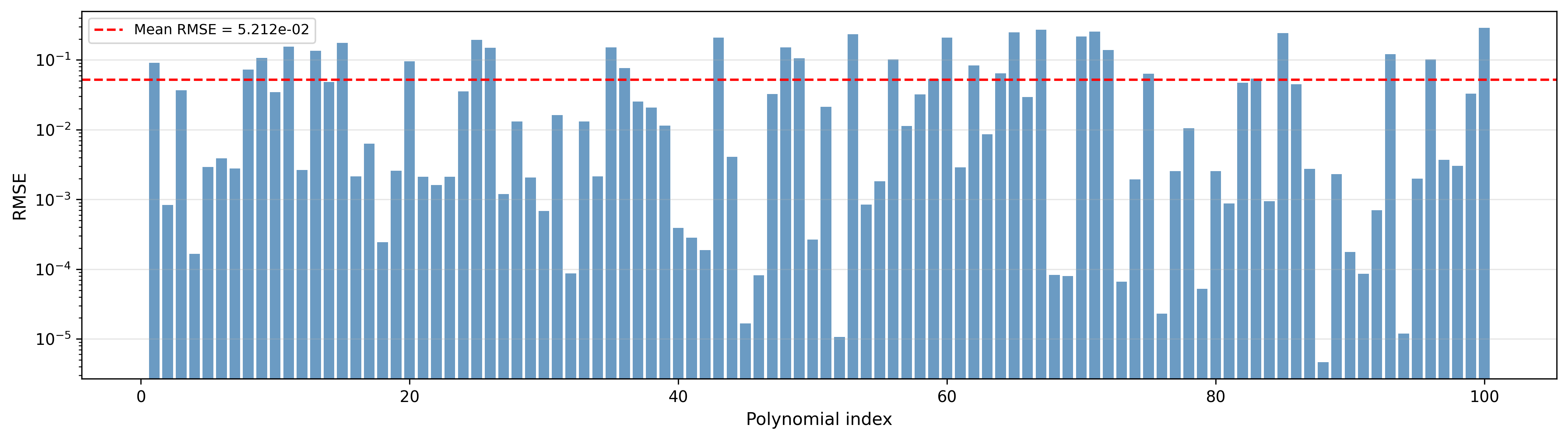}}
  \caption{Illustration of approximating  real degree $10$ polynomials using learning based WQSP under different partitions of $10$. A total of $100$ random polynomials were generated by sampling Chebyshev coefficients uniformly from $[-1/11,1/11]$. For each polynomial and partition, phase optimization was performed using the Levenberg--Marquardt method with $20$ random restarts, and the minimum approximation error was retained. Each subfigure is labeled by the corresponding partition, with the first subfigure representing standard QSP.}\label{tabledeg10}
\end{figure*}

Using the partition labels defined in Table~\ref{tab:wqsp_partitions}, we report the mean, median, and minimum approximation errors across all $100$ randomly generated degree-$10$ polynomials for each partition of $10$ in Figure \ref{fig:10poly}.

 We employ the least-squares solver from \texttt{scipy} with the Levenberg-Marquardt optimizer, minimizing the root mean square error (RMSE) between the target and approximated functions at the training points as the cost function. The Levenberg-Marquardt \cite{LM_method} algorithm is specifically suited to nonlinear least-squares and curve-fitting tasks, and does not generalize to arbitrary machine learning objectives. It switches between gradient descent and the Gauss-Newton method via an adaptive damping term, i.e. when far from a minimum, the damping term is large and the update resembles gradient descent, enabling large steps, whereas near a minimum, the damping term is small and the update reduces to smaller Gauss-Newton steps. The Gauss-Newton method approximates the curvature of the cost landscape to first order by evaluating the Jacobian of the residuals (one per training point) with respect to the trainable parameters, yielding a Jacobian of size $M \times N$, where $M$ is the number of parameters and $N$ is the number of training points.
\begin{figure*}[ht!]
    \centering
 \subfloat[ \label{mean} Mean Error ]{\includegraphics[width=\columnwidth]{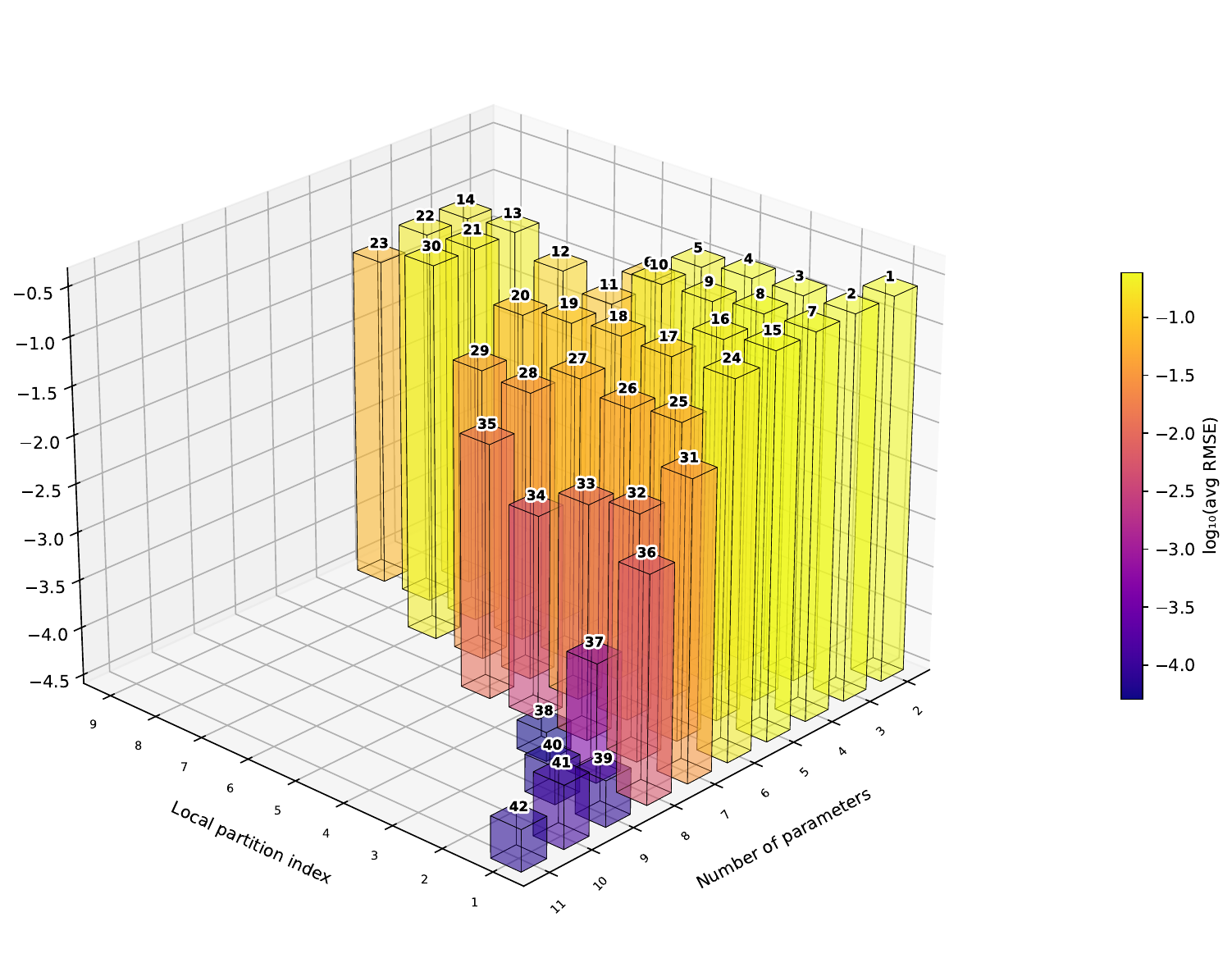}}
 \subfloat[ \label{median} Median Error]{\includegraphics[width=\columnwidth]{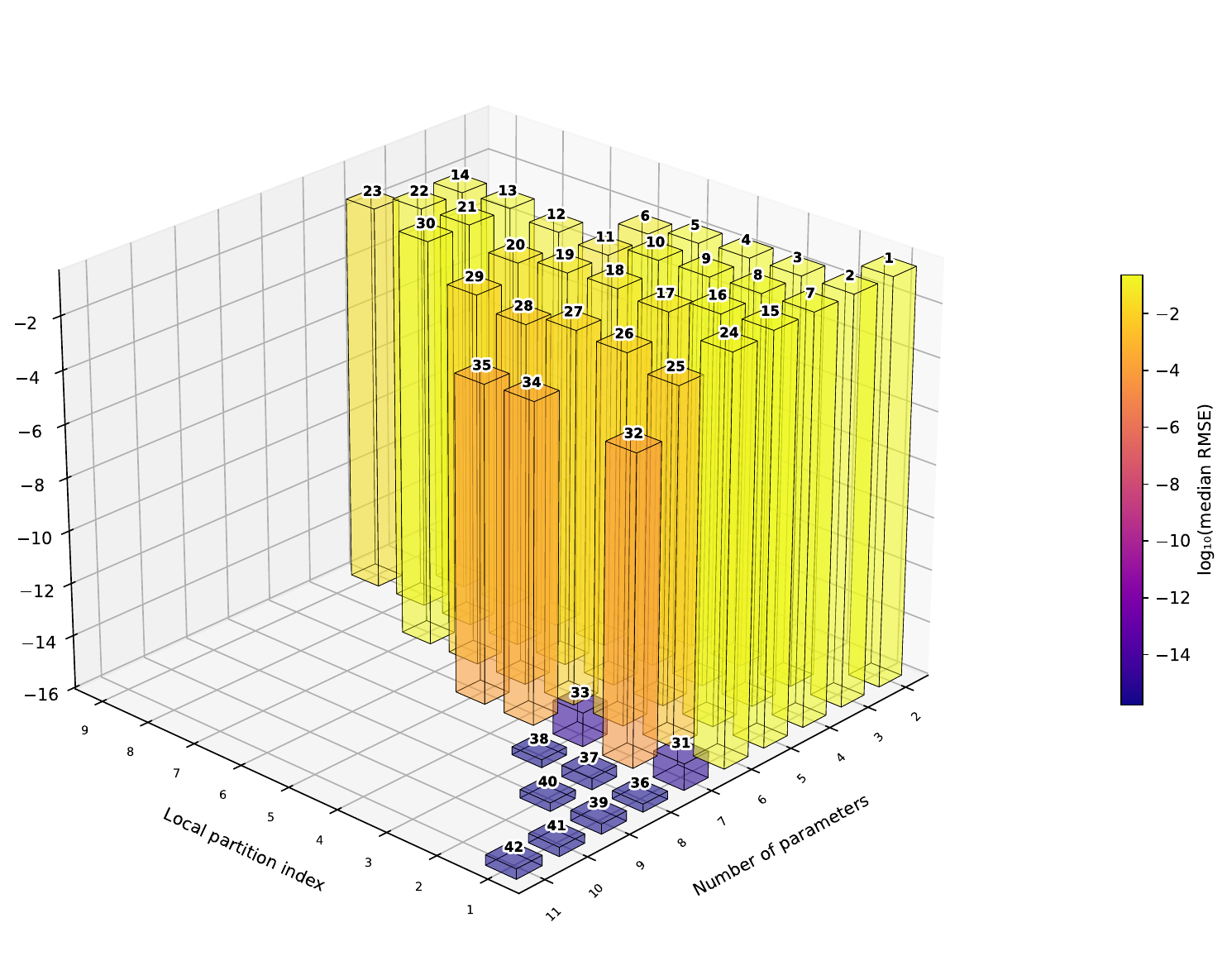}}\\
  \subfloat[ \label{min} Minimum Error ]{\includegraphics[width=\columnwidth]{images_R/wqsp_univariate_median_rmse_3dbar.pdf}}
   \caption{WQSP fitting results for real degree-$10$ target polynomials whose Chebyshev coefficients were independently sampled from the uniform distribution $U[-1/11,1/11]$. A total of $100$ random polynomial instances were generated. For each partition of $10$ (as labeled in Table~\ref{tab:wqsp_partitions}) illustrated by a bar in the figure, optimization was performed with $20$ random restarts and the solution with the lowest error was retained. The reported statistics correspond to the mean \ref{mean}, median\ref{median}, and minimum \ref{min} approximation errors across all $100$ samples. All of the errors have been presented in logarithmic scale.}
    \label{fig:10poly}
\end{figure*}
Although this approach does not scale to high-parameter regimes, in the low-parameter settings characteristic of WQSP circuits simulated in this work, it is known to exhibit substantially faster and more reliable convergence than standard gradient-descent or Nelder-Mead based optimizers used in conventional ML or QML pipelines, often reaching significantly lower cost function values. This training dynamics is particularly depicted in Figure \ref{fig:GDvsLM}. In our experiments, several circuits achieved an RMSE of the order of $10^{-16}$, a level of precision far beyond the reach of gradient-descent-based training.
\begin{figure*}[ht!]
    \centering
    \includegraphics[width=1.5\columnwidth]{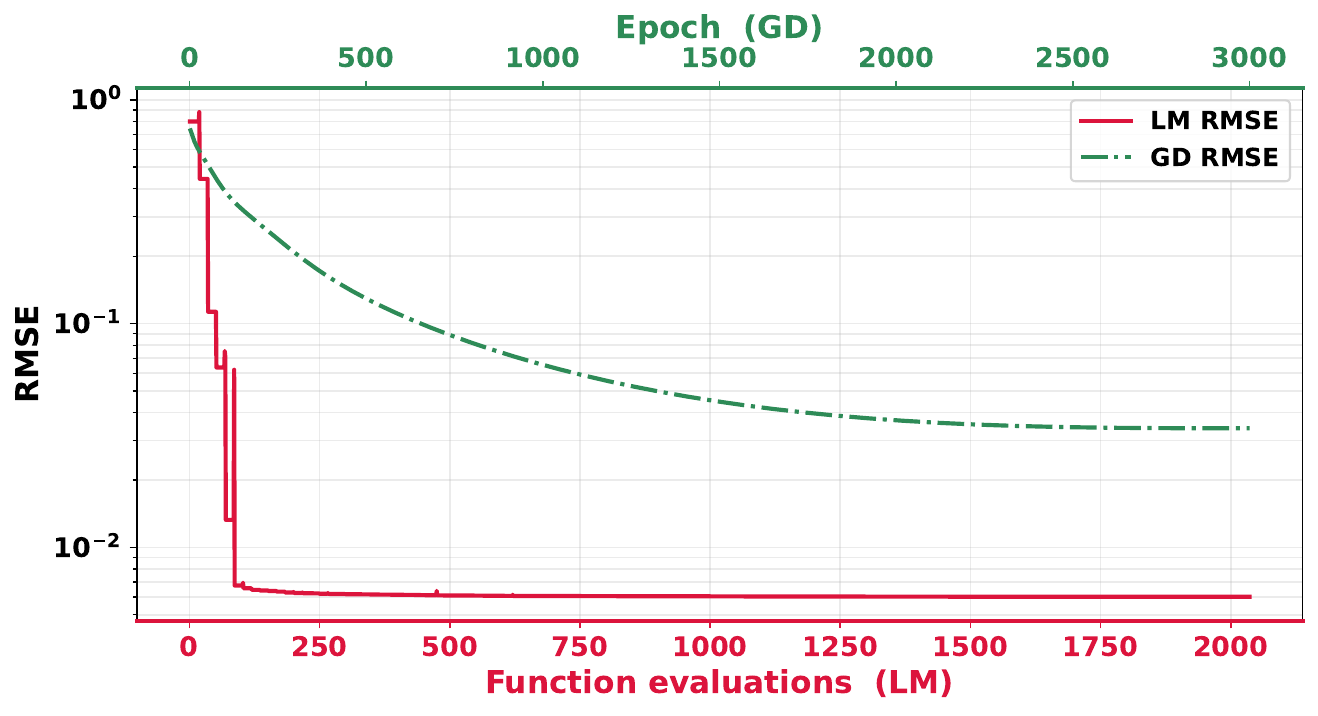}
\caption{Training dynamics of QSP circuit via gradient descent (GD) and Levenberg-Marquardt (LM) for univariate polynomial approximation using a single QSP circuit. LM leads to a much faster convergence to a much smaller value $O(10^{-16})$ as opposed to $O(10^{-4})$ in gradient descent.}
\label{fig:GDvsLM}
\end{figure*}
\begin{table*}[htbp]
\begin{center}
\caption{WQSP Partition Indexing}
\label{tab:wqsp_partitions}
\begin{tabular}{ccc|ccc}
\toprule
\textbf{Idx} & \textbf{Partition} & \textbf{\#Parameters} & \textbf{Idx} & \textbf{Partition} & \textbf{\#Parameters} \\
\midrule
1  & $[10]$              & 1 & 22 & $[3, 3, 3, 1]$       & 4 \\
2  & $[9, 1]$            & 2 & 23 & $[3, 3, 2, 2]$       & 4 \\
3  & $[8, 2]$            & 2 & 24 & $[6, 1, 1, 1, 1]$    & 5 \\
4  & $[7, 3]$            & 2 & 25 & $[5, 2, 1, 1, 1]$    & 5 \\
5  & $[6, 4]$            & 2 & 26 & $[4, 3, 1, 1, 1]$    & 5 \\
6  & $[5, 5]$            & 2 & 27 & $[4, 2, 2, 1, 1]$    & 5 \\
7  & $[8, 1, 1]$         & 3 & 28 & $[3, 3, 2, 1, 1]$    & 5 \\
8  & $[7, 2, 1]$         & 3 & 29 & $[3, 2, 2, 2, 1]$    & 5 \\
9  & $[6, 3, 1]$         & 3 & 30 & $[2, 2, 2, 2, 2]$    & 5 \\
10 & $[6, 2, 2]$         & 3 & 31 & $[5, 1, 1, 1, 1, 1]$ & 6 \\
11 & $[5, 4, 1]$         & 3 & 32 & $[4, 2, 1, 1, 1, 1]$ & 6 \\
12 & $[5, 3, 2]$         & 3 & 33 & $[3, 3, 1, 1, 1, 1]$ & 6 \\
13 & $[4, 4, 2]$         & 3 & 34 & $[3, 2, 2, 1, 1, 1]$ & 6 \\
14 & $[4, 3, 3]$         & 3 & 35 & $[2, 2, 2, 2, 1, 1]$ & 6 \\
15 & $[7, 1, 1, 1]$      & 4 & 36 & $[4, 1, 1, 1, 1, 1, 1]$ & 7 \\
16 & $[6, 2, 1, 1]$      & 4 & 37 & $[3, 2, 1, 1, 1, 1, 1]$ & 7 \\
17 & $[5, 3, 1, 1]$      & 4 & 38 & $[2, 2, 2, 1, 1, 1, 1]$ & 7 \\
18 & $[5, 2, 2, 1]$      & 4 & 39 & $[3, 1, 1, 1, 1, 1, 1, 1]$ & 8 \\
19 & $[4, 4, 1, 1]$      & 4 & 40 & $[2, 2, 1, 1, 1, 1, 1, 1]$ & 8 \\
20 & $[4, 3, 2, 1]$      & 4 & 41 & $[2, 1, 1, 1, 1, 1, 1, 1, 1]$ & 9 \\
21 & $[4, 2, 2, 2]$      & 4 & 42 & $[1, 1, 1, 1, 1, 1, 1, 1, 1, 1]$ & 10 \\
\bottomrule
\end{tabular}
\end{center}
\end{table*}\normalsize
This learning framework implemented through WQSP not only approximates prescribed polynomials but also naturally extends to arbitrary continuous functions. In Figure~\ref{tablesinx}, we provide an example of approximating the continuous function $\sin(2\pi x)$ using WQSP. The choice of this function is motivated by its fixed parity and we consider a $9$-degree polynomial representation of the function. It can be observed that standard QSP yields an exact fit along with other weight choices. However, for non-trivial weight partitions, WQSP achieves the approximation using fewer parameters at the expense of higher approximation error bounded by Theorem \ref{bounds}. Consequently, the WQSP neural network architecture provides a useful foundation for training expressive neural architectures such as Kolmogorov--Arnold Networks\cite{liu2024kan}.
\begin{figure*}[ht!]
\centering
 \subfloat[]{\includegraphics[width=0.31\textwidth]{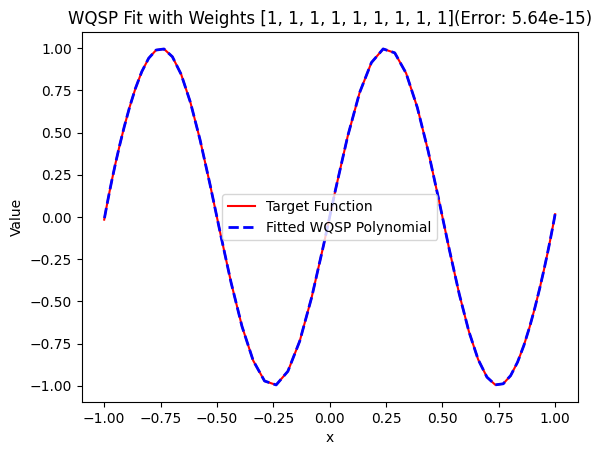}} \hfill
 \subfloat[]{\includegraphics[width=0.31\textwidth]{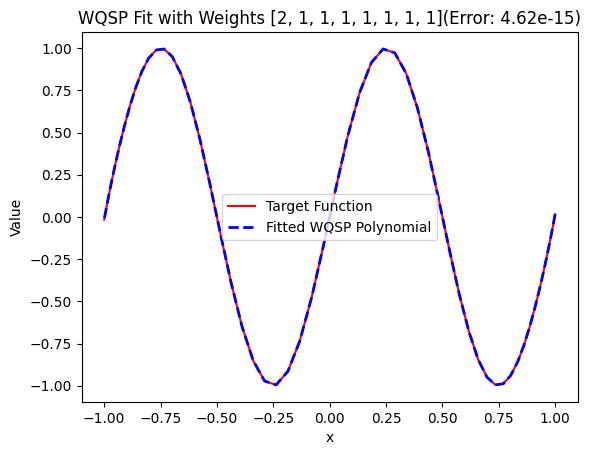}} \hfill
 \subfloat[]{\includegraphics[width=0.31\textwidth]{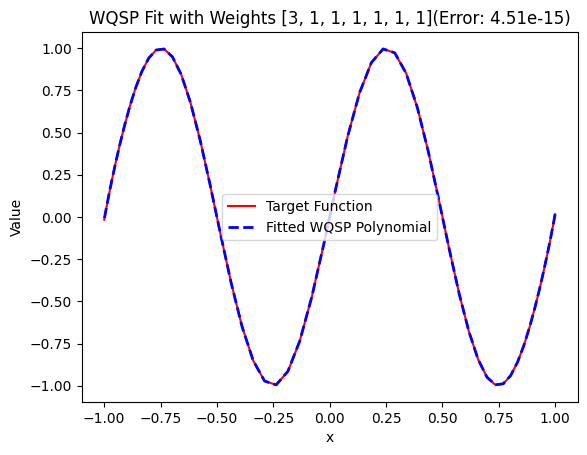}} \\
 \vspace{1ex}
 \subfloat[]{\includegraphics[width=0.31\textwidth]{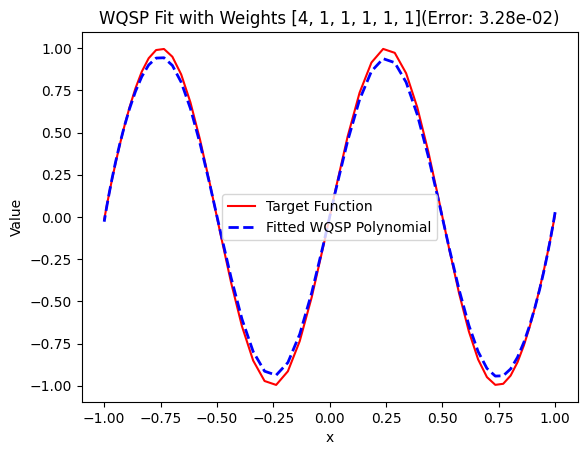}} \hfill
 \subfloat[]{\includegraphics[width=0.31\textwidth]{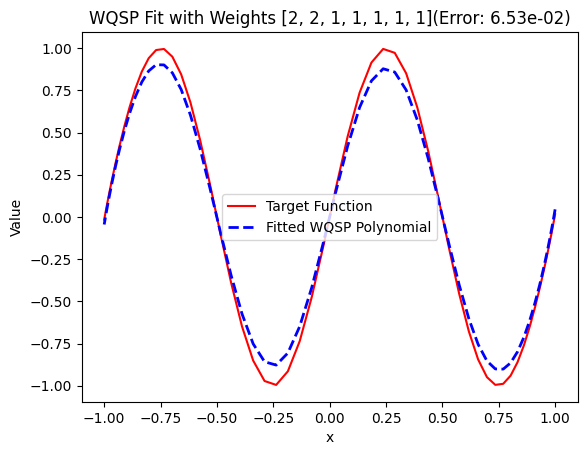}} \hfill
 \subfloat[]{\includegraphics[width=0.31\textwidth]{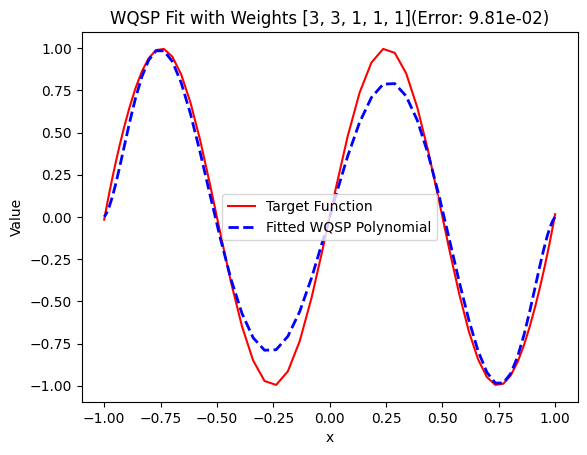}} \\
 \vspace{1ex}
 \subfloat[]{\includegraphics[width=0.31\textwidth]{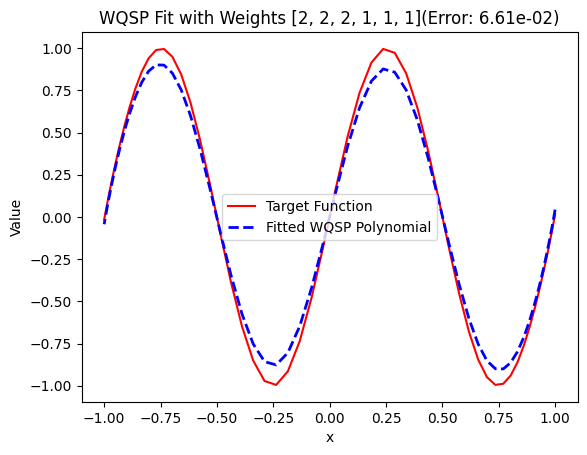}} \hfill
 \subfloat[]{\includegraphics[width=0.31\textwidth]{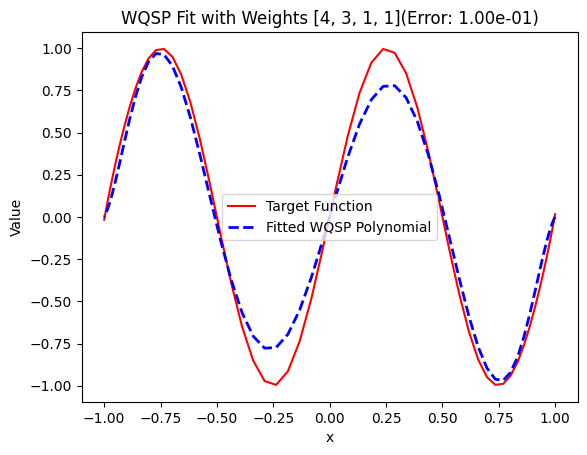}} \hfill
 \subfloat[]{\includegraphics[width=0.31\textwidth]{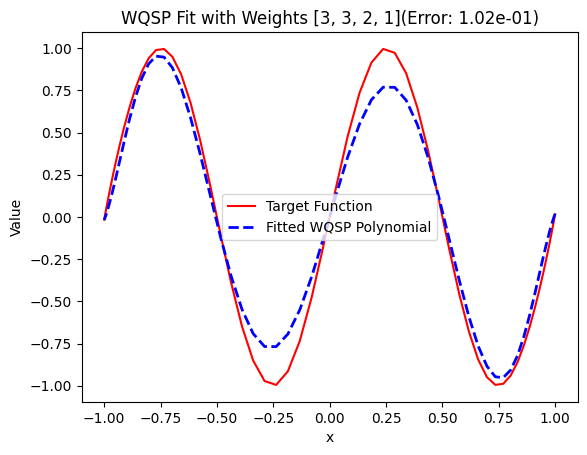}} 
\caption{Illustration of fitting $\sin(2\pi x)$ using WQSP. The function is approximated using a degree-$9$ polynomial, and several partitions of $9$ are selected as weight vectors. Optimization is performed using the Levenberg--Marquardt method with $50$ random restarts for each partition. The first subfigure corresponds to standard QSP and serves as the baseline error output. As observed in the subsequent subfigures, several non-trivial weight partitions achieve exact fitting as well, indicating that QSP contains substantial parameter redundancy. For more generic weight partitions, such as $(3,3,2,1)$, WQSP achieves a significant reduction in the number of parameters at the cost of increased approximation error.}\label{tablesinx}
\end{figure*} 
It is evident that the optimizer does not attain machine-precision accuracy in practice, despite the theoretical guarantees of exact realization under the proposed framework. Therefore, incorporating more specialized phase-recovery and optimization techniques from \cite{Chao2020Finding,lin2025mathematical,linlin2,linlininfinite,Linliniter,linlinenergy,Ying2022} to improve numerical precision, as well as extending the framework to very high-degree ancilla-free parity agnostic polynomial approximation, remains an important direction for future work. So far, we have discussed how to generate and approximate univariate polynomials and continuous functions in the WQSP regime. It is also of interest to understand how to extract these polynomials from their corresponding quantum circuits.

\section{Polynomial extraction from WQSP circuit}\label{polycirc}

It is obvious that for real polynomials, extracting them from WQSP is direct i.e. $|\bra{0}U_{\mathbf{\Phi}_\mathbf{w}}^{\mathrm{WQSP}}\ket{0}|^2$. However, for complex polynomials, the process is nontrivial because the output generates $|P(x)|^2$, $P(x)\in \mathbb{C}[x]$ and the real part and complex part is lost. For $P(x)\in \mathbb{C}[x]$ , one needs to extract the real and imaginary components of $P$. We adapt the Hadamard test approach \cite{NielsenChuang2010} where we consider the quantum circuit:

\begin{eqnarray}\label{realQSPextract}
    {\Qcircuit @C=1em @R=.7em {
 &\lstick{\ket{0}}&\gate{H}&\ctrl{1}&\gate{H}&\qw\\
 &\lstick{\ket{0}}&\qw&\gate{ U^{\mathrm{QSP}}_{\boldsymbol{\Phi}}}&\qw&\qw\\}}
\end{eqnarray}

This circuit in Equation \ref{realQSPextract}, the matrix 
$U^{\mathrm{QSP}}_{\boldsymbol{\Phi}_{\mathbf{w}}}$ acts upon the second qubit provided the first qubit is $\ket{1}$. This qubit also acts as an ancilla. \begin{enumerate}
    \item We begin with the starting state $\ket{00}$. 
    \item After applying $(H\otimes I)$, where $H=\frac{1}{\sqrt{2}}\bmatrix{1&1\\1&-1}$, we get $\ket{00}\rightarrow\frac{1}{\sqrt{2}}(\ket{00}+\ket{10})$.  
    \item Then controlled-$U$ is applied on the second qubit provided the first one is $\ket{1}$. Hence, we get the following
    \begin{eqnarray*}
        &&\frac{1}{\sqrt{2}}(\ket{00}+\ket{1}U\ket{1})\\&&=\frac{1}{\sqrt{2}}(\ket{00}+P(x)\ket{10}+i\sqrt{1-x^2}Q\ket{11})
    \end{eqnarray*}.
    \item We further apply $(H\otimes I)$ and we get the state
    \begin{eqnarray*}
        &&\frac{1}{2}\ket{0}\left((1+P(x))\ket{0}+i\sqrt{1-x^2}Q(x)\ket{1}\right)\\&&+\frac{1}{2}\ket{1}\left((1-P(x))\ket{0}-i\sqrt{1-x^2}Q(x)\ket{1}\right)
    \end{eqnarray*}
    \item Measuring $\ket{0}$ on the first qubit is \begin{eqnarray*}\hspace{-0.6cm}
        \text{Pr}(0)&=&\frac{1}{4}\|(1+P(x))\ket{0}+i\sqrt{1-x^2}Q(x)\ket{1}\|^2\\
        &=&\frac{1}{4}\left(1+2\text{Re}(P(x))+\underbrace{|P(x)|^2+(1-x^2)|Q(x)|^2}_{=1}\right)\\
        &=&\frac{1}{4}\left(2+2\text{Re}(P(x))\right)
    \end{eqnarray*}This implies
    \begin{eqnarray*}
       \text{Re}(P(x))={2\text{Pr}(0)-1}
    \end{eqnarray*}
\end{enumerate}

For Imaginary part of $P(x)$, one uses the quantum circuit

\begin{eqnarray}\label{imagQSPextract}
    {\Qcircuit @C=1em @R=.7em {
 &\lstick{\ket{0}}&\gate{H}&\gate{S}&\ctrl{1}&\gate{H}&\qw\\
 &\lstick{\ket{0}}&\qw&\qw&\gate{ U^{\mathrm{QSP}}_{\boldsymbol{\Phi}}}&\qw&\qw\\}}
\end{eqnarray} where $S=\bmatrix{1 & 0\\0&-i}$. In such a case $\text{Im}(P(x))=2\text{Pr}(0)-1$.
\section{Application of WQSP in Kolmogorov-Arnold Networks}\label{wqspKAN}
\begin{figure*}[t!]
    \centering
    \includegraphics[width=1\linewidth]{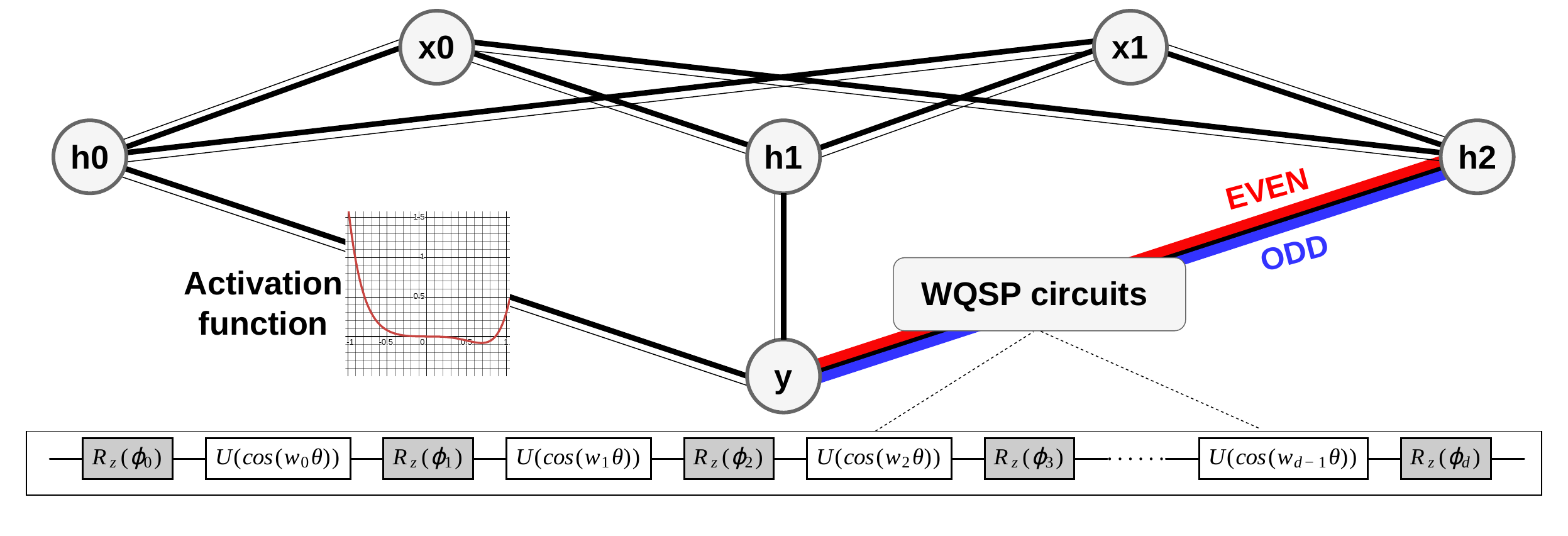}
 
    \caption{The proposed WQSP-based QKAN architecture. Each activation edge carries two WQSP circuits of depths $L$ and $L-1$ approximating the even and odd parity components of the activation function respectively, whose outputs are summed to produce an arbitrary univariate activation.}
    \label{fig:qkan_arch}
\end{figure*}
To showcase the applicability of WQSP in machine learning architectures, we apply this approach to realize activation functions in Kolmogorov–Arnold Networks (KANs) for multivariate function approximation. Kolmogorov–Arnold Networks (KANs) \cite{liu2024kan} have recently emerged as a compelling alternative architecture in this context. They are inspired by the Kolmogorov–Arnold Representation Theorem (KART) \cite{SCHMIDTHIEBER2021119}, also known as the Kolmogorov Superposition Theorem, originating from Hilbert’s 13th problem. KART states that any continuous function $f$ on a compact subset of $\mathbb{R}^n$  can be represented exactly as a finite superposition of continuous univariate functions. Motivated by this decomposition, KANs replace fixed activation functions with learnable univariate functions defined on edges, while aggregation is performed at nodes. 
We employ WQSP to parameterize learnable activation functions within Kolmogorov--Arnold Networks (KANs), resulting in a quantum-native architecture for multivariate function approximation. Each edge activation is realized as a WQSP-generated polynomial, integrating the compositional and interpretable structure of KANs with the expressive power of WQSP. It can be observed that WQSP-assisted KANs achieve compact, expressive, and numerically stable approximation of multivariate continuous functions. For Weighted Quantum Signal Processing, one can only generate polynomials with definite parity (either even or odd). Consequently, a general function $f(x)$ must be decomposed into its even and odd components, i.e.
\[
f(x) = \frac{f(x) + f(-x)}{2} + \frac{f(x) - f(-x)}{2}.
\].  The structure of WQSP-KANs is depicted in Figure \ref{fig:qkan_arch}.

Having validated the univariate approximation capability of QSP and WQSP, we turn to the multivariate setting and evaluate the full QKAN architecture on the target function \small\begin{eqnarray}\label{targetfn}\hspace{-1.5cm}
    y&&=f(x_0,x_1)\\\nonumber&&=\exp{\Biggl(-\frac{(x_0-0.5)^2+(x_1-0.5)^2}{0.2}\Biggr)}+0.3(x_0^2-x_1^2)
\end{eqnarray}\normalsize which combines a Gaussian peak centered at $(0.5,0.5)$ with a broad saddle region. We employ WQSP-assisted KANs with degree-10 polynomial activations, realized by WQSP circuits whose signal operator weights correspond to all 42 different integer partitions of $10$. Since the function is bivariate, we consider a $[2,3,1]$ architecture i.e. the architecture has two input nodes $x_0$ and $x_1$, three hidden nodes, and one output node $y$, giving $9$ activation edges in total. The number of trainable parameters per activation is kept equal across all architectures, namely, $2L$ where $L$ is the length of the weight vector for WQSP. This yields a total of $18L$ parameters per model. 

The QKAN parameters are randomly initialized and trained jointly using Levenberg-Marquardt method across all $9$ edges simultaneously with $20$ restarts and with the training residual set to the full network output. This joint optimization allows LM to account for the coupling between activations that arises from the summation and normalization operations at each hidden node. The number of training points is set to $18L+10$ in all cases.

In Figure~\ref{fig:threefigs}, we approximate the target function shown in Figure~\ref{fig:threefigs}(a), defined in Equation~\ref{targetfn}, using a WQSP-assisted KAN with degree-$10$ polynomial activations. The WQSP-KAN at $d=10$ visually captures the broad structure of the target function, including the Gaussian peak and saddle region, with a relatively smooth and consistent surface. The point-wise error surface reveals that the largest residuals occur near the boundary, while the saddle region is approximated with lower error. The WQSP-QKAN for weight $[3,3,2,1,1]$ i.e $L=5$  requires only $90$ total parameters compared to $198$ parameters for the base QSP-QKAN and produces a better result with error of magnitude $O(10^{-1.2})$, along with other weight partitions as illustrated in Figure \ref{fig:QKAN_d_10}. In these experiments, we naively assume that all activation polynomials have the same degree. Exploring optimal network architectures and heterogeneous degree assignments tailored to a given target function is left for future work. 

\begin{figure*}[ht!]
    \centering
 \subfloat[ \label{targ}Target function]{\includegraphics[width=\columnwidth]{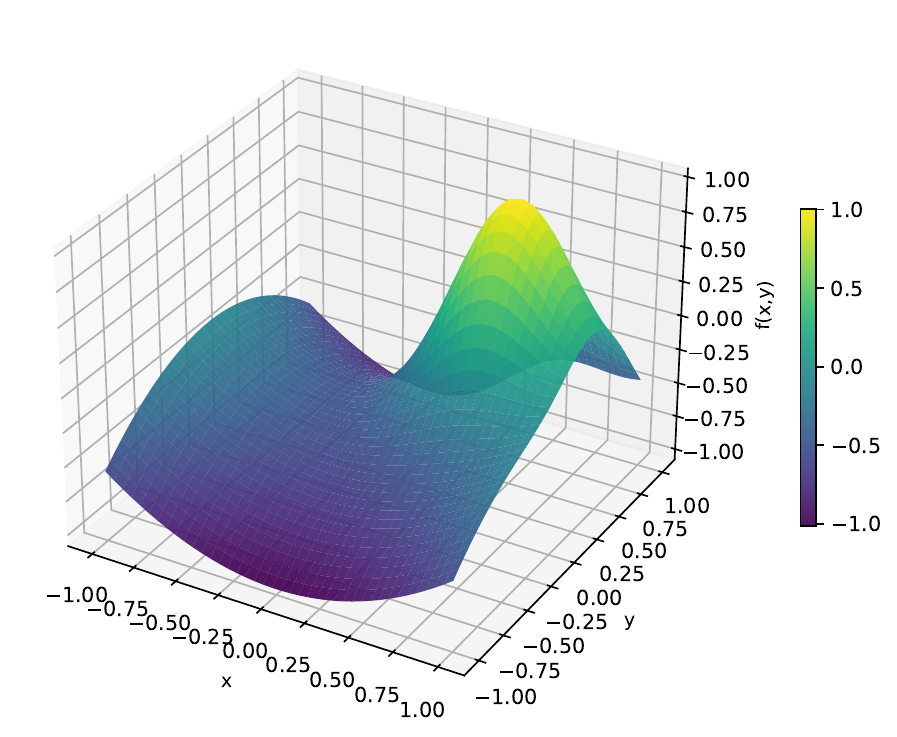}}
 \subfloat[ \label{appxwqsp} WQSP-KAN based approximation]{\includegraphics[width=\columnwidth]{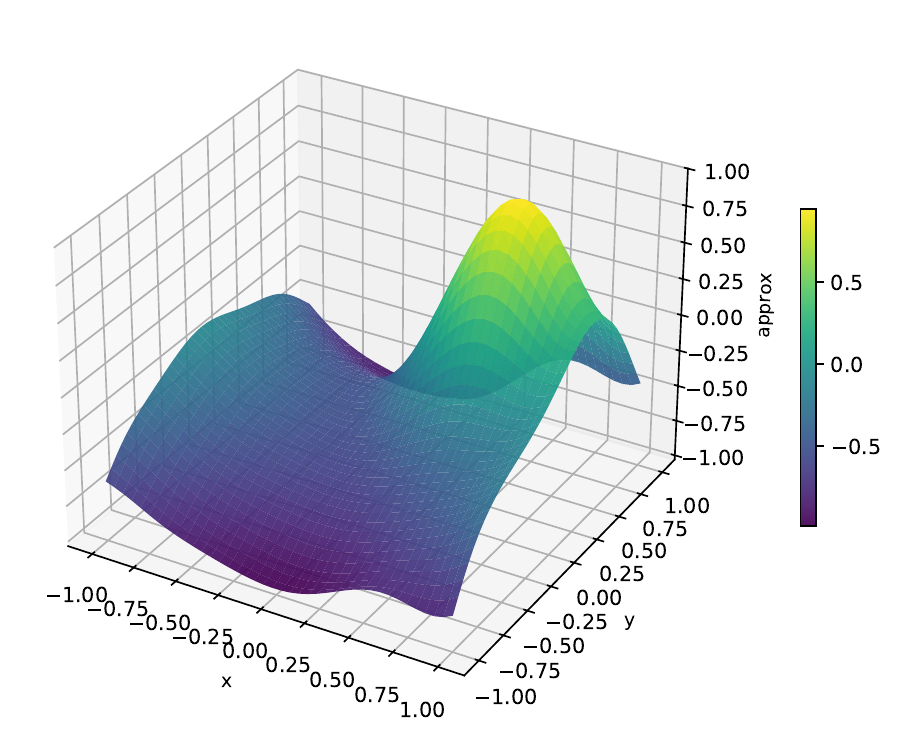}}\\
  \subfloat[ \label{errtarg} Point-wise Error ]{\includegraphics[width=\columnwidth]{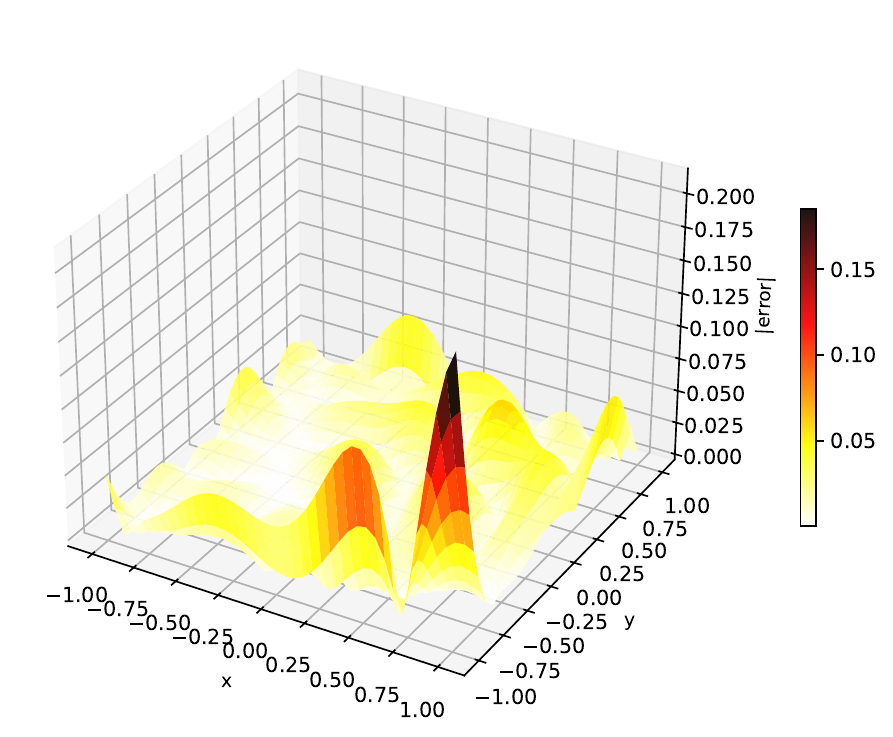}}
   \caption{Multivariate function approximation by a WQSP-QKAN with weight vector chosen to be $[3,3,2,1,1]$ i.e. every activation function is naively assumed to be a degree $10$ polynomial.}
    \label{fig:threefigs}
\end{figure*}

\begin{figure*}[t!]
    \centering
 \subfloat[ \label{2dplot} Approximation Error: WQSP-KAN]{\includegraphics[width=\columnwidth]{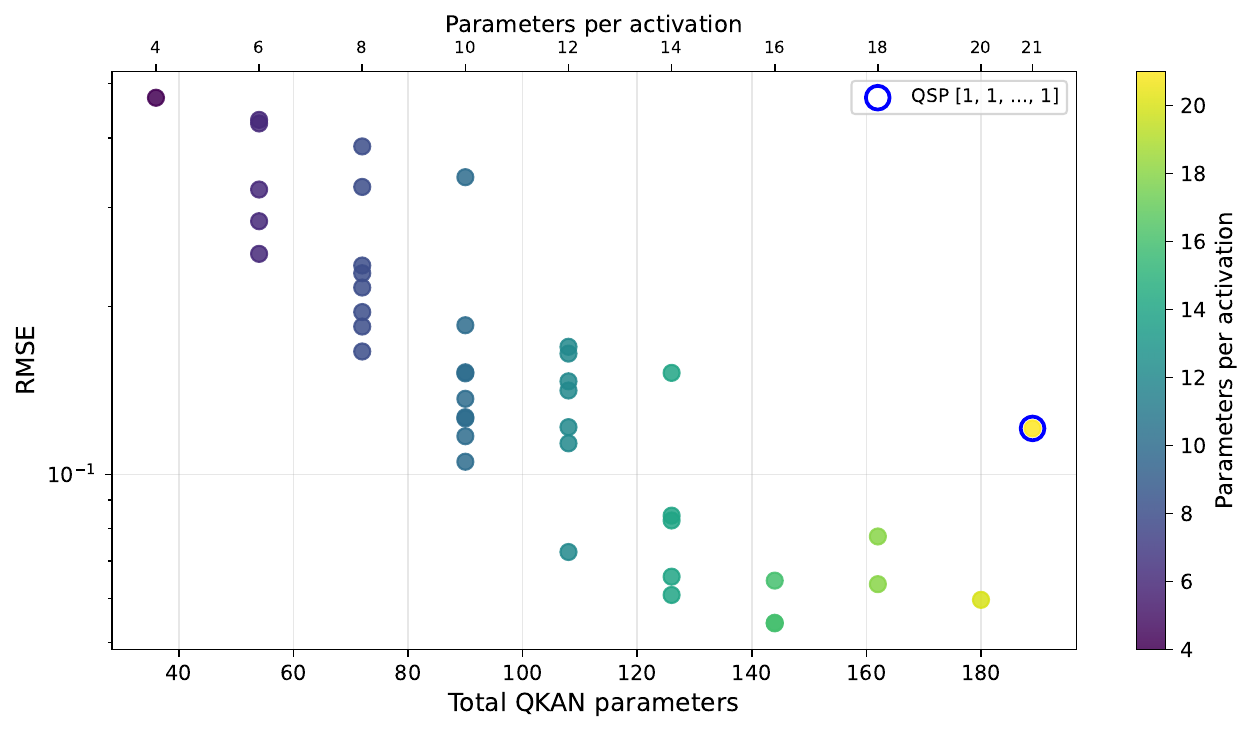}}
 \subfloat[ \label{3dplot} WQSP-KAN based approximation: partition based classification]{\includegraphics[width=\columnwidth]{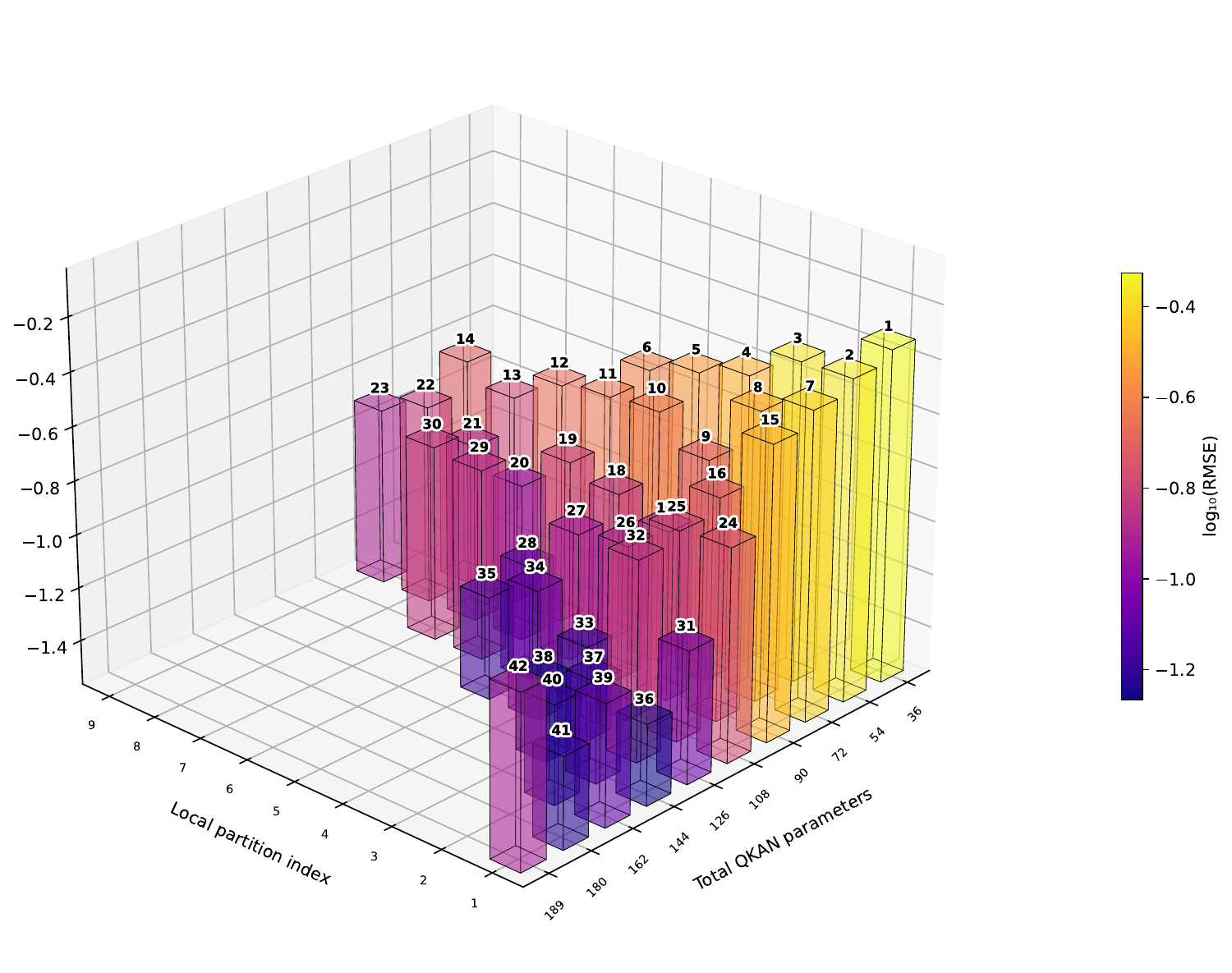}}
  \caption{(a) RMSE vs. parameter count for a $[2,3,1]$ QKAN with degree $10$ activations using WQSP across all integer partitions of $10$. Several partitions match or outperform QSP with fewer parameters. (b) 3D bar plot highlighting each individual partition of 10 by their corresponding index provided in Table \ref{tab:wqsp_partitions}. Clearly, it can be seen that, compared to QSP (index $42$), WQSP based KANs use far less parameters and provide better approximation errors for suitable weight choices }
    \label{fig:QKAN_d_10}
\end{figure*}

Thus, we observe that WQSP also proves to be an effective tool for integration into machine learning frameworks such as Kolmogorov--Arnold Networks. Its ability to represent expressive polynomial activations using significantly fewer trainable parameters enables compact models for multivariate function approximation. This highlights the broader applicability of WQSP beyond polynomial approximation, positioning it as a promising primitive for developing scalable quantum-inspired learning architectures with reduced parameter complexity.

\section{Conclusion} \label{sec:conclusion}

In this work, we addressed a major bottleneck of QSP namely, the growth in circuit depth and parameter count required for approximating high-degree polynomials. To overcome these limitations, we introduced Weighted Quantum Signal Processing, a structured extension of Quantum Signal Processing that incorporates fixed per-layer weights into the signal operators. When the weights are natural numbers greater than $1$, WQSP behaves as a pruned variant of QSP This modification enables the realization of higher-degree polynomial transformations while requiring significantly fewer trainable parameters.

We showed that polynomial approximation in WQSP, and consequently in QSP, reduces to solving two linear systems. Furthermore, we demonstrated that with an appropriate choice of weights, WQSP can achieve performance comparable to QSP while requiring fewer parameters and consequently shallower quantum circuits. Our analysis revealed that QSP contains a substantial degree of parameter redundancy, whereas WQSP serves as a gate-count- and circuit-depth-optimized realization of QSP. We attributed this reduction in gate count to the effective degrees of freedom induced by the nonzero coefficients of the target polynomial in the Chebyshev basis. We further observed that if the partition induced by the weight vector exceeds this degree-of-freedom threshold under suitable conditions, WQSP can recover solutions comparable to those obtained through QSP. However, for high-degree polynomials, identifying such optimal partitions becomes computationally expensive. Consequently, generic weight selections may introduce approximation errors, although they still provide substantial reductions in parameter count. Specifically, for a target polynomial of degree (d), standard QSP requires $\Theta(d)$ trainable phase parameters and a similar circuit depth. In contrast, WQSP with exponentially increasing weights $w_j=2^{j-1}$ (or alternatively $w_j=j$) achieves comparable approximation quality using only $k$ parameters where ($k=\Theta(\log_2 d)$) (or $\Theta(\sqrt{d})$), respectively). We also derived explicit upper bounds on the approximation error and showed that when the polynomial coefficients are sufficiently small, the approximation error becomes negligible even under weighted constructions. WQSP also provides a highly efficient, hardware-native framework for function generation. By utilizing $R_x$ rotations that map to a single physical pulse—rather than the multi-pulse Euler decompositions required by $\mathrm{SU}(2)$ rotation-based methods. Thus, our architecture significantly minimizes physical circuit depth and mitigates error accumulation in NISQ-era processors.

We also demonstrated that WQSP admits a natural interpretation as a quantum neural network architecture, where learning is performed through the optimization of phase angles. This learning-based formulation extends beyond polynomial fitting to the approximation of continuous functions, making WQSP a promising framework for compact function generators and learnable activation mechanisms in architectures such as Kolmogorov--Arnold Networks ~\cite{liu2024kan}, where expressive univariate function representations play a central role. To illustrate this broader applicability, we employed WQSP to parameterize the activation functions of a KAN and demonstrated its effectiveness in approximating multivariate continuous functions. While these initial results are encouraging, a comprehensive study of WQSP-assisted KANs for a wide range of multivariate functions remains an important direction for future research. In particular, investigating more complex benchmark functions and network architectures, performing architecture search and polynomial degree optimization, establishing tighter approximation and resource bounds, and exploring applications to scientific machine learning tasks, such as solving differential equations, are promising avenues for future work. Another key {open problem} is whether WQSP can be generalized to weighted oracle access within coherent QSVT settings. While our classical-input formulation incurs no physical gate overhead, extending this to unknown eigenvalues inside a black-box oracle faces known lower-bound constraints and remains a task for future research. We also note that all numerical experiments in this work are restricted to real-valued functions. To extend to complex functions $f(x)$, one must learn two real functions $u(x),v(x)$ such that $f(x)=u(x)+iv(x)$ and the procedure follows similarly.

Overall, WQSP provides a compact, structured, and trainable framework for polynomial and continuous function realization with reduced circuit resources while preserving expressive power comparable to state-of-the-art QSP constructions.

It is worth noting that the WQSP framework can, in principle, accommodate non-integer weights. However, such weights currently lack a clear physical interpretation within the standard quantum signal processing framework. Developing a comprehensive theoretical understanding of WQSP with non-integer weights, together with efficient algorithms for their optimization and implementation, remains an important direction for future research. Furthermore, for the approximation of very high-degree polynomials, it would be valuable to incorporate advanced phase-recovery and optimization techniques from~\cite{dong2024infinite,linlininfinite,Ying2022} to improve the numerical precision, computational efficiency, and stability of the proposed framework. Another promising direction is to extend WQSP to multivariate settings and integrate it with existing multivariate polynomial approximation frameworks~\cite{multivariablechuang,Rossi2025modularquantum}. More generally, WQSP may also provide an efficient mechanism for the simultaneous generation and approximation of multiple polynomials, following recent developments in this direction~\cite{laneve2023quantum}.
\section{Acknowledgment}
 R.S.S. acknowledges funding through Marie Skłodowska-Curie Actions-Ramon Llull AIRA (Grant No. 101126667). LP is supported by the National Research Foundation, Singapore through the National Quantum Office, hosted in A*STAR, under its Centre for Quantum Technologies Funding Initiative (S24Q2d0009). LP is also partially supported by A*STAR under its YIRG M25N8c0131.

\bibliography{ref.bib}
\end{document}